\documentclass[reqno]{amsart}
\usepackage{amsmath,amsfonts,amsgen,amstext,amsbsy,amsopn,amsthm, amssymb}

\newtheorem{lemma}{Lemma}[section]
\newtheorem{theorem}{Theorem}[section]
\newtheorem{proposition}{Proposition}[section]
\newtheorem{assumption}{Assumption}[section]

\theoremstyle{definition}
\newtheorem{definition}{Definition}[section]
\newtheorem{remark}{Remark}[section]

\newcommand\calC{{\mathcal{C}}}
\newcommand\Tr{{\rm Tr\,}} 
\newcommand\id{\mathbb{I}}
\newcommand\nn\nonumber

\numberwithin{equation}{section}

\newcommand{\R}{\mathbb{R}}
 
\newcommand{\N}{\mathbb{N}}
\newcommand{\Z}{\mathbb{Z}} 
\newcommand{\C}{\mathbb{C}}
 
\newcommand{\A}{\mathcal{A}} 
\newcommand{\F}{\mathcal{F}}
\newcommand{\E}{\mathcal{E}} 
\renewcommand{\H}{\mathcal{H}}
\newcommand{\W}{\mathcal{W}}

\newcommand{\Hr}{\mathbb{H}}

\newcommand{\gm}{\gamma} 
\newcommand{\al}{\alpha}

\newcommand{\half}{\mbox{$\frac 12$}}

\newcommand{\up}{\uparrow} 
\newcommand{\down}{\downarrow}

\newcommand\de{\mathcal D}
\newcommand\infspec{{\rm{inf\, spec\,}}} 
\newcommand\V{\mathcal{V}}
\newcommand\Trs{{\rm Tr}_0\,}

\newcommand\dig{\mathfrak{F}}

\newcommand{\be}{\begin{equation}}
\newcommand{\bea}{\begin{align}}
\newcommand{\eea}{\end{align}}

\newcommand{\dda}{{\rm d}}

\newcommand{\ee}{{\rm e}}

\DeclareMathOperator{\sgn}{sgn}
\DeclareMathOperator{\tr}{Tr}

\begin{document}

\title[The BCS functional of superconductivity]{The BCS functional of superconductivity and its mathematical properties}

\author{C. Hainzl}

\address{Mathematisches Institut, Universit\"at T\"ubingen\\ Auf der
  Morgenstelle 10, 72076 T\"ubingen, Germany\\ Email: christian.hainzl@uni-tuebingen.de}

\author{R. Seiringer}

\address{Institute of Science and Technology Austria (IST Austria), Am Campus 1, 3400 Klosterneuburg, Austria\\ Email: robert.seiringer@ist.ac.at}

\begin{abstract}
We review recent results  concerning the mathematical properties of the Bardeen--Cooper--Schrieffer (BCS) functional of superconductivity, which were obtained in a series of papers \cite{HHSS,FHNS,HS,HS1,HS2, HS3, FHSS, FHSS3}  partly in collaboration 
with R. Frank, E. Hamza, S. Naboko, and J.P. Solovej.  Our discussion includes, in particular, an investigation of the critical temperature for a general class of interaction potentials, as well as a study of its dependence on external fields. We shall explain how the Ginzburg--Landau model can be derived from the BCS theory in a suitable parameter regime. 
\end{abstract}

\date\today

\maketitle

\tableofcontents

\section{Introduction}

In this paper we shall review recent mathematical results on the Bardeen--Cooper--Schrieffer (BCS) theory  of superconductivity \cite{BCS}  that were obtained in a series of papers \cite{HHSS,FHNS,HS,HS1,HS2, FHSS, FHSS3}. Our primary goal is to give a  summary of the main results and the  methods which were developed for obtaining them.  We mainly concentrate on the 
mathematical aspects, and  refer to \cite{MR,de-Gennes,Leggett,leggett_quantum_liquids} for the physics background. We note, however, that while BCS theory was originally developed for the description of superconductors, where the basic constituents are charged particles, namely electrons, the theory has turned out to be applicable in a much larger context. For instance, it is used to describe ultracold gases of (fermionic) atoms, where the basic constituents are neutral particles, namely the atoms, and instead of superconductivity the relevant physical phenomenon of interest is superfluidity \cite{Leggett,NRS}. While the physics in this situation may be quite different, the mathematical description in the BCS theory is essentially the same.

We consider a system of fermionic particles with a two-body interaction, denoted by $V$. 
These particles could be electrons in a solid, or atoms in a cold atomic gas. In the latter case, the 
interactions are local, but in the former case it makes sense to introduce non-local interactions which arise as effective interactions through other degrees of freedom, like  phonons.  
Here, for definiteness, we stick to local interactions, described by an interaction potential of the form $V(x-y)$. From a mathematical point of view, this case is actually quite a bit harder than working with the effective interaction usually used in the theory of superconductivity, which is taken to be a rank-one projection, leading to a significantly simplified analysis. All our  results can   easily be extended to the case of non-local interactions as well.
We shall consider the influence of external fields, and shall refer to $A$ and $W$ as magnetic
and electric potential, respectively, even if the terminology is not quite correct in the case of neutral atoms. Effective forces of this kind can be obtained by other means for neutral particles as well, for instance via rotation of the trap in the case of cold gases.

In view of the large range of applicability of BCS theory, it makes sense to keep the interaction potential $V$ as general as possible. In the case of atomic gases the interaction can nowadays be tuned in the laboratory, for instance. 
Our main goal is to classify the interaction potentials $V$ for which the system shows a superconducting, or superfluid, phase. 
We will investigate the existence of a critical temperature below which this phase occurs, and study its dependence on  $V$ as well as on the external fields. 

From a mathematical point of view, we will give a precise definition what it means for the system to be in a superconducting phase. 
We shall also briefly sketch the physical significance of this definition.
Roughly speaking, 
superconductivity here means that the system displays macroscopic coherence  at thermal equilibrium, i.e., the particles or, more precisely, pairs of particles are correlated over macroscopic distances. Such a coherent behavior of the system is  guaranteed as soon as
the expectation value of pairs, also called the {\em Cooper-pair wavefunction}, does not vanish identically. This implication  of macroscopic coherence is built in in the BCS theory,  which will be
the starting point of our mathematical analysis. We shall introduce it in the next subsection. Before studying the corresponding  BCS functional in detail, we give 
in Section \ref{QM} a heuristic derivation of it starting from many-body quantum mechanics.

\subsection{The BCS energy functional}

We shall now introduce the BCS energy functional, which is the main focus of this paper. 
We consider a macroscopic sample of a fermionic system confined to a box $\calC \subset \R^d$ in $d$ spatial
dimensions.  Let $\mu\in\R$ denote the chemical
potential and $T\geq0$ the temperature of the sample. The fermions are assumed to 
interact through a local two-body potential $V$. Additionally, they are
subject to external electric and magnetic fields, with the electric potential denoted by  $W$ and the magnetic vector potential denoted by $A$.  

The BCS energy functional can naturally be viewed as a function of BCS states represented in terms of $2\times 2$ block matrices as 
\begin{equation}\label{def:Gammaa}
\Gamma=\left(\begin{matrix}\gamma&\alpha\\\overline{\alpha}&1-
\overline{\gamma}
\end{matrix}\right)
\end{equation}
satisfying the constraint $0 \leq \Gamma \leq 1$ as an operator on $L^2(\calC)\oplus
L^2(\calC)\cong L^2(\calC)\otimes \C^2$.
The bar denotes complex
conjugation, i.e., $\bar \alpha$ has the integral kernel $\overline
{\alpha(x,y)}$. In particular, $\Gamma$ is hermitian, implying that
$\gamma$ is hermitian and $\alpha$ is symmetric, i.e., $\gamma(x,y) =
\overline{\gamma(y,x)}$ and $\alpha(x,y)=\alpha(y,x)$. The quantity $\gamma$ describes the one-particle density matrix of the system, with its diagonal $\gamma(x,x)$ being the local particle density. While $\alpha$ is, by definition, a bounded operator on $L^2(\calC)$,  it is more naturally to think of its kernel  as a two-particle wave function. In fact it describes the 
expectation value of pairs and is  referred to as the {\em Cooper-pair} wavefunction. The $2\times 2$ matrix of operators $\Gamma$ in (\ref{def:Gammaa}) is also called the generalized one-body density matrix. Note that  there are no
spin variables in $\Gamma$. This is due to the assumed $SU(2)$-invariance which we discuss in the next section. This invariance implies that the full, spin-dependent Cooper-pair wavefunction is the product of $\alpha$ with an antisymmetric spin
singlet. This explains why $\alpha$ itself is symmetric so that we obtain
the antisymmetric fermionic character of the full, spin-dependent,
pair wave function.

The  BCS  energy functional takes the form 
\begin{align}\nonumber
  \mathcal{F}({\Gamma})&= {\rm Tr\,} \left[ \left(
      \left(-i \nabla + A(x) \right)^2 -\mu + W(x)\right) \gamma \right] - T\, S(\Gamma)
  \\  &\quad + \iint_{\calC \times \calC}
  V(x-y) |\alpha(x,y)|^2 \, {dx \, dy} \,, \label{def:rBCS}
\end{align}
where $\gamma$ and $\alpha$ are the entries in the first line of $\Gamma$, as in (\ref{def:Gammaa}), and the trace in the first term is over $L^2(\calC)$. The entropy $S$ of $\Gamma$ takes the usual form   $S(\Gamma)= - {\rm Tr\,} \Gamma \ln \Gamma$, where the trace is now over the doubled space $L^2(\calC)\oplus
L^2(\calC)\cong L^2(\calC)\otimes \C^2$ (see Appendix~\ref{ss:vne} for details).

Minimization of $\mathcal{F}$ over all admissible $\gamma$ and $\alpha$ leads to the BCS energy
\begin{equation}
F(T,\mu) = \inf_{\Gamma, 0\leq \Gamma\leq 1} \mathcal{F}(\Gamma)\ ,
\end{equation}
and the corresponding minimizer, which is the  
 BCS equilibrium state at temperature $T$ and chemical potential $\mu$. 
By definition, the system is in a {\em superconducting} phase 
if for a corresponding minimizer the pair wavefunction  $\alpha$ does not vanish identically. 
As already mentioned in the introduction and explained in more detail in Section~\ref{sec:LRO} below, this non-vanishing of $\alpha$ is related to the correlation of pairs over macroscopic distances. This concept is known 
 in the physics literature as {\em long range order} (LRO) and is connected to the occurrence of a  phase transition. 
Such a  coherence over a macroscopic distance is responsible for the vanishing of resistance in a metal 
or  of friction in an atomic gas. 

\subsection{Brief summary of mathematical results}

We approach the study of the BCS functional \eqref{def:rBCS} in two steps. In the first step, we determine the minimizer in the translation-invariant case and investigate the critical temperature in the absence of external fields. In the second step, we shall tackle the external field problem and introduce weak and slowly varying external fields. That is, there will be two relevant length scales, the microscopic scale which is determined by the interaction $V$, and the macroscopic scale determined by the external potentials $W$ and $A$. In a suitable parameter regime close to the critical temperature for the translation-invariant problem, we will be able to reduce the question of superconductivity to a study of the Ginzburg--Landau (GL) functional   
in the limit where the ratio between the microscopic and macroscopic scale goes to zero. The latter is much easier to analyze from a mathematical point of view. 

\subsubsection{Translation-invariant case} 

A study of the BCS functional \eqref{def:rBCS} in the trans\-lation-invariant case will be presented in Section \ref{ss:ti}. There we omit the external fields $A$ and $W$ and reduce the system to translation-invariant states, i.e., $\gamma$ and $\alpha$ have kernels of the form 
$\gamma= \gamma(x-y)$ and $\alpha=\alpha(x-y)$. The corresponding functional, which then has to be calculated per unit volume, is  significantly simplified in this case. 
In a joint work with E. Hamza and J.P. Solovej we showed in \cite{HHSS} the existence of a critical temperature, denoted by $T_c$, such that 
below $T_c$ the pair wavefunction $\alpha$ does not vanish identically and there exists a corresponding non-trivial solution  of  the BCS gap equation. 
For this to occur the potential $V$ needs to have at least some attractive part. However, the corresponding Schr\"odinger operator $-\nabla^2 + V(x)$ does not need to have a bound state. 
For $T \geq T_c$ there is no non-vanishing solution of the gap equation and hence $\alpha \equiv 0$. The critical temperature $T_c$ turns out to have a simple characterization in terms of the spectrum of a certain linear operator, which is similar to a Schr\"odinger operator but has a modified dispersion relation in the kinetic energy. This linear criterion allows to apply spectral methods to a study of the critical temperature. This was initiated in a joint work with R.L. Frank and S. Naboko in \cite{FHNS}, where necessary and sufficient conditions on $V$ for the strict positivity of $T_c$ were derived. 
Furthermore in \cite{FHNS,HS} a formula was derived for the 
critical temperature $T_c(\lambda V)$ in the weak coupling limit $\lambda \to 0$. The corresponding spectral theory is based on the fact
that the effective kinetic energy degenerates on a manifold of co-dimension one \cite{LSW,HS4}. 
Finally, a study of the behavior of $T_c$ in the low density limit was done \cite{HS2}; in this limit the value of the scattering length of $V$ becomes important, and a  formula well-known in the physics literature can be reproduced. All these results and the corresponding mathematical techniques will be thoroughly discussed in Section  \ref{birsch}.

\subsubsection{The case of weak and slowly varying fields}

Based on the results in the translation-invariant case we showed jointly with R.L. Frank and J.P. Solovej in \cite{FHSS} that close to the critical temperature $T_c$ 
and in the limit of slowly varying and weak external fields $V$ and $A$, the BCS energy $F(T,\mu)$ is determined by minimizing a suitable  Ginzburg--Landau functional. The corresponding Ginzburg--Landau order parameter $\psi$ describes the center-of-mass variations of the Cooper-pair wavefunction $\alpha$ due to the external fields, and varies on the macroscopic scale.  In contrast, the variation in the relative coordinate is described by the 
solution of the translation-invariant problem and lives on the microscopic scale. 
This result can be viewed as an a-posteriori justification of   Ginzburg--Landau theory, which was introduced originally solely on phenomenological grounds. 
It represents a rigorous version of arguments due to  Gorkov \cite{gorkov} and others \cite{E,de-Gennes}.

In \cite{FHSS3} we extend our result to the question of how the external fields alter the 
critical temperature. These results are presented in Section \ref{sgl} together with the main ideas of the proof. For some technical details, we will have to refer to the original papers.  In Section \ref{timin} we show, in addition, that under a certain positivity assumption (which can be proved for a class of interaction potentials $V$), in the absence of external fields, the translation-invariant minimizer of the BCS functional  is indeed  the minimizer, i.e.,   that the translation symmetry is not broken.

\medskip

Before presenting our results in detail in Sections~\ref{ss:ti} and~\ref{sgl}, we will give a brief overview of the relevant mathematical structure in Section~\ref{QM}, together with a heuristic derivation of the BCS functional~\eqref{def:rBCS} from quantum statistical mechanics.

\section{Mathematical background and \lq\lq derivation\rq\rq\ of the BCS functional} \label{QM}

\subsection{Quantum many-body systems} 

We consider a system of spin $\half$ fermions confined to a cubic box
$\Lambda \subset \R^d$, with periodic boundary conditions. The particles interact via a two-body potential $V$, which is an even and real-valued function with suitable regularity. For instance, for much of our work we shall assume that $V \in L^1(\R^d)$ so that it has a bounded Fourier transform, 
denoted by 
\begin{equation}
\hat V(p)= \frac 1 {(2\pi)^{d/2}}\int_{\R^d} V(x) e^{-ipx}dx\,.
\end{equation}
It is convenient to treat the system grand-canonically, i.e., not to fix the number 
of particles, hence we shall use the Fock space formalism in order to describe the system.  

\subsubsection{Fock space}

Let $\mathcal{H}$ be an abstract Hilbert space. Later we will use $\mathcal{H} = L^2(\Lambda) \otimes \C^2$ as appropriate for the description of spin $\frac 12$  particles in the box 
$\Lambda$.  
The Hilbert space describing a system
of $n$ identical fermions is given by the anti-symmetric tensor product
$\mathcal{H}^{(n)} = \bigwedge^n \mathcal{H}$. This space is spanned by simple vectors, called Slater determinants, of the form
\begin{equation}\label{def:slater}
  \psi_1 \wedge \dots \wedge \psi_n := \frac 1{\sqrt{n!}}\sum_{\sigma \in S_n} (-1)^\sigma
  \psi_{\sigma(1)} \otimes \dots \otimes \psi_{\sigma(n)}\, , \qquad \psi_i \in \mathcal{H},
\end{equation}
where $S_n$ denotes the set of all permutations of $n$ elements,
i.e. the symmetric group, and $(-1)^\sigma$ is the sign of a
permutation $\sigma$. 
The normalization is chosen in such a way that the vector \eqref{def:slater} has norm one if the $n$ functions $\psi_j$ form an orthonormal set.  
Let us mention that  $\psi_1 \wedge \dots \wedge \psi_n$ can also be written as a determinant
of an $n\times n$ matrix in the form
$$  \psi_1 \wedge \dots \wedge \psi_n(x_1,\dots,x_n) = \frac 1{\sqrt{n!}} \det \left[ \psi_i(x_j)) \right]_{1\leq i,j\leq n} \,.$$

The corresponding Fock space is given by the direct sum
\begin{equation*}
  \mathcal{F}_{\mathcal{H}} := \bigoplus_{n=0}^\infty \mathcal{H}^{(n)},
\end{equation*}
where $\mathcal{H}^{(0)} := \mathbb{C}\Omega$ and $\Omega$ is the
vacuum state with $\langle \Omega| \Omega \rangle := 1$.
It is a Hilbert space 
with the natural inner product induced by the vector space sum. In
particular, sectors of different particle numbers are orthogonal to
each other.
To any vector $\phi \in \mathcal{H}$ in the one-particle Hilbert
space, one can associate a creation operator $a^\dagger(\phi) :
\mathcal{F}_{\mathcal{H}} \to \mathcal{F}_{\mathcal{H}}$ and an
annihilation operator $a(\phi) :
\mathcal{F}_{\mathcal{H}} \to \mathcal{F}_{\mathcal{H}}$.
The creation operator $a^\dagger(\phi)$ is defined as usual as acting on a vector $\psi^{(n)} = \psi_1 \wedge \dots \wedge \psi_n \in \mathcal{H}^{(n)}$
 by
\begin{equation*}
  \big(a^\dagger(\phi) \psi \big)^{(n+1)} := (n+1)^{-1/2}
  \phi \wedge \psi^{(n)},
\end{equation*}
where $\phi \wedge \Omega = \phi$. This definition extends to general vectors in $\mathcal{F}$ simply by linearity.  
The annihilation operator $a(\psi)$ is defined to be the adjoint
operator of $a^\dagger(\psi)$ and acts as follows:
If $\psi^{(n)} = \psi_1 \wedge \dots \wedge \psi_n$, then
\begin{equation*}
  \big(a(\phi)\psi\big)^{(n-1)}
  = \sqrt{n} \sum_{i=1}^n (-1)^{i-1} \langle \phi| \psi_i\rangle
  \psi_1 \wedge \dots \wedge \widehat{\psi_i}\wedge \dots \wedge \psi_n,
\end{equation*}
where the hat indicates that $\psi_i$ is omitted in the wedge product.
Note that $\left.a(\phi)\right|_{\mathcal{H}^{(0)}} \equiv 0$.
By this construction, the creation and annihilation operators fulfill
the canonical anti-commutation relations (CAR)
\begin{equation} \label{anticomm}
\begin{split}
  \{a(\phi), a^\dagger(\psi)\} &= \langle \phi| \psi\rangle\\
  \{a(\phi), a(\psi)\} &= \{a^\dagger(\phi), a^\dagger(\psi)\} = 0, 
\end{split}
\end{equation}
where $\{A,B\} = AB+BA$ is the anti-commutator.

\subsubsection{Hamiltonian and Gibbs states}

Let $\{\varphi_j\}_{j\in\mathbb{N}}$ be an orthonormal basis of $\mathcal{H}$.
Then we can define the Hamiltonian of the system in terms of the
creation and annihilation
operators $$ a^\dagger_j := a^\dagger(\varphi_j) , \quad a_j:= a(\varphi_j) $$
via
\begin{equation}\label{redham}
\Hr = \sum_{j,k} T_{jk}\, a^\dagger_{j} a_{k} + \frac 12 \sum_{ijkl} V_{ijkl}  a^\dagger_i a^\dagger_j a_{l} a_{k},
\end{equation}
where $T_{jk}$ denotes  the matrix elements of the one-particle part of the energy, denoted by $h$, i.e., 
$$ T_{jk} = \langle \varphi_j | h \varphi_k \rangle\,.
$$
Typically, the $h$ to be considered here has the form 
 $$h =  \left(-i \nabla + A(x) \right)^2  + W(x) \,,$$
i.e., it  
includes the kinetic and potential energy of one  particle in external fields represented by  the scalar potential $W$ and  magnetic vector potential $A$. 
Moreover, $V_{ijkl} $ denotes the matrix elements of the two-particle interaction, i.e, 
$$ V_{ijkl} =\langle \varphi_i\otimes
  \varphi_j | V(x-y) \varphi_k\otimes \varphi_l \rangle.
$$
The Hamiltonian $\Hr$ in (\ref{redham}) conserves particle number, i.e., it leaves the subspaces $\mathcal{H}^{(n)} \subset \mathcal{F}_{\mathcal {H}}$ invariant. On the $n$-particle subspace, it acts as 
\begin{equation}
H_n = \sum_{j=1}^n h_j  + \sum_{1\leq i<j\leq n} V(x_i-x_j)
\end{equation}
where the subscript $j$ on $h_j$ indicates that $h$ acts on the $j$th tensor factor.

Let $$ \N = \sum_i a^\dagger_i a_i$$ be the number operator. The {\em Gibbs state} corresponding to  temperature $T$ and chemical potential $\mu$ is given by 
$$ \rho_\beta = Z^{-1} e^{ - \beta( \Hr - \mu \N)},$$
with $$Z = \tr e^{ - \beta( \Hr - \mu \N)}$$ and $\beta = 1/T$. 
Let us emphasize that $\rho_\beta$ also depends on the box $\Lambda = [0,L]^d$, 
so in fact we should write $\rho_\beta = \rho_{\beta,\Lambda}$. 
The trace is over the full Fock space and the {\em grand canonical potential} is given 
by $$ F = - \frac 1 \beta \ln Z = - \frac 1 \beta \ln \tr e^{ - \beta( \Hr - \mu \N)}.$$
Physically, it corresponds to the negative of the volume times the pressure of the system. 
According to the Gibbs variational principle, the Gibbs state $\rho_\beta$ minimizes the 
{\em pressure functional} 
\begin{equation}\label{def:pressure}
 \F(\rho) = \tr (\Hr - \mu \N)\rho - T S(\rho),
\end{equation}
defined on the set of all states, i.e, all $\rho$ satisfying 
$$ \rho \geq 0 \,\,\, {\rm and} \,\,\, \tr \rho =1.$$  
Here, $S(\rho) = - \tr \rho \ln \rho$ denotes the {\em entropy} of the state $\rho$.

In quantum statistical mechanics one is interested in expectation values of observables $A$, with $A$ being an operator on the Fock-space. 
The expectation of an observable in the state $\rho$ is given by
$$ \langle A \rangle_\rho = \tr A \rho.$$
Of particular interest are expectation values of observables in the Gibbs state 
$ \langle A \rangle_{\rho_\beta}$, in particular in the thermodynamic limit. 
It is precisely in this limit where phase transitions show up.
Concerning the occurrence of {\em superconductivity} we would like to know if in the thermodynamic limit $\Lambda \to \R^3$ 
there is a long range coherence in the expectation value of particle pairs. In particular, the main problem is to establish if, and for which interactions $V$,
one has 
\begin{equation}\label{holygrail}
 \lim_{|x|\to \infty} \lim_{\Lambda \to \R^3} \langle a^\dagger(f) a^\dagger(g) a(T_x g) a(T_x f) \rangle_{\rho_{\beta,\Lambda}} \neq 0\,,
 \end{equation}
 where $T_x : L^2(\R^d) \to L^2(\R^d)$ denotes translation by $x\in \R^d$. 
However, this question is out of reach of present day mathematics. For this reason we follow the way of Bardeen--Cooper--Schrieffer \cite{BCS} and 
restrict the pressure functional $\F$ in \eqref{def:pressure} to a particular set of states which is much easier to handle than the full set of states, but still allows to describe the relevant features of the pairing mechanism, i.e., includes states that satisfy (\ref{holygrail}).  
These states are called BCS states or quasi-free states (or also generalized Hartree--Fock states). 

\subsubsection{Quasi-free or BCS states} 

Bardeen, Cooper and Schrieffer \cite{BCS} had the fundamental idea that the phenomenon of superconductivity is due to the condensation of particle pairs.
At zero temperature, they had the following type of states in mind. Take a pair of particles and assume that it forms a spin-singlet. Such a pair can then be described
via an appropriate wave function as $\psi(x_1,x_2)=\phi(x_1 - x_2) \tfrac 1{\sqrt 2} (\up \down - \down\up)$, where $\phi$ is symmetric. 
Then one can  define an $N$-particle state that is just a condensation in this pair, i.e., 
$$\Psi(x_1,....,x_N) = \A \left[ \psi(x_1,x_2)  \psi(x_3,x_4) \cdots \psi(x_{N-1},x_N) \right] \,,$$
where $\A$ means that we have to anti-symmetrize the whole state. Such type of states can be generalized to {\em quasi-free} or {\em generalized Hartree--Fock} states, see \cite{BLS94}.
Their particular feature is that all $n$-point functions  can be expressed by two-point functions. This means, in particular, that the corresponding energy can be conveniently expressed by
such two point-functions. To this aim we will reduce our study to the following set of states.

\begin{definition}
A  state $\rho$ on $\F_{\mathcal{H}}$ 
is  a \emph{quasi-free state} if the following
``Wick Theorem'' holds:
\begin{equation}\label{def:wick}
\begin{split}
  \langle a^\#_1 a^\#_2\cdots a^\#_{2n}\rangle_\rho
  &= \sum_{\sigma \in S_{2n}'}(-1)^\sigma
  \langle a^\#_{\sigma(1)}a^\#_{\sigma(2)} \rangle_\rho \cdots
  \langle a^\#_{\sigma(2n-1)}a^\#_{\sigma(2n)}\rangle_\rho \\
  \langle a^\#_1 a^\#_2\cdots a^\#_{2n+1} \rangle_\rho
  &= 0,
\end{split}
\end{equation}
where each $a^\#_j$ can stand for $a^\dagger(f_j)$ or $a(f_j)$ for some $f_j\in\mathcal{H}$, 
and where $S_{2n}'$ is the subset of $S_{2n}$ containing the
permutations $\sigma$ which satisfy $\sigma(1) < \sigma(3) <\ldots <
\sigma(2n-1)$ and $\sigma(2j-1) < \sigma(2j)$ for all $1\leq j \leq
n$.
\end{definition}

Examples of quasi-free states include rank-one projections onto Slater-determi\-nants, and exponentials of quadratic Hamiltonians. A characterization of quasi-free states will be given in Lemma~\ref{lemma:qf} in the Appendix.

Quasi-free states have the property that the expectation value of any polynomial in the creation and annihilation operators can be expressed solely in terms of two quantities, the one-particle density matrix $\langle a^\dagger_i a_j \rangle_\rho$ and the pairing term $\langle a_i a_j \rangle_\rho$. Our Hamiltonian $\Hr$ in \eqref{redham} is in fact the sum of  quadratic and  quartic monomials. 
Note that for quartic monomials in the
creation and annihilation operators the Wick rule \eqref{def:wick} reduces to the condition
\begin{equation*}
    \langle a^\#_1 a^\#_2 a^\#_3 a^\#_4 \rangle_\rho = 
    \langle a^\#_{1}a^\#_{2}\rangle_\rho
    \langle a^\#_{3}a^\#_{4}\rangle_\rho
    - \langle a^\#_{1}a^\#_{3}\rangle_\rho \langle
    a^\#_{2}a^\#_{4} \rangle_\rho
    + \langle a^\#_{1}a^\#_{4}\rangle_\rho \langle a^\#_{2}a^\#_{3}\rangle_\rho.
\end{equation*}
Hence the expectation value $\langle \Hr \rangle_\rho$ becomes
\begin{align}\nonumber
\langle \Hr \rangle_\rho & = \sum_{j,k} T_{jk}\, \langle a^\dagger_{j} a_{k}\rangle_\rho \\ & \quad + \frac 12 \sum_{ijkl} V_{ijkl} \left(  \langle a^\dagger_i a_l \rangle_\rho \langle a^\dagger_{j} a_{k} \rangle_\rho  -  \langle a^\dagger_i a_k \rangle_\rho \langle a^\dagger_{j} a_{l} \rangle_\rho + \langle a^\dagger_i a^\dagger_j \rangle_\rho \langle a_{l} a_{k} \rangle_\rho 
\right)\,.  \label{rep:ham}
\end{align}
The terms in the last line are referred to as direct, exchange and pairing term, respectively.

To each quasi-free state $\rho$ one can associate a self-adjoint operator $\Gamma:
\mathcal{H}\oplus \mathcal{H} \rightarrow\mathcal{H}\oplus
\mathcal{H}$, 
called the \emph{generalized one-particle density matrix} of $\rho$, 
defined by\footnote{The complex conjugation of $\phi \in \mathcal{H}$ is denoted by $\overline{\phi}$. In an abstract Hilbert space, this simply means $\overline{\phi} = J \phi$ for an anti-linear involution $J$. In the concrete setting of $\mathcal{H} = L^2(\R^d)$ below it will always mean that $\overline{\phi}(x) = \overline{\phi(x)}$. }
\begin{equation}
  \label{eq:1-pdm}
  \langle (\phi_1, \phi_2)| \Gamma (\psi_1, \psi_2) \rangle
  = \langle [a^\dagger(\psi_1)+a(\overline{\psi_2})] 
  [a(\phi_1)+a^\dagger(\overline{\phi_2})]\rangle_\rho.
\end{equation}
In fact, a quasi-free state $\rho$ is uniquely determined by specifying $\Gamma$. Moreover, the pressure functional $\F$ in \eqref{def:pressure} for quasi-free states can be conveniently expressed solely  in terms of $\Gamma$.  
We can of course express $\Gamma$ in terms of the one-particle density matrix $\gamma$ and the pairing expectation $\alpha$. These are defined as 
the  
operators $\gamma, \alpha:
\mathcal{H} \rightarrow \mathcal{H}$ with expectation values 
\begin{equation}
\label{eq:gamma_alpha}
\begin{split}
  \langle \phi | \gamma \psi \rangle &= \langle a^\dagger(\psi)
  a(\phi) \rangle_\rho  \\
  \langle \phi | \alpha \overline{ \psi} \rangle &= \langle a(\psi)
  a(\phi)\rangle_\rho \,.
\end{split}
\end{equation}
Hence $\Gamma$ can be naturally written as an operator-valued  $2\times 2$ matrix of the form
\begin{equation}\label{defgamma}
\Gamma  = \left(\begin{matrix}\gamma &  \alpha\\
\alpha^\dagger & 1-  \bar \gamma
\end{matrix}\right),
\end{equation}
where $\bar \gamma$ is defined by $ \bar \gamma \phi : = \overline{\gamma \overline{\phi}  }.$ 
Note that the definition  \eqref{eq:gamma_alpha} implies  that 
\begin{equation}
\langle \phi| \alpha^\dagger \psi\rangle = \overline{ \langle a(\overline{\phi}) a(\psi)\rangle_\rho} = 
- \overline{ \langle a(\psi)a(\overline{\phi}) \rangle_\rho} = - \overline{ \langle \overline{\phi}| \alpha \overline{\psi}\rangle}= -  \langle \phi| \overline{\alpha} \psi\rangle,
\end{equation}
that is, $\alpha^\dagger = - \overline{\alpha}$.

The operator $\Gamma$ has an important property reflecting the fermionic structure, namely
\begin{equation}\label{prop:gamma}
  0\leq \Gamma \leq 1
\end{equation}
as an operator on $\mathcal{H}\oplus
\mathcal{H}$. 
This can easily be seen as an immediate consequence of the anti-commutation relations (\ref{anticomm}), which imply that
\begin{align}\nonumber
  \langle (\phi_1, \phi_2)| \Gamma (\phi_1, \phi_2) \rangle & = \langle [a(\phi_1)+a^\dagger(\overline{\phi_2})]^\dagger 
  [a(\phi_1)+a^\dagger(\overline{\phi_2})]\rangle_\rho \\ & =  \|\phi_1\|^2 + \|\phi_2\|^2 -  \langle  [a({\phi_1})+a^\dagger(\overline{\phi_2})]
  [a({\phi_1})+a^\dagger(\overline{\phi_2})]^\dagger\rangle_\rho \,.
\end{align}
While the expression on the first line is clearly non-negative, the one on the second line is less than $ \|\phi_1\|^2 + \|\phi_2\|^2$, which shows (\ref{prop:gamma}). 

Eq.~\eqref{rep:ham} shows that the energy expectation value $ \langle \Hr \rangle_\rho$ can be expressed in terms of the generalized density matrix $\Gamma$ associated to $\rho$. 
The same now holds true for the von-Neumann entropy of $\rho$, given by $S(\rho) = - \tr \rho \ln \rho$.  
We show in Lemma~\ref{lemma:entropy} in the Appendix that in terms of the generalized one-body density matrix $\Gamma$ 
 the entropy of a quasi-free state can be expressed as 
 \begin{equation*}
  S(\rho) = -\tr \big(\Gamma \ln \Gamma \big) = S(\Gamma),
\end{equation*}
where on the right  side the trace is over $\mathcal{H} \oplus \mathcal{H}$. 
It is interesting to note that the only functions for which the
relation $\tr_\F f(\rho) = \tr_{\mathcal H \oplus \mathcal{H}} f(\Gamma)$ holds for all quasi-free states are multiples of the function $f(x)= x \ln x$.

\subsubsection{Representation of the energy $\langle \Hr \rangle_\rho $ in terms of $\gamma$ and $\alpha$} 

Our next goal is to give an explicit and convenient representation of the  expectation value of the Hamiltonian (\ref{redham})  in  a quasi-free state, i.e., to express
 $\langle \Hr \rangle_\rho $ in terms of  $\gamma$ and $\alpha$. For this purpose,  we recall that we consider a system of spin $\frac 12$ particles in a box $\Lambda=[0,L]^d$, i.e., we have $\mathcal{H} = L^2(\Lambda) \otimes \C^2$. As an orthonormal basis in this space, we can use  the plane wave vectors $$\phi_j = L^{-d/2} e^{ikx} \sigma,$$ with 
$k\in (2\pi/L) \Z^d$, and $\sigma$ indicates the spin  $\sigma\in
\{\uparrow,\downarrow\}$. In this case, $j$ stands for the indices $\{k,\sigma\}$. 

With the definitions of $\gamma$ and $\alpha$ in \eqref{eq:gamma_alpha}, the expression (\ref{rep:ham}) for the 
 energy expectation value in a state $\rho$ becomes
\begin{align}\nonumber
\langle \Hr \rangle_\rho & = 
\sum_{j,k} T_{jk}\, \langle \phi_k | \gamma| \phi_j \rangle  \\ \nn & \quad + \frac 12  \sum_{i,j,k,l} V_{ijkl}   \Big( \langle \phi_k | \gamma| \phi_i \rangle \langle \phi_l | \gamma |\phi_j \rangle - 
 \langle  \phi_l | \gamma |\phi_i \rangle \langle \phi_k | \gamma |\phi_j \rangle \\ & \qquad \qquad  \qquad\qquad + \overline{ \langle \phi_j| \alpha |\bar \phi_i \rangle} \langle \phi_k | \alpha |\bar \phi_l \rangle 
\Big)\,. \label{2.16}
\end{align}
For the one-particle terms we can write
 \begin{multline}
\sum_{j,k} T_{jk}\, \langle \phi_k | \gamma | \phi_j \rangle  =
\sum_{j,k} \langle \phi_{j}|   \left[\left(-i \nabla + A(x) \right)^2 + W(x)\right] |\phi_{k} \rangle \, \langle \phi_{k} | \gamma |\phi_j \rangle \\ 
    =  \sum_{j} \langle \phi_{j}| \left[\left(-i \nabla + A(x) \right)^2 + W(x)\right]  \gamma |\phi_{j} \rangle  = \tr\left[\left(-i \nabla + A(x) \right)^2 + W(x)\right] \gamma,
    \end{multline}
 using that $\{\phi_i\}$ is an orthonormal basis. 
The interaction terms can be conveniently rewritten as follows. For the first (direct) term, we obviously have 
\begin{equation}
\langle \phi_k | \gamma | \phi_i \rangle \langle \phi_l | \gamma |\phi_j \rangle  = \langle \phi_k\otimes \phi_l | \gamma \otimes \gamma |  \phi_i\otimes \phi_j \rangle\,.
\end{equation}
For the second (exchange) term, we can write
\begin{equation}
\langle \phi_l | \gamma | \phi_i \rangle \langle \phi_k | \gamma |\phi_j \rangle  = \langle \phi_k\otimes \phi_l | {\rm Ex} \, \gamma \otimes \gamma|  \phi_i\otimes \phi_j \rangle \,,
\end{equation}
where he operator ${\rm Ex}$ exchanges the two particles. Finally, for the third (pairing) term, we have
\begin{equation}
\overline{ \langle \phi_j| \alpha |\bar \phi_i \rangle } \langle \phi_k | \alpha |\bar \phi_l \rangle = \langle \alpha |  \phi_k\otimes \phi_l  \rangle  \langle  \phi_i\otimes \phi_j | \alpha \rangle\,,
\end{equation}
where on the right  side we identify the operator $\alpha$ via its kernel with a two-particle wave function in $\H\otimes \H$. In particular, since 
 $\phi_i\otimes \phi_j$ is an orthonormal basis for 
$\H\otimes \H$, we can rewrite the terms in the second line in \eqref{2.16} as 
\begin{equation}
\frac 12   \big( \Tr_{\H\otimes\H} V  \left[ \gamma \otimes \gamma - {\rm Ex}  \gamma \otimes \gamma \right]  + \langle \alpha | V| \alpha \rangle \big)\,.
\end{equation}

For convenience, let us introduce the integral kernels $ \gamma_{\sigma,\tau}(x, y)$ and  $\alpha_{\sigma,\tau}(x,y)$, defined via 
\begin{equation}
\label{eq:gamma_alpha_kernel}
\begin{split}
  \langle a^\dagger(\phi_{l,\tau})
  a(\phi_{k,\sigma}) \rangle_\rho & = |\Lambda|^{-1} \iint_{\Lambda\times\Lambda} e^{- i ( kx - ly  )} \gamma_{\sigma,\tau}(x, y) dx\, dy  \\
   \langle a(\phi_{l,\tau})
  a(\phi_{k,\sigma}) \rangle_\rho & = |\Lambda|^{-1} \iint_{\Lambda\times\Lambda} e^{-i ( kx + ly  )} \alpha_{\sigma,\tau}(x, y) dx\, dy \,.
\end{split}
\end{equation}
Since, by definition, the interaction $V$ does not depend on spin, we end up with the following 
expression for the energy expectation value, called 
 the Bogoliubov--Hartree--Fock energy, 
 \begin{equation}
  \label{eq:F_BCS:0}
  \begin{split}
  \langle \Hr \rangle_\rho &=
  \tr( [-i \nabla + A(x)]^2 + W(x))\gamma
\\
      &\quad +\frac{1}{2} \sum_{\sigma, \tau
      \in\{\uparrow, \downarrow\}}\iint_{\Lambda\times\Lambda}
    \gamma_{\sigma,\sigma}(x,x) \gamma_{\tau,\tau} (y,y) V(x-y)
  \, dx\, dy\\
    &\quad -\frac{1}{2} \sum_{\sigma, \tau
      \in\{\uparrow, \downarrow\}}\iint_{\Lambda\times\Lambda}
    | \gamma_{\sigma,\tau}(x, y)|^2 V(x-y) \,dx\, dy\\
    &\quad +\frac{1}{2} \sum_{\sigma, \tau
      \in\{\uparrow, \downarrow\}}\iint_{\Lambda\times\Lambda}
    |\alpha_{\sigma,\tau}(x,y)|^2 V(x-y) \,dx\,dy.
    \end{split}
\end{equation}
There are three  terms resulting from the interparticle interaction. The first is usually referred to as {\em direct} energy. It is simply  the classical density-density interaction  energy of a particle distribution with density  given by the 
diagonal of $\gamma$, summed over the spin variables.  The second term is called {\em exchange} energy and is characteristic for fermionic systems. 
These two terms are also present in the usual Hartree--Fock functional.   Of particular interest for us is the third term, the {\em pairing} interaction term. 
In order to simplify our mathematical study we will actually omit the first two terms and only retain the pairing term. This leads to the BCS energy functional.

\subsection{BCS energy functional} 

In order to obtain the BCS functional in the form \eqref{def:rBCS} we have to do two more simplifications.
First we restrict to $SU(2)$-invariant states and, secondly, we neglect the direct and exchange energies.
The restriction to $SU(2)$-invariant states  effectively gets rid of the spin degrees of freedom and hence simplifies the analysis. It corresponds to the assumption that the Cooper-pair wavefunction is in a spin singlet state, i.e., the antisymmetry is in the spin variables.

\subsubsection{Reduction to $SU(2)$-invariant states}

Let $S \in SU(2)$ denote an arbitrary rotation in spin-space. 
As will be explained in the Appendix, a quasi-free state $\rho$ is $SU(2)$-invariant 
if the corresponding one-particle density matrix $\Gamma$ satisfies 
$$  \mathcal{S}^\dagger \Gamma \mathcal{S} = \Gamma,$$
with 
\begin{equation}
 \mathcal{S} = \left(\begin{matrix}S&0\\ 0 &\bar S \end{matrix}\right)\,.
 \end{equation} 
In terms of $\gamma$ and $\alpha$ this $SU(2)$ invariance means $$ S^\dagger \gamma S = \gamma, \quad S^\dagger \alpha \bar S = \alpha. $$
It is an elementary fact from linear algebra that 
if a matrix $M \in \mathbb{C}^{2\times 2}$ has the property $S^\dagger M S =
M$ for all $S\in SU(2)$, then $M$ must be a multiple of the identity matrix. 
Similarly, the property $S^\dagger M \bar{S} = M$ for all $S \in SU(2)$ implies that $M$ is a scalar multiple  of the second Pauli matrix  
\begin{equation}
\sigma^{(2)} =  \left( \begin{matrix}0&-i \\ i &0 \end{matrix}\right).
 \end{equation} 
Therefore the requirement of $SU(2)$-invariance implies that the kernels $\gamma_{\nu,\tau}(x,y)$ and $\alpha_{\nu,\tau}(x,y)$
have to be of the form
\begin{equation}
\begin{split}
\gamma_{\nu,\tau}(x,y) & = \gamma(x,y) \delta_{\nu.\tau}\\
\alpha_{\nu,\tau}(x,y) & = \alpha(x,y) \sigma^{(2)}_{\nu,\tau}\,.
\end{split}
\end{equation}
The kernels $\gamma(x,y)$ and $\alpha(x,y)$ now define operators on the space $L^2(\Lambda)$,
with $\gamma$ being self adjoint and $\alpha$  symmetric, $\alpha(x,y) = \alpha(y,x)$. 

Hence, with the assumption of $SU(2)$-invariance of the states, the energy functional can naturally be expressed in terms of the spin-independent quantities
\begin{equation}
\Gamma  = \left(\begin{matrix}\gamma &  \alpha\\
\overline \alpha & 1-  \bar \gamma
\end{matrix}\right),
\end{equation}
with $\gamma$, $\alpha$ being operators on $L^2(\Lambda)$ instead of operators
on $L^2(\Lambda) \otimes \C^2$. 
In this way we get rid of the spin-dependence, at the expense of adding appropriate factors of 2
in $\langle \Hr\rangle_\rho$ and $S(\rho)$. More precisely, 
\begin{equation} \label{2.28}
  \begin{split}
  \langle \Hr \rangle_\rho - T S(\rho)  &=
   2 \tr( [-i \nabla + A(x)]^2 + W(x))\gamma - 2 T S(\Gamma)
\\ & \quad +2 \iint_{\Lambda\times\Lambda}
    \gamma(x,x) \gamma(y,y) V(x-y)
  \,dx\,dy \\ 
    & \quad - \iint_{\Lambda\times\Lambda}
    | \gamma(x, y)|^2 V(x-y) \,dx\,dy
 \\ & \quad +\iint_{\Lambda\times\Lambda}
    |\alpha(x,y)|^2 V(x-y) \,dx\,dy \,.
  \end{split}
\end{equation}
where the trace in the first term is now over $L^2(\Lambda)$. 
Note that all terms got multiplied by a factor of 2, except for the direct interaction term, which got multiplied by a factor of four.

\subsubsection{Omitting direct and exchange energies}

As a last simplification we omit the direct and exchange terms
\begin{equation}\label{omit} 2 \iint_{\Lambda\times\Lambda}
    \gamma(x,x) \gamma(y,y) V(x-y)
  \,dx\,dy - \iint_{\Lambda\times\Lambda}
    | \gamma(x, y)|^2 V(x-y)\, dx\, dy
\end{equation}
from the energy functional \eqref{2.28}. 
These terms complicate the analysis significantly but are not expected to affect the general physical picture much.  In fact, in the case of cold atomic gases Leggett argued  in \cite{Leggett} that since the range of the interaction $V$ in typical experiments is very small compared to the 
other characteristic lengths of the system (like the scattering length of $V$ and the mean particle distance)  the terms (\ref{omit}) only lead to a change in the effective chemical potential and hence can be absorbed into the choice of $\mu$. For translation-invariant systems, this was actually made precise in \cite{BHS} where it was indeed  proved that \eqref{omit} can be omitted for interactions with very short range.

As a result,  we finally arrive at the BCS energy functional
 \begin{equation}
\F(\Gamma)  =
   \tr( [-i \nabla + A(x)]^2 + W(x))\gamma -  T S(\Gamma)
    +\iint_{\Lambda\times\Lambda}
    |\alpha(x,y)|^2 V(x-y) \, dx\, dy\,,
    \end{equation}
where we divided by $2$ and replaced $V$ by $2V$ for convenience. 
This is exactly the 
expression for the BCS functional given in \eqref{def:rBCS} above.

\subsection{Collective behavior and long range order} \label{sec:LRO}

The BCS theory is designed in such a way that whenever $\alpha$ is not identically zero, then the system displays a macroscopically coherent behavior, resulting in the loss of  resistance in the case of superconductors, or the loss of frictions in the case of superfluids. In terms of correlation functions this behavior is characterized as a certain {\em long range order} (LRO), i.e., the absence of decay of suitable expectation values as function of the spatial distance.

The relevant correlation function here is a four-point function, where a pair of particles is created close to one another in some region of space, and annihilated in some other region far away. If we introduce operator-valued distributions $a(x,\sigma)$ and $a^\dagger(x,\sigma)$ in the usual way via
\begin{equation}
a^\dagger (f) = \sum_{\sigma\in\{\uparrow,\downarrow\}} \int_{\R^d} f(x,\sigma) a^\dagger(x,\sigma) dx
\end{equation}
for $f\in L^2(\R^d)\otimes \C^2$, the relevant quantity to look at is the behavior of 
\begin{equation}\label{asin}
 \lim_{\Lambda\to \R^d} \langle  a^\dagger(x,\up)a^\dagger(x,\down) a(y,\down)a(y,\up) \rangle_\rho
 \end{equation}
as the distance $|x-y|$ gets large.
For $SU(2)$-invariant quasi-free states, this expectation value takes the form
\begin{equation}
\langle  a^\dagger(x,\up)a^\dagger(x,\down) a(y,\down)a(y,\up) \rangle_\rho = \gamma(x,y)^2 + \overline {\alpha(x,x)} \alpha(y,y) \,.
\end{equation}
The first term involving $\gamma$ necessarily decays at large distances since $0\leq \gamma\leq 1$ as an operator. But the last term involving $\alpha$ does not decay and is simply a constant in the translation-invariant case. The property of long range order for this type of states is therefore simply equivalent to the non-vanishing of the pairing term $\alpha$.

This type of long range order can also be interpreted as a {\em Bose--Einstein condensation} (BEC) of fermion pairs. In the bosonic case, long range order already shows up in the two-point function, i.e., in the expectation value $\langle a^\dagger(x) a(y) \rangle = \gamma(x,y)$. Exponential decay is expected above a critical temperature, while for small temperatures one expects that $\gamma(x,y)$ converges to a non-zero limit as $|x-y|\to \infty$ (after taking the thermodynamic limit, of course). This is exactly the phenomenon of BEC for bosons. For fermions, $\gamma$ always decays, but the long range order can be visible if one replaces particles by pairs of particles, as in (\ref{asin}). Again, one expects decay for large temperature, but absence of decay below a critical temperature. While the rigorous analysis for this question for the full quantum many-body problem is still out of reach of present-day mathematics, one can investigate the corresponding problem in the simplified BCS model, which is the purpose of this paper.

\section{BCS functional restricted to translation-invariant states}\label{ss:ti}

Our  goal in this section will be to study the BCS  energy functional in the absence of external fields $A$ and $W$. In this case, it makes sense to consider only translation-invariant states, and to calculate the energy of an infinite system per unit volume. (Concerning a justification of this restriction, see Section~\ref{timin} below.) In other words, we assume the one-particle density matrix and Cooper-pair wavefunction to be of the form
 $$\gamma(x,y) = \gamma(x-y), \quad \alpha(x,y) = \alpha(x-y) \,.$$
 It is then natural to express these quantities in terms of their Fourier transform, i.e.,
 \begin{equation}
 \begin{split}
 \gamma(x-y)  & = (2\pi)^{-d} \int_{\R^d} \hat\gamma(p) e^{i p (y-x) } dp \\ 
  \alpha(x-y)  & = (2\pi)^{-d} \int_{\R^d} \hat\alpha(p) e^{i p (y-x) } dp \\ 
 \end{split}
 \end{equation}
 where $0\leq \hat\gamma(p)\leq 1$ and $\hat\alpha(p) = \hat\alpha(-p)$. Moreover, we can combine the Fourier coefficients to a $2\times 2$ matrix, as in (\ref{def:Gammaa}), 
 \begin{equation}\label{def:Gammaaa}
\Gamma(p) =\left(\begin{matrix}\hat\gamma(p)&\hat\alpha(p)\\\overline{\hat\alpha(p)}&1-
\hat\gamma(-p)
\end{matrix}\right)
\end{equation}
satisfying the constraint $0 \leq \Gamma(p) \leq 1$ for every $p\in\R^d$.

If we plug this into \eqref{def:rBCS} and formally calculate the energy per unit volume of an infinite system, we arrive at 
\begin{equation}\label{freeenergya}
 \int_{\R^d} (p^2-\mu)\hat \gm(p)\frac {dp}{(2\pi)^d} +\int_{\R^d} |\alpha(x)|^2 V(x)dx + T  \int_{\R^d} {\rm
Tr}_{\C^2}\left[\Gamma(p)
  \ln \Gamma(p)\right]\frac {dp}{(2\pi)^d} \,.
\end{equation}
In order to avoid having to write the factors $(2\pi)^d$, it turns out to be convenient to redefine $\alpha(x)$ in the form
\begin{equation}
\tilde \alpha(x) = (2\pi)^{d/2} \alpha(x) = (2\pi)^{-d/2} \int_{\R^d} \hat\alpha(p) e^{- i p x } dp \,.
 \end{equation}
 This way, also the interaction term in \eqref{freeenergya}, when expressed in terms of $\tilde\alpha$, gets a factor $(2\pi)^{-d}$, like all the other terms. Hence we can simply multiply the energy functional by $(2\pi)^d$ to get rid of all these factors. In the following, we also drop the tilde from $\alpha$, and thus arrive at the {\em translation-invariant BCS functional}
\begin{equation}\label{freeenergy}
\F(\Gamma)= \int_{\R^d} (p^2-\mu)\hat \gm(p)dp+\int_{\R^d} |\alpha(x)|^2 V(x)dx -T
S(\Gamma)\,,
\end{equation}
where the entropy $S$ is given by 
\begin{equation} S(\Gamma) = - \int_{\R^d} {\rm
Tr}_{\C^2}\left[\Gamma(p)
  \ln \Gamma(p)\right]dp \, .
  \end{equation}
    
 Let us  remark that in the
case of the Hubbard model a functional similar to (\ref{freeenergy}) was studied in \cite{BLS94}.

\subsection{Minimization of the BCS functional}

In the following we shall investigate the properties of the translation-invariant BCS functional \eqref{freeenergy} and its minimizer. In particular we derive its Euler--Lagrange equation, the  {\em BCS gap equation}. For simplicity, we shall restrict our attention to the physically most relevant case $d=3$ in the remainder of this section. The results in this subsection are taken from \cite{HHSS}; parts of the proofs have been modified to make them simpler and more transparent.

\begin{proposition}[Existence of minimizers]
  \label{prop:minimizer}
Let $\mu \in \mathbb{R}$, $0 \leq T < \infty$, and let $V\in 
L^{3/2}(\mathbb{R}^3)$ be
  real-valued.  Then the BCS functional 
  $\mathcal{F}$ in \eqref{freeenergy} is bounded from below and attains a minimizer $(\gamma,\alpha)$ on
  \begin{equation}\label{def:Def}
    \mathcal{D} = \left\{ \Gamma \ \text{of the form (\ref{def:Gammaaa})} \ \middle|
      \begin{array}{l}
        \scriptstyle{\hat{\gamma}\in L^1\displaystyle{(}\scriptstyle\mathbb{R}^3,(1+p^2) dp\displaystyle{)}\scriptstyle,}\\
        \scriptstyle{\alpha\in H^1(\mathbb{R}^3,dx),}
      \end{array}
      0\leq \Gamma \leq 1 
    \right\}.
  \end{equation}
  Moreover, the function
  \begin{equation}
    \label{eq:Delta}
    \Delta(p) =   \frac 2 {(2\pi)^{3/2}} \int_{\R^3} \hat V(p-q) \hat{\alpha}(q) dq
  \end{equation}
  satisfies the BCS gap equation
  \begin{equation}\label{eq:gap}
      \frac 1 {(2\pi)^{3/2}}\int_{\mathbb{R}^3} \hat{V}(p-q)
      \frac{\Delta(q)}{K_{T}^{\Delta}(q)} dq =
      -\Delta(p)\,.
  \end{equation}
\end{proposition}

In (\ref{eq:gap}) we have introduced the notation
  \begin{align}
    \label{eq:K}
    K_{T}^{\Delta}(p) &=
    \frac{E_{\Delta}(p)}{\tanh\big(\frac{E_{\Delta}(p)}{2T}\big)},\\
    \label{eq:E}
    E_{\Delta}(p) &= \sqrt{(p^2 - \mu)^2 + |\Delta(p)|^2}\,.
  \end{align}
Note that for $T=0$ we obtain $$K_{0}^{\Delta}(p) = E_{\Delta}(p).$$

\begin{proof} We only sketch the proof of the proposition, and refer to  \cite{HHSS} for some of the details. 
We start by showing that the functional $\F$ in (\ref{freeenergy}) is bounded from below.
To control the entropy we borrow $1/4$ of the kinetic energy to obtain 
$$
\F(\Gamma) \geq C_1+\frac 34\int (p^2-\mu)\hat \gm(p)dp+\int
|\al(x)|^2 V(x) dx\,,
$$
where
$$ C_1=\inf_{(\gm,\alpha)\in\de}\left( \frac 14\int 
(p^2-\mu)\hat \gm(p)dp-\frac 1\beta S(\Gamma) \right)=-\frac1\beta
\int\ln(1+e^{-\frac\beta4(p^2-\mu)})dp\,.
$$
To compute $C_1$, we have used that in the absence of an interaction $\alpha$ can be taken to be zero, which follows e.g. from the concavity of the entropy. 
Because of our  assumption that $V\in L^{3/2}$, we have  
$ 0 \geq C_2=\inf {\rm spec} \, (p^2/4+V) > - \infty$. Using in addition that $ \hat \gamma(p) \geq |\hat \alpha(p)|^2$, 
 we obtain  
$$
\frac 14 \int p^2\hat\gamma(p) dp + \int V(x) |\alpha(x)|^2 dx \geq C_2 \int |\hat \alpha(p)|^2 dp \geq C_2 \int \hat \gamma(p)dp \,.
$$
Using again  $|\hat\al(p)|^2 \leq \hat \gamma(p)\leq 1$, we further conclude that
\begin{equation}\label{bA}
\F(\Gamma)\geq -A+\frac18
\|\al\|_{H^1(\R^3,dx)}^2+\frac{1}{8}\|\gm\|_{L^1(\R^3,(1+p^2)dp)},\,
\end{equation}
where
$$
A=-C_1- \int \left[p^2/4 - 3 \mu/4 - 1/4 + C_2\right]_-dp\,,
$$
with $[\,\cdot\,]_- = \min\{\,\cdot\,,0\}$ denoting the negative part. This shows, on the one hand, that the functional is bounded from below. On the other hand, it also demonstrates that the energy dominates the relevant norms of $\alpha$ and $\gamma$ on the right side of (\ref{bA}). The proof of existence of minimizers then follows in a rather straightforward way using lower semicontinuity in an appropriate topology; 
we refer to \cite{HHSS} for the details. 
More interesting is the derivation of the BCS gap equation, which we present next.

Since the set $\mathcal{D}$ in \eqref{def:Def} is convex, 
a minimizer $\Gamma = (\gamma,\alpha)$
  of $\mathcal{F}
$ satisfies the inequality
  \begin{equation}
    \label{eq:minimizer}
    0 \leq \left. \frac{\dda}{\dda t} \right|_{t=0} \mathcal{F}\big(\Gamma +
    t(\tilde{\Gamma}-\Gamma)\big)
  \end{equation}
  for all $\tilde{\Gamma} \in \mathcal{D}$.
    A simple
  calculation using
  \begin{equation*}
    S(\Gamma) =
    -\int_{\mathbb{R}^3}\tr_{\mathbb{C}^2}\Gamma\ln\Gamma dp
    =-\frac{1}{2}\int_{\mathbb{R}^3}\tr_{\mathbb{C}^2}\big(\Gamma\ln(\Gamma)+(1-\Gamma)\ln(1-\Gamma)\big)
    dp
  \end{equation*}
  shows that
  \begin{equation}\label{calc1o}
    \left. \frac{\dda}{\dda t} \right|_{t=0}
    \mathcal{F}\big(\Gamma + t(\tilde{\Gamma}-\Gamma)\big) =
    \frac{1}{2}\int_{\mathbb{R}^3}\tr_{\mathbb{C}^2}  (\tilde{\Gamma}-\Gamma) \left[  H_\Delta
     + T 
      \ln\Big(\frac{\Gamma}{1-\Gamma}\Big) \right]dp,
  \end{equation}
  with
  \begin{equation*}
    H_\Delta = \left(
      \begin{matrix}
        p^2 - \mu& \Delta(p) \\
        \bar{\Delta}(p) & -p^2 +  \mu
      \end{matrix}
    \right),
  \end{equation*}
  using the definition
  \begin{equation}\label{dpp}
    \Delta = 2 (2\pi)^{-3/2} \hat{V} * \hat{\alpha}\,.
  \end{equation}
  For $T>0$, it follows from the positivity (\ref{eq:minimizer}) that the eigenvalues of $\Gamma(p)$ stay away from $0$ and $1$ for every fixed, finite $p\in\R^3$, and thus $\tilde\Gamma(p) - \Gamma(p)$ can take values in some open ball containing the zero matrix. It then follows immediately from (\ref{eq:minimizer}) and (\ref{calc1o}) that 
 the Euler--Lagrange equation takes the simple form
  \begin{equation}
    \label{eq:el_Gamma}
    0 = H_\Delta + T \ln\left(\frac{\Gamma}{1-\Gamma}\right),
  \end{equation}
  which is equivalent to
  \begin{equation*}
    \Gamma = \frac{1}{1+\ee^{\frac{1}{T}H_\Delta}}.
  \end{equation*}
  This conclusion can be extended also to the case $T=0$, where the equation simply becomes $\Gamma = \theta(-H_\Delta)$, with $\theta$ denoting the Heaviside step function.
  
  To obtain the explicit form (\ref{eq:gap}) of the BCS gap equation, note that 
   $H_\Delta^2 = E_{\Delta}^2\,
  \id_{\mathbb{C}^2}$ and, therefore,
  \begin{multline}
    \Gamma = \frac{1}{1+\ee^{\frac{1}{T}H_\Delta}}
    = \frac{1}{2} -
    \frac{1}{2}\tanh{\frac{H_\Delta}{2T}} = \frac{1}{2} -
    \frac{1}{2}\frac{H_\Delta}{E_{\Delta}}
\tanh{\frac{E_{\Delta}}{2T}} 
  \\ \label{matequ}
 = \frac{1}{2} -\frac{1}{2 K_{T}^{\Delta}}H_\Delta
= \left(
      \begin{matrix}
        \frac{1}{2} - \frac{p^2-{\mu}}{2 K_{T}^{\Delta}} &
        -\frac{\Delta}{2 K_{T}^{\Delta}} \\
        -\frac{\bar{\Delta}}{2 K_{T}^{\Delta}} & \frac{1}{2} +
        \frac{p^2 - \mu}{2 K_{T}^{\Delta}}
      \end{matrix}
    \right),
  \end{multline}
for the simple reason that $x\mapsto x^{-1} \tanh x $ 
is an even function, and thus 
$$\frac{\tanh \frac{H_\Delta}{2T}}{H_\Delta} =  \frac{\tanh \frac{E_{\Delta}}{2T}}{E_{\Delta}} \id_{\mathbb{C}^2} = \frac 1{K_{T}^{\Delta}} \id_{\mathbb{C}^2}\,.$$
From equation \eqref{matequ} we can read off the Euler--Lagrange equations for $\alpha$ and $\gamma$, which are
 \begin{align}
    \label{eq:el_gamma}
    \hat{\gamma}(p) &= \frac{1}{2} - \frac{p^2 -
    \mu}{2 K_{T}^{\Delta}(p) }
    \\
    \label{eq:el_alpha}
    \hat{\alpha}(p) &= \frac{1}{2}\Delta(p)
    \frac{\tanh\big(\frac{E_{\Delta}(p)}{2T}\big)}{E_{\Delta}(p)} = \frac{\Delta(p)}{2  K_{T}^{\Delta}}\,.
  \end{align}
Performing the convolution with $\hat V$ on both sides of the last equation and using the relation \eqref{dpp} results in the BCS gap equation \eqref{eq:gap}. 
\end{proof}

\begin{remark}
Observe that in the {\em non-interacting} case $V=0$ 
the minimizer of the BCS functional (\ref{freeenergy}) is simply given by 
\begin{equation}\label{G00}
\Gamma_0 =  \frac{1}{1+\ee^{\frac{1}{T}H_0}} = \left(
      \begin{matrix}
        \gamma_0 &
     0\\
       0 &1 -\gamma_0      \end{matrix}
    \right)
\end{equation}
    with $$\gamma_0(p) = \frac{1}{1+\ee^{\frac{1}{T} (p^2 - \mu)}} $$
    the usual Fermi-Dirac distribution.
\end{remark}

We note that with the help of \eqref{eq:el_alpha} the BCS gap equation (\ref{eq:gap}) can equivalently be written in the form
  \begin{equation}\label{Gapeq}
    (K_{T}^{\Delta} + V) {\alpha} = 0,
  \end{equation}
where $K_{T}^{\Delta}$ is interpreted as a multiplication operator in Fourier space, and $V$ as a multiplication operator in configuration space. 
This form of the equation will turn out to be useful later on.
Proposition~\ref{prop:minimizer} gives no information on whether $\Delta \neq 0$ or, equivalently, $\alpha \neq 0$ in a minimizer of the BCS functional, which would mean that  the system is  in a superconducting phase. If the transition from a superconducting to a normal state at the critical temperature is {\em not} first order, and $\Delta$ vanishes continuously as the critical temperature is approached from below, one would expect that the critical temperature $T_c$  is determined by the
 {\em linear }equation $(K_{T_c}^0 + V) \alpha =0$. This is indeed the case, as we are going to show now. 
 
 Observe  that the map
$$ \R_+ \ni \Delta  \to K_{T}^\Delta(p)$$
is  monotone increasing  for all $p$. 
This implies the operator inequality
 $$ K_{T}^{\Delta} + V \geq K_{T}^{0} + V\, .$$
Moreover, if $\Delta$ does not vanish identically 
$$ \langle \alpha | (K_{T}^{0} + V) | \alpha \rangle < \langle \alpha |  (K_{T}^{\Delta} + V) ) \alpha \rangle =0,$$
since, in fact $ K_{T}^{\Delta}(p) > K_{T}^{0}(p) $ for all $p$ in the set where $\Delta(p) \neq 0$, and also 
$\hat \alpha$ does not vanish on this set. In particular, this implies that the operator $K^0_T + V$ must have a negative eigenvalue. (Note that the essential spectrum of $K_T^0 + V$ starts at $\inf_p K_T^0(p) = 2 T$.) 

We will now show that also the converse holds true:  if $K_{T}^{0} + V$ has a negative eigenvalue, then the minimizer of the BCS functional has a non-trivial $\alpha$ and, in particular, the BCS gap equation has a non-vanishing solution. 
This  implies that the question whether the BCS gap equation has a non-trivial solution is reduced to the study of the spectrum of a linear operator. This is the content of the following theorem, which was proved in \cite{HHSS}.
This linear characterization of the existence of solutions to the nonlinear BCS equation represents a considerable simplification of the analysis. In particular, it proves the existence of solutions to \eqref{eq:gap} for a large class of interaction potentials $V$, and hence generalizes previous results \cite{e1,e2,e3,e4}
 valid only for non-local $V$Õs under suitable assumptions.

\begin{theorem}[Linear criterion for non-trivial solutions]\label{thm:minimizer}
  Let $V\in L^{3/2}(\R^3)$, $\mu\in\R$, and $\infty > T= \frac 1\beta \geq 0$.
  Then the following statements are equivalent:
\begin{itemize}
\item[(i)]
The normal state $(\gamma_0,0)$ is unstable under pair formation,
i.e.,
\begin{equation*}
\inf_{(\gamma,\alpha)\in\de} \F(\gm,\al) <
\F(\gm_0,0)\,.
\end{equation*}

\item[(ii)] There exists a pair $(\gamma,\alpha)\in \de$, with $\alpha
  \neq 0$, such that 
$\Delta(p) = 2 (2\pi)^{-3/2}  \hat V \ast\hat \alpha(p)$ satisfies the BCS gap equation
\begin{equation}\label{bcsgapequationfiniteT}
\phantom{\int\int}
 \Delta=- \frac 1 {(2\pi)^{3/2}} \hat V \ast \frac{\Delta}{K_{T}^\Delta}\,.
\end{equation}
\item[(iii)] The linear operator
\begin{equation}\label{linopfinT}
K_{T}^{0} + V ,\quad \quad K_{T}^0 = \frac {p^2 - \mu} {\tanh \frac{p^2 - \mu}{2T}} \,,
\end{equation}
 has at least one negative eigenvalue.
\end{itemize}
\end{theorem}

\begin{proof}
We have already shown above that $(i)\Longrightarrow (ii) \Longrightarrow (iii)$, hence the only thing left to show is the direction $(iii) \Longrightarrow (i)$. We consider, for simplicity, only the case $T>0$, and leave the similar analysis of the $T=0$ case to the reader. 
 Consider the  function
 \begin{equation}\label{F2}
  t \to \F\left(\Gamma_0 + t \left(
    \begin{smallmatrix}
      0 & \varphi(p)\\ \bar \varphi(p) & 0
    \end{smallmatrix}\right)\right)\,,
 \end{equation}
 with $\Gamma_0$ the normal state defined in (\ref{G00}), and $\varphi$ a rapidly decaying function (or of compact support) such that the argument of $\F$ in \eqref{F2} is in $\mathcal{D}$ for $t$ small enough. 
The first derivative of \eqref{F2} at  $t=0$ vanishes because $\Gamma_0$ is a stationary point of $\F$. 
The second derivative can be calculated straightforwardly, with the result that  
\begin{equation}\label{comp2}
 \frac {
d^2}{dt^2} \Big |_{t=0}  \F\left(\Gamma_0 + t \left(
    \begin{smallmatrix}
      0 & \varphi(p)\\ \bar \varphi(p) & 0
    \end{smallmatrix}\right) \right) =2 \langle \varphi | (K_{T}^{0} + V) | \varphi \rangle\,.
 \end{equation}
 Under the assumption of existence of a negative eigenvalue of $K_{T}^{0} + V$, this can be made negative by an appropriate choice of $\varphi$. 
This proves the statement. 

The only non-trivial part of the calculation in (\ref{comp2}) concerns the second derivative of the entropy. To compute it, one can either calculate the eigenvalues of the $2\times 2$ matrix in the argument of (\ref{F2}) and then differentiate twice. Or, alternatively and slightly more elegantly, one can use a contour integral representation for the 
logarithm in the following way. With $f(s) = \frac 12 (s\ln s +(1-s) \ln (1-s) )$ and
$$G(p)= \left( \begin{matrix}
      0 & \varphi(p)\\ \bar \varphi(p) & 0
    \end{matrix}\right) $$
we have, for every $p\in \R^3$, 
\begin{equation}\label{xxy}
\frac d{dt} \Tr f(\Gamma_0 + t G)  =  \Tr f'(\Gamma_0 + t G)G \,,
\end{equation}
and  we thus have to differentiate the right side of (\ref{xxy}) at $t=0$.  For this purpose, we rewrite this expression as 
$$
\Tr f'(\Gamma_0 + t G)G = \frac 1{2\pi i} \oint_C dz\, f'(z) \Tr \frac{1}{z-\Gamma_0 - t G} G 
$$
where $C$ is a circle centered at $1/2$ with radius less than $1/2$ but big enough to enclose the eigenvalues of $\Gamma_0(p) + t G(p)$. (Such a radius exists for every given $p\in \R^3$.) 
Since  $\{\Gamma_0,G\} = G$
 and $$ \left \{G , \frac 1{z-\Gamma_0}  \right\} =    \frac 1{z-\Gamma_0}  \{ G, z-\Gamma_0 \}  \frac 1{z-\Gamma_0}  = (2z-1)  \frac 1{z-\Gamma_0} G \frac 1{z-\Gamma_0} ,$$
with $\{A,B\} = A B + BA$ denoting the anti-commutator,  we obtain 
 \begin{multline}
\frac d{dt} \Big |_{t=0}  \Tr f'(\Gamma_0 + t G)G = \frac 1{2\pi i}  \int_C dz\,  f'(z) \Tr  \frac 1{z-\Gamma_0} G \frac 1{z-\Gamma_0} G \\ = 
\frac 1{2\pi i}\Tr  \int_C dz\,  \frac{f'(z)}{2z-1} \Tr\left\{  \frac 1{z-\Gamma_0}, G \right \} G = \frac 1{2\pi i} \int_C dz\,  \frac{2 f'(z)}{2z-1}\Tr   \frac 1{z-\Gamma_0} G^2\\ = \Tr \frac{2f'(\Gamma_0)}{2\Gamma_0-1} G^2 =
\Tr \frac{(H_0/T)}{\tanh \frac {H_0}{2T} }G^2 = \frac 2 T K_T^0(p) |\varphi(p)|^2 \,.
 \end{multline}
In combination with a dominated convergence argument to interchange the integration in $p$ and the differentiation in $t$, this implies the desired result (\ref{comp2}).
\end{proof}

\subsection{Critical temperature} 

Theorem \ref{thm:minimizer} enables a precise definition of the critical
temperature for the translation-invariant BCS model (\ref{freeenergy}), by
\begin{equation}\label{crittemp}
T_c (V): = \inf \{ T \, | K_{T} +  V \geq 0\}.
\end{equation}
Here and in the following, we drop the superscript $0$ in the kinetic energy, and simply write $K_T$ instead of $K_T^0$. 
In other words, the critical temperature $T_c$ is given by the value of $T$ such that 
\begin{equation}\label{infspec}
 \infspec \left(  K_{T_c} +  V \right) = 0 \,.
 \end{equation}
  As already remarked above, the infimum of the spectrum is necessarily an eigenvalue if $T_c > 0$ since the essential spectrum of $K_T + V$ starts at $2T$.  If there is no such $T>0$ such that (\ref{infspec}) is satisfied, then $T_c(V)$ is zero. This is the case if and only if $|-\nabla^2 - \mu| + V(x) \geq 0$. 
The uniqueness of the critical temperature follows from the fact that 
the function $K_{T}(p)$ is point-wise monotone increasing in $T$. This
implies that for any potential $V$, there is a unique critical temperature
$0\leq T_c ( V)< \infty$ that separates two phases, a {\em
superconducting} phase for $ 0 \leq T < T_c(V)$ from a {\em
normal} phase for $ T_c( V) \leq T < \infty$. Note that $T_c(V) = 0$ means that there is no superfluid phase for the potential $ V$.

With the aid of the linear criterion (\ref{infspec}) we can classify the potentials for which
$T_c( V)
> 0$, and give sufficient conditions on $V$ for this to happen. We shall also evaluate the
asymptotic behavior of $T_c(V)$ in certain limiting parameter regimes, like weak coupling, for instance.
For this purpose,  we introduce a  coupling parameter $\lambda$.
Note that if one can show that the lowest eigenvalue of $K_{T} + \lambda V$ is negative for small enough coupling then it is negative for 
all larger values of $\lambda$. In fact, $\infspec (K_T + \lambda V)$ is a concave and monotonously decreasing function of $\lambda$.

By applying the Birman--Schwinger principle, which we review in Section~\ref{birsch} below, one observes that
the critical temperature $T_c$ can be characterized by the
fact that the compact operator
\begin{equation}\label{eq:bsop}
\lambda (\sgn V) |V|^{1/2} K_{T_c}^{-1} |V|^{1/2}
\end{equation}
has $-1$ as its lowest eigenvalue. If $\mu>0$, then this operator is singular for $T_c
\to 0$, and the key observation in \cite{FHNS} was that its singular part can be 
represented by the operator $\lambda \ln(1/T_c)\V_\mu$, where $\V_\mu:
\, L^2(\Omega) \mapsto L^2(\Omega)$ is given by
\be\label{defvm} \big(\V_\mu u\big)(p) = \frac 1{(2\pi)^{3/2}} \int_{\Omega}\hat V(\sqrt \mu(p-q)) u(q) \,d\omega(q)\,.  \end{equation}
Here, $\Omega$ denotes the unit $2$-sphere,
and $d\omega$ denotes Lebesgue measure on $\Omega$. We note that
the operator $\V_\mu$ has appeared already earlier in the literature
\cite{BY,LSW}. As a result of this analysis, one can obtain an asymptotic formula for $T_c(\lambda V)$ as $\lambda \to 0$, expressed in terms of the lowest eigenvalue of the operator $\V_\mu$. This is the content of Theorem~\ref{thm2.2} below.

The analysis presented here is somewhat similar in spirit to the one concerning
the lowest eigenvalue of the Schr\"odinger operator $-\nabla^2+\lambda V$ in
\emph{two} space dimensions \cite{simon}. This latter case is
considerably simpler, however, as $p^2$ has a unique minimum at $p=0$,
whereas $K_{T}(p)$ takes its minimal value on the Fermi sphere
$p^2=\mu$, meaning that its minimum is highly degenerate. Hence, in
our case, the problem is reduced to analyzing a map from the $L^2$
functions on the unit sphere 
to itself. In fact it would be a map from the Fermi-sphere to itself, but for sake of convenience we rescaled it to the unit sphere.

Let us denote the lowest eigenvalue of $\V_\mu$ as
$$e_\mu(V) := \infspec  \V_\mu\,.$$
As the next theorem shows, whenever this eigenvalue is negative then the critical temperature
is non-zero for {\em all} $\lambda > 0$, and we can evaluate its
asymptotics. The following was proved in \cite[Theorem 1]{FHNS}.

\begin{theorem}[Critical temperature for weak coupling]\label{thm2.2}
  Let $V\in L^{3/2}(\R^3)\cap L^1(\R^3)$ be real-valued, and let
  $\mu>0$.

Assume that $e_{\mu}(V)<0$. Then $T_c(\lambda V)$ is non-zero for
all $\lambda >0$, and
\begin{equation}\label{symptbeh}
\lim_{\lambda\to 0}   \lambda\, \ln \frac{\mu}{T_c(\lambda V)} =
-\frac{1}{\sqrt{\mu} e_{\mu}(V)} \,.
\end{equation}
\end{theorem}

Negativity of $e_\mu(V)$ thus implies the existence of a superconducting phase in the BCS model for all values of the coupling constant. 
A sufficient
condition for $e_\mu(V)$ to be negative is $ \int_{\R^3}V(x) dx < 0$, since the latter is proportional to the trace of the operator $\V_\mu$. But one can
easily find other examples.  Eq.~\eqref{symptbeh} shows that the
critical temperature behaves like $ T_c(\lambda V) \sim \mu e^{
  1/(\lambda \sqrt \mu e_{\mu}(V))}$ as $\lambda\to 0$.  In particular, it is exponentially small
in the coupling.

\subsubsection{Radial potentials.}\label{radpot}

In the special case of radial potentials $V(x)$, depending only on
$|x|$, the spectrum of $\V_\mu$ can be determined more explicitly.
Since $\V_\mu$ commutes with rotations in this case, all its
eigenfunctions are given by spherical harmonics. For $\ell$ a
non-negative integer, the eigenvalues of $\V_\mu$ are then given by
$\frac{1}{2\pi^2} \int V(x) |j_\ell(\sqrt\mu |x|)|^2 dx$,
with $j_\ell$ denoting the spherical Bessel functions. These
eigenvalues are $(2\ell+1)$ fold degenerate.  In particular, we then
have
$$
e_\mu(V) = \inf_{\ell \in \N} \, \frac{1}{2\pi^2} \int_{\R^3} V(x)
\left|j_\ell(\sqrt\mu |x|)\right|^2 dx\,.
$$
We remark that
$\sum_{\ell\in\N}(2\ell+1)|j_\ell(r)|^2=1$, hence we recover the statement  above that $\int_{\R^3} V(x) dx < 0$ 
implies that $e_\mu(V)$ is negative.

If $\hat V$ is non-positive, it is easy to see that the infimum is
attained at $\ell=0$.  This follows since the lowest eigenfunction can
be chosen non-negative in this case, and is thus not orthogonal to the
constant function. Since $j_0(r)=\sin(r)/r$, this means that $$e_\mu(V)
=\frac 1{2\pi^2} \int_{\R^3} V(x) \frac{\sin^2(\sqrt{\mu}|x|)}{\mu |x|^2}
dx$$ for radial potentials $V$ with non-positive Fourier transform.
If $\hat V$ does not have a definite sign, then the lowest eigenvalue $e_\mu(V)$ can be degenerate, and the same has to be true for 
the operator $K_T + \lambda V$ in the case of $T$ and $\lambda$ small enough. In \cite{FL} it was  shown that for any given angular momentum $\ell$ there is a potential $V$ such that 
the degeneracy of $e_\mu(V)$ is $2\ell +1$.

In the limit of small $\mu$ we can use the asymptotic behavior $j_\ell(r)
\approx r^\ell/(2\ell+1)!!$ to observe that, in case $\int V(x)dx<0$,
$e_\mu(V) \approx \frac{1}{2\pi^2} \int V(x) dx$ as $\mu\to 0$.
Note that $(\lambda/4\pi)\int V(x) dx$ is the first Born approximation
to the {\it scattering length} of $2\lambda V$, which we denote by
$a_0$. Thus, replacing $\lambda e_\mu(V)$ by $2 a_0/\pi$ and
writing $\mu = k_{\rm f}^2$, we arrive at the expression
$T_c\sim e^{\pi/(2k_{\rm f}a_0)}$ for the critical temperature, which
is well established in the physics literature \cite{gorkov,NRS,zwerger}. It is valid not only at weak coupling, as will be shown in Section~\ref{ss:ldl} below.

\medskip

In the next section, we will give the proof of Theorem~\ref{thm2.2}, as well as further results on the asymptotic behavior of the BCS critical temperature. A key tool will be the Birman--Schwinger reformulation of the Schr\"odinger equation, which was already mentioned above.

\subsection{Birman--Schwinger argument} \label{birsch}

For a general real-valued potential $V$ let us use the
notation
\begin{equation*}
  V^{1/2}(x) = (\sgn V(x)) |V(x)|^{1/2} \,.
\end{equation*}
Fix a coupling parameter $\lambda>0$, and consider the operator $ K_{T} + \lambda V$.
Recall that  the critical temperature $T_c$ was defined in a way 
 that for $T=T_c$ the operator $ K_{T} + \lambda V$
has  $0$ as lowest eigenvalue eigenvalue. If $\psi$ is the
corresponding eigenvector, one can rewrite the eigenvalue equation
in the form
$$ -\psi = \lambda K_{T}^{-1} V \psi = \lambda K_{T}^{-1} |V|^{1/2} V^{1/2} \psi.$$
Multiplying this equation by $V^{1/2}(x)$, one obtains 
an eigenvalue
equation for $\varphi = V^{1/2} \psi$,
$$\varphi = V^{1/2} \psi = - \lambda V^{1/2} K_{T}^{-1}  |V|^{1/2} V^{1/2} \psi = - \lambda V^{1/2} K_{T}^{-1}  |V|^{1/2} \varphi.$$
For $T>0$,  this argument works in both
directions and is called the Birman--Schwinger principle. In particular it  tells us that the critical
temperature $T_c$ is determined by the fact that for this value of
$T$ the smallest eigenvalue of
\begin{equation}\label{defofbt}
  B_T = \lambda V^{1/2}K_{T}^{-1}|V|^{1/2}
\end{equation}
equals $-1$.  

To be precise, we just argued that $-1$ is an eigenvalue of $B_{T_c}$. 
 Moreover, because of strict monotonicity of
  $K_{T}$ in $T$, $-1$ is not an eigenvalue of $B_T$ for any $T >
  T_c$. This implies that $B_{T_c}$ has no eigenvalue less than $-1$,
  for otherwise there would be a $T>T_c$ for which $B_T$ has
  eigenvalue $-1$ since the eigenvalues of $B_T$ depend continuously
  on $T$ and approach $0$ as $T\to \infty$. In the same way, one argues that if $T_c = 0$ then $B_T$ does not have an eigenvalue
  less than or equal to $-1$ for any $T>0$.

Observe that although $B_T$ is not self-adjoint, it has real
spectrum. This follows from the fact that $B_T$ is isospectral to the self-adjoint operator
\begin{equation}\label{sabs}
\lambda K_T^{-1/2} V K_T^{-1/2}\,.
\end{equation}
Instead of (\ref{defofbt}), we could thus as well work with (\ref{sabs}), but it turns out that for our asymptotic analysis it is more convenient to use $B_T$.

\subsubsection{Weak coupling limit}

Let $\dig: L^1(\R^3) \to L^2(\Omega)$ denote the (bounded)
operator which maps $\psi\in L^1(\R^3)$ to the Fourier transform of
$\psi$, restricted to the unit sphere $\Omega$, with 
$$(\dig \psi)(p) = \frac 1{{2\pi}^{3/2}} \int_{\R^3} \psi(x) e^{i \sqrt{\mu} x\cdot p}  dx.$$ 
Since $V\in L^1(\R^3)$, the 
multiplication by $|V|^{1/2}$ is a bounded operator from $L^2(\R^3)$ to
$L^1(\R^3)$, and hence $\dig |V|^{1/2}$ is a bounded operator from
$L^2(\R^3)$ to $L^2(\Omega)$. To see the reason for this definition, we rewrite
$$ V^{1/2}K_{T}^{-1}|V|^{1/2} = V^{1/2} \left[ \frac 1{K_{T} } - \frac 1{p^2}  \right] |V|^{1/2} + V^{1/2} \frac 1{p^2} |V|^{1/2} ,$$
and decompose
\begin{multline} \label{deco}
\int_{\R^3} e^{i p\cdot (x-y) } \left[  \frac 1{K_{T} } - \frac 1{p^2}  \right] dp = \int_{\R^3} e^{i \sqrt{\mu} \frac{p}{|p|}\cdot (x-y) } \left[  \frac 1{K_{T} } - \frac 1{p^2}  \right] dp \\
+ \int_{\R^3} \left( e^{i p\cdot(x-y)} - e^{i \sqrt{\mu} \frac{p}{|p|}\cdot (x-y) }  \right)\left[  \frac 1{K_{T} } - \frac 1{p^2}  \right] dp \,.
\end{multline}
With
\begin{equation}\label{muT}
m_\mu(T) = \frac 1{4\pi} \int_{\R^3} \left( \frac 1{K_{T}(p)}
  - \frac 1{p^2}\right) dp\,  ,
\end{equation}
we can rewrite the first term on the right side of \eqref{deco} as 
$$ \int_{\R^3} e^{i \sqrt{\mu} \frac{p}{|p|}\cdot (x-y) } \left[  \frac 1{K_{T} } - \frac 1{p^2}  \right] dp = m_\mu(T)   \int_{\Omega} e^{i \sqrt{\mu} \omega \cdot (x-y) } d \omega \,,$$
where $d\omega$ denotes  Lebesgue measure on the sphere $\Omega$. 
Note that $\int_{\Omega} e^{i \sqrt{\mu} \omega \cdot (x-y) } d \omega$ is the integral kernel of the operator $ \dig^\dagger \dig$. Denoting by $A_T$ the operator whose integral kernel is given on the second line on the right side of \eqref{deco}, we thus have 
\begin{equation}
  \frac 1{K_{T} } - \frac 1{p^2}  = m_\mu(T) \dig^\dagger \dig + A_T \,.
\end{equation}
Let further 
\begin{equation}\label{defmt}
  M_T = K_{T}^{-1} - m_\mu(T) \dig^\dagger \dig = A_T + p^{-2}.
\end{equation}
The following lemma \cite[Lemma~3.2]{FHNS} shows that $ V^{1/2}
M_T |V|^{1/2}$ is a bounded operator on $L^2(\R^3)$, and its
norm is bounded uniformly in $T$. In particular,  the
singular part of $B_T$ as $T\to 0$ is entirely determined by
$V^{1/2}\dig^\dagger \dig |V|^{1/2}$.

\begin{lemma}\label{lem33}
$ V^{1/2}
M_T |V|^{1/2}$ is a bounded operator on $L^2(\R^3)$ with a norm  bounded uniformly in $T$. 
\end{lemma}

\begin{proof}
Boundedness of $V^{1/2} p^{-2} |V|^{1/2}$ follows from the Hardy--Littlewood--Sobolev and the H\"older inequality\footnote{Here and the following, we shall use the notation $C$ for generic constants, possible having a different value in each appearance.}
 \begin{multline}
 \langle f | V^{1/2} \frac 1{p^2} |V|^{1/2}| g\rangle = C \iint_{\R^3\times \R^3} \bar f(x) V^{1/2}(x) \frac 1{|x-y|} |V|^{1/2}(y) g(y) dx dy \\ \leq C \| V^{1/2} f \|_{6/5} \|g |V|^{1/2} \|_{6/5} \leq C \|V\|_{3/2} \|f\|_2 \|g\|_2. 
 \end{multline}
Hence it suffices to investigate the operator $V^{1/2} A_T |V|^{1/2}$. 

By integrating out the angular variable, we can obtain the following bound on the integral kernel of $A_T$:
\begin{multline}
\left|  \int_{\R^3} dp \left( e^{i p\cdot(x-y)} - e^{i \sqrt{\mu} \frac{p}{|p|}\cdot (x-y) }  \right)\left[  \frac 1{K_{T} } - \frac 1{p^2}  \right]   \right| \\ = \left| 
\int dp \left [ \frac {\sin(|p||x-y|)}{|p||x-y| } - \frac{\sin (\sqrt{\mu} |x-y| ) }{\sqrt{\mu} |x-y|} \right ] \left[  \frac 1{K_{T} } - \frac 1{p^2}  \right]  \right| \\
\leq C \int dp \left | \frac 1{K_{T} } - \frac 1{p^2}  \right | \frac{| |p| - \sqrt{\mu}|}{|p| + \sqrt{\mu}} ,
\end{multline}
using that $|\sin a/a - \sin b/b| \leq C |a-b|/|a+b|$. 
Since for $T=0$, $K_{T} = |p^2 - \mu| =||p| - \sqrt{\mu}|||p| + \sqrt{\mu}|$, we see that 
 $$ \int dp \left | \frac 1{K_{T} } - \frac 1{p^2}  \right | \frac{| |p| - \sqrt{\mu}|}{|p| + \sqrt{\mu}} $$
 is bounded,  uniformly in $T$ for $T$ in bounded intervals. It then follows right away that 
 $$ \| V^{1/2} A_T |V|^{1/2}\|_{\rm HS} \leq C \|V\|_1\,,$$
 with $\| \, \cdot\, \|_{\rm HS}$ denoting the Hilbert--Schmidt norm. This completes the proof of the uniform bound for bounded $T$. By a more careful analysis, one can actually show that the bound is uniform for all $T>0$, and we refer the reader to \cite[Lemma~3.2]{FHNS} for the details.
\end{proof}

Since $V^{1/2} M_T |V|^{1/2}$ is uniformly bounded, we can choose
$\lambda$ small enough such that  $1+\lambda V^{1/2} M_T |V|^{1/2}$
is invertible. We can then write $1+ B_T$ as
\begin{align}\label{1ba}
  1+ B_T &= 1+ \lambda V^{1/2} \left( m_\mu(T) \dig^\dagger \dig +
    M_T\right) |V|^{1/2} \\ \nonumber &= \left(1+ \lambda V^{1/2} M_T
    |V|^{1/2} \right) \left( 1 + \frac{\lambda m_\mu(T)}{1+ \lambda
      V^{1/2} M_T |V|^{1/2}} V^{1/2} \dig^\dagger \dig |V|^{1/2}\right)\,.
\end{align}
In particular,  $B_T$ having an eigenvalue $-1$  is equivalent to
\begin{equation}\label{a1}
  \frac{\lambda m_\mu(T)}{1+ \lambda V^{1/2} M_T |V|^{1/2}} V^{1/2}
  \dig^\dagger \dig |V|^{1/2}
\end{equation}
having an eigenvalue $-1$. The operator in (\ref{a1}) is
isospectral to the selfadjoint operator
\begin{equation}\label{b1}
\dig |V|^{1/2} \frac {  \lambda m_\mu(T) }{1+ \lambda V^{1/2} M_T
    |V|^{1/2}} V^{1/2} \dig^\dagger\,,
\end{equation}
acting on $L^2(\Omega_\mu)$.

Consequently, the defining equation for the critical temperature can be written
as \be\label{equcrittemp} \lambda m_\mu(T_c) \, \infspec \dig
|V|^{1/2} \frac { 1}{1+ \lambda V^{1/2} M_{T_c} |V|^{1/2}} V^{1/2}
\dig^\dagger = -1\,. \end{equation}   Up to first order in $\lambda$ the equation \eqref{equcrittemp}
  reads \be\label{crittempfirst} \lambda m_\mu(T_c)\, \infspec \dig
  [V- \lambda V M_{T_c}V + O(\lambda^2) ]\dig^\dagger = -1\,, \end{equation} where the
  error term $O(\lambda^2)$ is uniformly bounded in $T_c$.  Note that
  $\dig V \dig^\dagger = \, \V_\mu$, defined in (\ref{defvm}). Assume
  now that $e_\mu = \infspec \V_\mu$ is strictly negative. Since
  $V^{1/2} M_{T_c} V^{1/2}$ is uniformly bounded, it follows
  immediately that
$$
\lim_{\lambda \to 0} \lambda m_\mu(T_c) =- \frac 1{\infspec \dig V
  \dig^\dagger} = - \frac 1 { \, e_\mu}\,.
$$
Together with the asymptotic behavior $m_\mu(T) \sim \mu^{1/2}
\ln(\mu/T)$ as $T\to 0$, this implies the leading order behavior of
$\ln (\mu/T_c)$ as $\lambda\to 0$ and proves the statement of Theorem \ref{thm2.2}. \hfill \qed 

\bigskip

Theorem~\ref{thm2.2} gives the leading order asymptotic behavior of $T_c(\lambda V)$. However, by going to the next order one can go one step further and determine the right constant in front of the exponential. 
In fact,  it is possible to define an operator $\W_\mu$ via the quadratic form
\begin{equation}\label{deno2}
  \lim_{T\to 0} \langle u| \dig V M_T V \dig^\dagger| u\rangle  = \langle u| \W_\mu |u \rangle\,.
\end{equation}
It then follows from \eqref{crittempfirst}  that 
\begin{equation}\label{denof}
  \lim_{\lambda\to 0} \left( m_\mu(T_c) + \frac 1{\infspec \left
        (\lambda  \V_\mu - \lambda^2 \W_\mu\right)} \right) =
  0\,.
\end{equation}
The following Theorem \cite[Theorem 1]{HS} is then a consequence  of  the   asymptotic behavior \cite[Lemma 1]{HS}
  \begin{equation}\label{lemresult}
    m_\mu(T) = {\sqrt\mu}
\left( \ln \frac \mu T + \gamma-2 + \ln\frac 8\pi  + o(1)\right)
  \end{equation}
  in the limit of small $T$, where $\gamma\approx 0.5772$ is Euler's
  constant.

\begin{theorem}[Refined weak coupling asymptotics] \label{constant}
  Let $V\in L^1(\R^3)\cap L^{3/2}(\R^3)$ and let $\mu>0$.  Assume that
  $e_\mu = \infspec \V_\mu <0$, and let $b_\mu(\lambda)$ be defined 
  by
  \begin{equation}\label{defbm}
  b_\mu(\lambda) =\infspec \left
        (\lambda  \V_\mu - \lambda^2 \W_\mu\right) \,.
\end{equation}
 Then the critical temperature $T_c$ for the BCS
  equation is strictly positive and satisfies
  \begin{equation}\label{themeq}
    \lim_{\lambda \to 0} \left(\ln\left(\frac\mu {T_c}\right) +
      \frac {1}{ \sqrt{\mu}\, b_\mu(\lambda)}\right) = 2 - \gamma - \ln(8/\pi)\,.
  \end{equation}
\end{theorem}

The theorem states that in the weak coupling limit, the formula
\begin{equation}\label{formula}
  T_c (\lambda V) = \mu \left(  \frac 8 \pi  e^{\gamma-2} + o(1) \right) e^{1/( \sqrt{\mu} b_\mu(\lambda))} \,.
\end{equation}
holds for the BCS critical temperature. 

\subsubsection{Low density limit} \label{ss:ldl}

While in the previous subsection we investigated the critical temperature in the case of low coupling, we are now 
interested in the low density limit $\mu \to 0$ at fixed interaction
potential $V$. In this  regime, $T_c$ turns out to be related to the
scattering length of $2V$. As shown in  \cite{HS2}, the latter can be conveniently defined as follows:

\begin{definition}\label{def2}
Let  $V\in L^1(\R^3)\cap L^{3/2}(\R^3)$ be real-valued, and let $V^{1/2}(x) = {\rm sgn}(V(x)) |V(x)|^{1/2}$. If $-1$ is not in the spectrum of the Birman--Schwinger operator  $V^{1/2}
\frac 1{p^2} |V|^{1/2}$, then the {\it scattering length} of $2V$ is given by 
\begin{equation}\label{defa}
a = \frac 1{4\pi} \langle |V|^{1/2}| \frac 1 {1+V^{1/2} \frac 1{p^2} |V|^{1/2}} |V^{1/2}\rangle \,.
\end{equation}
If $1+V^{1/2}
\frac 1{p^2} |V|^{1/2}$ is not invertible, $a$ is infinite.
\end{definition}

Note that since $V\in L^{3/2}(\R^3)$ by assumption, the
Birman--Schwinger operator is of Hilbert--Schmidt class. As above, although it is not self-adjoint, its spectrum is real. The fact that $-1$ is
not in its spectrum means that $p^2 + V$ does not have a zero
eigenvalue or resonance. Eq.~(\ref{defa}) is the natural definition of
the scattering length for integrable potentials.  In the appendix of \cite{HS2}, it is 
explained why this definition coincides with the usual concept of scattering length 
found in quantum mechanics textbooks. 

We can now state the behavior of the critical temperature in the 
low density limit $\mu \to 0$. The following is proved in \cite{HS2}.

\begin{theorem}[Critical temperature at low density]\label{main}
  Assume that $V(x)(1+|x|)\in L^1(\R^3)\cap L^{3/2}(\R^3)$ is
  real-valued, and $\mu>0$. Assume further that the spectrum of $V^{1/2} \frac 1{p^2}
  |V|^{1/2}$ is contained in $(-1,\infty)$, and that the scattering
  length $a$ in (\ref{defa}) is negative.  Then the critical
  temperature $T_c$ satisfies
\begin{equation}\label{eq:main}
\lim_{\mu \to 0}  \left(\ln\frac \mu {T_c} + \frac \pi {2\sqrt\mu\, a}\right) =  2-\gamma -\ln\frac 8\pi 
\end{equation}
with $\gamma\approx 0.577$ denoting Euler's constant. 
\end{theorem}

In other words,
$$
T_c(V) = \mu \left( \frac 8\pi e^{\gamma -2} + o(1) \right) e^{\pi/(2\sqrt \mu a)}
$$
as $\mu\to 0$. This formula is well-known in the physics literature
\cite{gorkov,Leggett,NRS,randeria}. The operator $V^{1/2} \frac 1{p^2} |V|^{1/2}$
having spectrum in $(-1,\infty)$ implies, in particular, that $p^2 +V$
does not have any bound states. 

\begin{proof}[Sketch of the Proof]
According to the Birman--Schwinger principle, discussed in Section \ref{birsch},  $T_c$ is determined by the fact that for
$T=T_c$ the smallest eigenvalue of
$$
B_T = V^{1/2}\frac 1 {K_{T}} |V|^{1/2}
$$ 
equals $-1$.  Alternatively, $T_c$ is the largest $T$
such that $1+B_T$ has an eigenvalue 0.

We start again with a convenient decomposition of $B_T$, this time as 
$$
B_T= V^{1/2} \frac 1{K_{T}} |V|^{1/2} = V^{1/2} \frac 1{p^2}
|V|^{1/2} + m_\mu(T) |V^{1/2}\rangle\langle |V|^{1/2}| + A_{T,\mu} \,,
$$
where we use, for  convenience, this time the definition
$$
m_\mu(T) = \frac 1{(2\pi)^3} \int_{\R^3} \left( \frac 1{K_{T}(p)} - \frac 1 {p^2} \right) dp\,.
$$
Note that this differs from the definition \eqref{muT} by a factor of $2\pi^2$. 
Explicitly, $A_{T,\mu}$ is the operator with integral kernel
$$
A_{T,\mu}(x,y) = V(x)^{1/2} |V(y)|^{1/2}  \frac1{(2\pi)^3}\int_{\R^3} \left( e^{ip(x-y)} - 1\right) \left( \frac 1{K_{T}(p)} - \frac 1 {p^2} \right) dp\,.
$$
We note that, as in (\ref{lemresult}), 
\begin{equation}\label{defm}
m_\mu(T) = \frac {\sqrt{\mu}}{2\pi^2} \left( \ln\frac \mu T + \gamma -2 +\ln \frac 8 \pi + o(1) \right)
\end{equation}
  for small $\mu$, uniformly in $T$ for $T\leq C\mu$. This was shown in \cite[Lemma~1]{HS}.

Since $1+V^{1/2}p^{-2}|V|^{1/2}$ is invertible by assumption, we can write
\begin{multline}
1+B_T = \left( 1+V^{1/2} \frac 1{p^2} |V|^{1/2}\right) \\ \times \left( 1 + \frac{m_\mu(T)}{1+V^{1/2}p^{-2}|V|^{1/2}}\left( |V^{1/2}\rangle\langle |V|^{1/2}| + \frac{A_{T,\mu}}{m_\mu(T)}\right) \right)\,.
\end{multline}
Since the first term is invertible we know that the critical temperature is the largest $T$ such that 
$$ 1 + \frac{m_\mu(T)}{1+V^{1/2}p^{-2}|V|^{1/2}}\left( |V^{1/2}\rangle\langle |V|^{1/2}| + \frac{A_{T,\mu}}{m_\mu(T)}\right)$$
has an eigenvalue  $0$. 
We again rewrite 
\begin{align}\nonumber
& 1 + \frac{m_\mu(T)}{1+V^{1/2}p^{-2}|V|^{1/2}}\left( |V^{1/2}\rangle\langle |V|^{1/2}| + \frac{A_{T,\mu}}{m_\mu(T)}\right) \\ \nonumber &= \left(1 + \frac{1}{1+V^{1/2}p^{-2}|V|^{1/2}} A_{T,\mu} \right) \\ & \quad\quad  \times 
\left( 1 + \frac 1{1 + \frac{1}{1+V^{1/2}p^{-2}|V|^{1/2}} A_{T,\mu} } \frac{m_\mu(T)}{1+V^{1/2}p^{-2}|V|^{1/2}} |V^{1/2}\rangle\langle |V|^{1/2}| \right). \label{2f}
\end{align}
By an analysis similar to the one of Lemma~\ref{lem33}, one can show that  the Hilbert--Schmidt norm of $A_{T,\mu}$ tends to $0$ as $\mu \to 0$. 
More precisely \cite[Lemma 1]{HS2}, 
\begin{equation}\label{AT} \lim_{\mu \to 0} \sup_{T \leq C\mu} \frac 1{\mu^{1/4} m_\mu(T)} \|A_{T,\mu}\|_{\rm HS} = 0.\end{equation}
Since this result implies that for small enough $\mu$  the first factor on the right side of (\ref{2f}) 
 does not have a zero eigenvalue, the second does for $T=T_c$.
This means that $$  \frac 1{1 + \frac{1}{1+V^{1/2}p^{-2}|V|^{1/2}} A_{T,\mu} } \frac{m_\mu(T)}{1+V^{1/2}p^{-2}|V|^{1/2}} |V^{1/2}\rangle\langle |V|^{1/2}| $$
has $-1$ as eigenvalue. Since this latter operator is rank one,  its trace has to be $-1$, which leads to the equation 
$$ -\frac{1}{m_\mu(T)}   = \langle |V|^{1/2}|\frac 1{1 + \frac{1}{1+V^{1/2}p^{-2}|V|^{1/2}} A_{T,\mu} }  \frac 1{  1 + V^{1/2} \frac 1{p^2} |V|^{1/2} } |V^{1/2}\rangle .$$

Expanding now 
$$ \frac 1{1 + \frac{1}{1+V^{1/2}p^{-2}|V|^{1/2}} A_{T,\mu} } = 1 - \frac{1}{1+V^{1/2}p^{-2}|V|^{1/2}} A_{T,\mu}  \frac 1{1 + \frac{1}{1+V^{1/2}p^{-2}|V|^{1/2}} A_{T,\mu} } $$
and using the definition (\ref{defa}) for the scattering length,   as well as the asymptotic behavior of $m_\mu(T)$ in \eqref{defm}, 
we obtain the equation
$$ - \frac{\pi}{2a \sqrt{\mu} } = \ln\frac \mu T + \gamma -2 +\ln \frac 8 \pi + o(1)\,.$$
The error terms have been absorbed in $o(1)$  using \eqref{AT}. This implies the statement of Theorem \ref{main}.
\end{proof}

\subsubsection{Zero-range limit} 

Another limit in which the critical temperature can be  calculated explicitly is the limit when the range of the interaction potential goes to zero.
In other words, one considers a sequence of potentials $V$ converging to a contact interaction. Such contact interactions are thoroughly studied in the literature \cite[chap
I.1.2-4]{albeverio} and are known to arise as a one-parameter family of self-adjoint extensions of the Laplacian on $\mathbb{R}^3\setminus\{0\}$. 
The relevant parameter uniquely determining the extension is, in fact, the scattering length, which we assume to be negative, in which case the resulting operator is non-negative, i.e., 
there are no bound states. In other words, we consider a sequence of potentials $V_\ell$ with range $\ell$ going to zero, and  require that the  scattering length $a(V_\ell)$ converges to a negative value as $\ell \to 0$, i.e., 
$$\lim_{\ell \to 0} a(V_{\ell}) = a < 0.$$
For ways to construct such a sequence of potentials, we refer to \cite{albeverio} or \cite{BHS,BHS2}. 

By using similar methods as the ones discussed in this section, it was shown in \cite{BHS2} that for suitable sequences $V_\ell$ 
the corresponding solution to the BCS gap equation $\Delta_\ell$ converges to a constant $\Delta$ as $\ell\to 0$. Moreover, $\Delta$  
satisfies the following BCS gap equation for numbers
\begin{equation}\label{se}
  -\frac{1}{4\pi a} =
  \frac{1}{(2\pi)^3}\int_{\mathbb{R}^3}\left(\frac{\tanh \big( \frac{ \sqrt{(p^2 - \mu)^2 + |\Delta|^2}}{2T} \big)}{\sqrt{(p^2 - \mu)^2 + |\Delta|^2}}
    -\frac{1}{p^2}\right) dp.\,
\end{equation}
 In the  physics literature \cite{Leggett, randeria, NRS}, superfluid  states  are usually characterized via exactly this  equation.
 Due to monotonicity properties, as discussed above, this order parameter $\Delta$ does not vanish below the  critical temperature $T_c$, which is now uniquely defined by 
 $$  -\frac{1}{4\pi a} = \frac{1}{(2\pi)^3}\int_{\mathbb{R}^3}
    \left(
      \frac{\tanh\big(\frac{p^2-\mu}{2T_c}\big)}{p^2-\mu}
      -\frac{1}{p^2} \right) dp
$$
for $a<0$.

\subsubsection{Zero temperature and  energy gap}

In the following we give a short description of the 
 zero temperature case $T=0$. For more details see \cite{HS1,HS2}. In this case, it is
natural to formulate a functional depending only on $\alpha$ instead
of $\gamma$ and $\alpha$. Since for $T=0$ the minimizer of the BCS functional \eqref{freeenergy} is a projection, i.e., $\Gamma^2 = \Gamma$, if follows  that any minimizer for $T=0$ necessarily satisfies the equation $| \hat\alpha(p)|^2 = \hat\gamma(p)(1-\hat\gamma(p))$. It is therefore obvious  that the optimal choice of
$\hat\gamma(p)$ in $\F$ for given $\hat\alpha(p)$ is 
\begin{equation}\label{gal}
  \hat\gamma(p) = \left\{ \begin{array}{ll}
      \half (1+\sqrt{1-4|\hat\alpha(p)|^2}) & {\rm for\ } p^2< \mu \\
      \half (1-\sqrt{1-4|\hat\alpha(p)|^2}) & {\rm for\ } p^2>\mu
    \end{array}\right.\,.
\end{equation}
Subtracting an unimportant constant, this leads to the {\it zero temperature
BCS functional}
\begin{equation}\label{deffa}
  \F_0(\alpha)
  =\frac 12 \int_{\R^3} |p^2-\mu|\left(1-\sqrt{1-4|\hat\al(p)|^2}\right)dp+ \lambda \int_{\R^3}
  V(x)|\alpha(x)|^2\,dx\,.
\end{equation}

The variational equation satisfied by a minimizer of (\ref{deffa}) is then
\begin{equation}\label{bcset}
  \Delta(p) = -\frac \lambda{(2\pi)^{3/2}} \int_{\R^3} \hat V(p-q)
  \frac{\Delta(q)}{E_\Delta(q)} \, dq\,,
\end{equation}
with $\Delta(p) =  2E_\Delta(p) \hat \alpha(p)$. 
This is simply the BCS equation (\ref{eq:gap}) at $T=0$.  For a solution
$\Delta$, the {\it energy gap} $\Xi$ is defined as
\begin{equation}\label{defxi}
\Xi = \inf_p E_\Delta(p) = \inf_p \sqrt{(p^2-\mu)^2 + |\Delta(p)|^2}\,.
\end{equation}
It has the interpretation of an energy gap in the corresponding
second-quantized BCS Hamiltonian (see, e.g., \cite{MR} or the
appendix in \cite{HHSS}).

A priori, the fact that the order parameter $\Delta$ is non-vanishing
does not imply that $\Xi>0$. Strict positivity of $\Xi$ turns out to
be related to the continuity of the corresponding $\hat\gamma$ in
(\ref{gal}). In fact, it was shown in \cite{HS1} that
if $V$ decays fast enough, i.e., $V(x)|x| \in L^{6/5}(\R^3)$, the two
properties, $\Xi > 0$ and $\hat\gamma$ continuous, are equivalent. As proved in \cite{HS1}, both properties hold true under the assumption that $\int V < 0$:

\begin{proposition}[BCS energy gap]\label{psofgap}
  Let $V \in L^{3/2}(\R^3) \cap L^1(\R^3)$, with $ V(x)|x| \in L^{6/5}(\R^3) $ and
  $\int V  < 0$. Let $\alpha$ be a minimizer of the zero-temperature BCS
  functional (\ref{deffa}). Then $\Xi$ defined in (\ref{defxi}) is strictly
  positive, and the corresponding momentum distribution $\hat\gamma$ in
  (\ref{gal}) is continuous.
\end{proposition}

One of the difficulties involved in evaluating $\Xi$ is the potential
non-uniqueness of minimizers of (\ref{deffa}), and hence
non-uniqueness of solutions of the BCS gap equation (\ref{bcset}). The
gap $\Xi$ may depend on the choice of $\Delta$ in this case. For
potentials $V$ with non-positive Fourier transform, however, one can
prove the uniqueness of $\Delta$ and, in addition, one can 
derive the precise asymptotic behavior of $\Xi$ in the weak coupling limit $\lambda\to 0$.

The following is proved in \cite[Theorem 2]{HS}. 

\begin{theorem}[Energy gap at weak coupling]\label{gap}
  Assume that $V\in L^1(\R^3)\cap L^{3/2}(\R^3)$ is radial, with $\hat
  V(p)\leq 0$ and $\hat V(0)<0$, and $\mu>0$. Then there is a unique minimizer (up
  to a constant phase) of the zero-temperature BCS functional (\ref{deffa}). The corresponding energy gap,
$ \Xi = \inf_p \sqrt{ (p^2-\mu)^2 + |\Delta(p)|^2}\,, $ is strictly
positive, and satisfies
\begin{equation}\label{themeq2}
  \lim_{\lambda \to 0} \left(\ln\left(\frac\mu \Xi \right) +
    \frac {1}{ \sqrt{\mu}\, b_\mu(\lambda)}\right) = 2 - \ln(8)\,.
\end{equation}
Here, $b_\mu(\lambda)$ is defined in (\ref{defbm}).
\end{theorem}
The Theorem says that, for small $\lambda$,
$$
\Xi \sim \mu \frac{8}{e^2} e^{1/( \sqrt{\mu} b_\mu(\lambda))} \,.
$$
In particular, in combination with Theorem~\ref{constant}, we obtain
the {\em universal ratio}
$$
\lim_{\lambda\to 0} \frac{ \Xi}{T_c} = \frac \pi{e^\gamma} \approx
1.7639\,.
$$
That is, the ratio of the energy gap $\Xi$ and the critical
temperature $T_c$ tends to a universal constant as $\lambda\to 0$,
independently of $V$ and $\mu$. This property has been observed
before for the original BCS model with rank one interaction
\cite{BCS,MR}.

\section{Derivation of the Ginzburg--Landau functional}\label{sgl}

We now return to the original BCS functional \eqref{def:rBCS} lacking translation invariance and including external fields. We are interested in the regime where the system is \lq\lq almost\rq\rq\ translation-invariant, in the sense that the external fields are weak and slowly varying. In other words, they vary only on very large length scales when compared to microscopic distances, hence we have to study the system in a very large box $\calC$. Assume that $\calC$ has linear size $h^{-1}$, with $h$ a small parameter. We then think of the external fields as functions of the form  $W(hx)$ and $A(hx)$, respectively, i.e., they vary on the scale of the system size. The interaction potential $V$, in contrast, is $h$-independent and varies on the microscopic scale of order one. Mathematically it is somewhat more convenient to rescale all lengths by $h$, and consider the macroscopic scale to be order one while the microscopic length scale is $h$. I.e., we work on a fixed, $h$-independent domain $\calC$ with external fields $W(x)$ and $A(x)$, but the interaction potential is now of the form $V(x/h)$. Moreover, the momentum operator takes the form $-i h \nabla$ in these new variables. This explains the choice of the letter $h$ for our small parameter; it  shows up as an effective Planck constant in the scaling we choose.

The external forces applied to the system thus vary on the size of the
sample, the macroscopic scale, whereas the fermions interact on a much smaller scale, the microscopic scale.
The ratio of these length scales is denoted by  $h$. Additionally we assume the external fields to be weak. The appropriate magnitude turns out to be such that  the fields are of the form $hA(x)$ and  $h^2W(x)$, respectively. They then lead to a relative change in energy of the order $h^2$. Correspondingly, we shall consider a temperature range where $T$ is very close to the critical temperature (for the translation-invariant problem), more precisely $|T-T_c| \sim h^2$. In this regime, the Cooper-pair wavefunction will be very small and the BCS gap equation effectively becomes linear. The external fields then lead to a variation of the pair wavefunction in its center-of-mass coordinate on the macroscopic scale, and this variation turns out to be described by the Ginzburg--Landau (GL) model. The behavior in the relative coordinate turns out not to be effected by the external fields. This is due to the above mentioned fact that we are in a regime where the gap equation effectively becomes linear, hence variations in magnitude do not effect the solution of the equation. 

In order to avoid having to deal with boundary conditions at the boundary of the sample $\calC$, it is more convenient to think of the system as infinite and periodic, and to calculate all energies per unit volume. More precisely, let $\calC = [0,1]^d$ be the unit cube in $\R^d$, and consider states $\Gamma$ that are periodic, i.e., commute with translations by one in any of the $d$ coordinate directions. Assume also that $W$ and $A$ are periodic functions. The relevant BCS energy functional under consideration here is then of the form 
\begin{align}\nonumber
  \mathcal{F}({\Gamma})&= {\rm Tr\,} \left[ \left(
      \left(-ih \nabla + hA(x) \right)^2 -\mu + h^2W(x)\right) \gamma \right] - T\, S(\Gamma)
  \\& \quad  + \iint_{\R^d\times \calC}
  V(h^{-1}(x-y)) |\alpha(x,y)|^2 \, {dx \, dy} \,, \label{def:rBCS2}
\end{align}
where the entropy $S(\Gamma)$  takes the usual form  $S(\Gamma)= - {\rm Tr\,} \Gamma \ln \Gamma$, and where 
 $\Tr$ stands now for
the {\em trace per unit volume}. More precisely, if $B$ is a periodic
operator  then $\Tr B$ equals, by definition the (usual)
trace of $\chi B$, with $\chi$ the characteristic function of
$\calC$. Note that due to the infinite size of the system, one of the integrations in the last interaction term in (\ref{def:rBCS2}) now extends over all of $\R^d$. In the remainder of this section, we will study the BCS functional (\ref{def:rBCS2}). We will follow closely \cite{FHSS,FHSS2,FHSS3}. For simplicity,
  we restrict our attention to the case $d=3$ here, but a similar
  analysis applies in one and two dimensions as well.

\subsection{The Ginzburg--Landau model}\label{sec:gl}

The Ginzburg--Landau model \cite{GL} is a phenomenological model for
superconducting materials. The model is macroscopic in nature. There are no single particles involved. 
The state of superconductivity is described by a macroscopic wave function, also called order parameter, describing the collective motion of  
the particles like in the case of a BEC or a laser. 
Imagine a sample of such a material
occupying a three-dimensional box $\calC$. If $W$ denotes
 a scalar external potential, and $A$ a magnetic vector
potential, the Ginzburg--Landau functional is given as
\begin{align}\nonumber
\mathcal{E}^{\rm GL}(\psi) & = \int_{{\calC}} \biggl(
|(-i\nabla + 2A(x))\psi(x)|^2
   + \lambda_1 {W}(x) |\psi(x)|^2 \\& \qquad\quad  -\lambda_2 D |\psi(x)|^2 + \lambda_3 |\psi(x)|^4  \biggl) dx\,, \label{def:GL}
\end{align}
where $\psi\in H^1(\calC)$. In fact, in the case of an infinite, periodic system considered here, $\psi$ will have to be restricted to $H_{\rm per}^1(\calC)$, the periodic functions that are locally in $H^1$. Note that there is no normalization condition on $\psi$. It is, for instance, allowed to be identically zero, corresponding to the normal state of vanishing superconductivity.

Note the factor $2$ in front of the
$A$-field, which is reminiscent of the fact that $\psi$
describes pairs of particles. The microscopic functional \eqref{def:rBCS2} does not have such a factor $2$, as it describes single
particles. 
The form of the  Ginzburg--Landau functional is chosen in the precise way as it comes out in the limit $h\to 0$, close to the critical temperature $T_c$, from the 
BCS functional \eqref{def:rBCS2}. The $\lambda_i$ in \eqref{def:GL} are real
parameters, with $\lambda_2, \lambda_3>0$.  
 They  only depend on the microscopic parameters 
of the system, namely on $V$ and on the chemical potential $\mu$.
The parameter $D$, however, is related to the difference of $T$ to the critical temperature $T_c$ (for the translation-invariant problem), 
in the form $$ T = T_c(1 - D h^2).$$ 
By simple rescaling of $\psi$, we
could take $\lambda_3=2 |D\lambda_2|$ if $D \neq 0$. If $D>0$, one could
then complete the square in the second line of (\ref{def:GL}) and
write it as $\lambda_3 ( 1 - |\psi(x)|^2 )^2$ instead\footnote{This
  is, in fact, the convention used in Ref.~\cite{FHSS}.}, a formulation frequently found in the literature.  Since $D$
can have either sign, however, we prefer to use the more general
formulation in (\ref{def:GL}) here. 
The function $\psi$ is interpreted as the order parameter of the
system. In the absence of external fields, i.e., for $W=0$ and ${
  A}={ 0}$, the minimum of (\ref{def:GL}) is attained at
$|\psi(x)|^2= D\lambda_2/(2\lambda_3)$ for $D>0$, or at $\psi\equiv 0$ for
$D\leq 0$, respectively.
This is an easy consequence of the fact that the graph of the function  $ t \mapsto -\lambda_2 D t^2 + \lambda_3 t^4$ has the shape of a 
Mexican hat, or double well, potential if $D > 0$, with its minimum at $t^2 = D\lambda_2/(2\lambda_3)$, while it is 
a simple single well with global minimum at $t=0$ if $D \leq 0$.

Since the quartic term in the GL functional (\ref{def:GL}) is positive, and $\psi=0$ is a critical point of the functional, we conclude that the 
 GL functional $\mathcal{E}^{\rm GL}(\psi)$ is minimized by a non-trivial $\psi$, i.e., a $\psi \not\equiv 0$, 
if and only  if the Hessian of  $\mathcal{E}^{\rm GL}(\psi)$ at $\psi =0$  has a negative eigenvalue, i.e., if 
$$ \frac 12  \frac{d^2}{ d t^2}\mathcal{E}^{\rm GL}(t \phi) \Big|_{t =0} = \langle \phi |  (-i\nabla + 2A(x))^2
   + \lambda_1 {W}(x) -\lambda_2 D| \phi\rangle $$
can be made negative for a $\phi \in H^1_{\rm per}(\calC)$. 
This leads directly to the definition of the critical parameter $D_c$ given by
\begin{equation}\label{critparD}
 D_c = \frac 1 {\lambda_2} \infspec \left((-i\nabla + 2A(x))^2
   + \lambda_1 {W}(x)\right),
\end{equation}
where the operator on the right side is understood as an operator on $L^2(\calC)$ with periodic boundary conditions.
Observe the analogy to the parameter $T_c$ which was determined by the second derivative of the translation-invariant BCS functional in the previous section. 
As explained below, the critical GL parameter $D_c$ actually turns out to correspond to the shift of order $h^2$ in the critical temperature of the BCS functional (\ref{def:rBCS2}) due to the introduction of the external fields.

We note that there is a considerable literature \cite{FH,GST,Serfaty} concerning functionals of the 
type (\ref{def:GL}) and their minimizers, usually with an additional term added corresponding to the 
magnetic field energy. Such a term plays no role here since $A$ is considered a fixed, external field. 

\subsection{Main results} We shall now explain our main results concerning the connection between the BCS functional (\ref{def:rBCS2}) and the Ginzburg--Landau model (\ref{def:GL}). They are valid under the following assumptions on the potentials entering the definition of (\ref{def:rBCS2}). 

\begin{assumption}\label{as0}
  We assume both $W$ and $A$ to be periodic with period $1$. We
  further assume that $\hat W(p)$ and $|\hat A(p)|(1+|p|)$ are
  summable, with $\hat W(p)$ and $\hat A(p)$ denoting the
  Fourier coefficients of $W$ and $A$, respectively. In particular,
  $W$ is bounded and continuous and $A$ is in $C^1(\R^3)$.
\end{assumption}

\begin{assumption}\label{as1}
  The potential $V$ is radial and is such that $T_c > 0$ for the translation-invariant problem,  and that $K_{T_c} + V$ has
  a non-degenerate ground state eigenvalue $0$, whose corresponding eigenfunction will be denoted by $\alpha_0$. 
\end{assumption}

The operator $K_{T}$ was introduced in the previous section. Sufficient conditions for  Assumption \ref{as1} to be satisfied are discussed there. We note that it is not necessary to assume that $V$ is radial; we do this here only for simplicity so that some of the formulas below simplify.

As before, we define the normal state $\Gamma_0$ as the minimizer of the BCS functional in the absence of interactions, i.e., when $V=0$. That is, 
\begin{equation}\label{def:gamma0}
  \Gamma_0 := \left( \begin{matrix} \gamma_0 & 0 \\ 0 & 1 -\bar\gamma_0 \end{matrix}\right)  = \frac 1{1+e^{H_0/T}}
\end{equation}
where 
$$\gamma_0 = \frac 1{1+e^{((-ih\nabla+hA(x))^2 + h^2 W(x) -
  \mu)/T}}$$ and \begin{equation}\label{def-h0}
  H_0 = \left( \begin{matrix} \left(-i h \nabla + h A(x) \right)^2 -\mu + h^2 W(x)  & 0 \\ 0 &- \left( \left(i h \nabla + h A(x) \right)^2 -\mu + h^2 W(x)  \right) \end{matrix}\right) .
  \end{equation}
We have
\begin{equation}\label{f0}
  \F(\Gamma_0) = - T\, \Tr \ln\left( 1+ \exp\left(-\left( \left(-i h \nabla + h A(x) \right)^2 -\mu + h^2 W(x)  \right) /T \right) \right) \,,
\end{equation}
which is $O(h^{-3})$ for small $h$. 

We consider temperatures $T$ close to $T_c$, of the form 
\begin{equation}\label{def:D}
T=T_c(1-Dh^2)
\end{equation}
for some $h$-independent parameter $D\in \mathbb{R}$.\footnote{The
  results in Ref.~\cite{FHSS} were stated for $D>0$, but the
  proof is equally valid for $D\leq 0$.}
 The following theorem, proved in \cite{FHSS},  shows the emergence of the GL model (\ref{def:GL}) in  
the asymptotic behavior of $F(T,\mu)= \inf_\Gamma \mathcal{F}(\Gamma)$, and the corresponding minimizers, as $h\to 0$.

\begin{theorem}[Derivation of the GL model]\label{main:thm}
There exists a $\lambda_0>0$ and parameters $\lambda_1$, $\lambda_2$ and $\lambda_3$ in the GL functional (\ref{def:GL})
such that 
\begin{equation}\label{equthm}
\inf_{\Gamma}\mathcal{F}(\Gamma)=\mathcal{F}(\Gamma_0) + \lambda_0\, h \inf_\psi\mathcal{E}^{\rm GL}(\psi) +o(h) 
\end{equation}
as $h\to0$, where $\Gamma_0$ is the normal state \eqref{def:gamma0}. 

Moreover, if $\Gamma$ is a state such that $\mathcal{F}(\Gamma)\leq \mathcal{F}(\Gamma_0) + \lambda_0 h  \inf_\psi\mathcal{E}^{\rm GL}(\psi) + o(h)$, then there exists a $\psi_0 \in H^1_{\rm per}(\calC)$ with $\mathcal{E}^{\rm GL}(\psi_0) \leq \inf_\psi\mathcal{E}(\psi)+ o(1)$ such that 
 the corresponding Cooper-pair wavefunction $\alpha$ satisfies
\begin{equation}
\|\alpha-\alpha_{\rm GL}\|^2_{L^2}\leq o(1)\|\alpha_{\rm GL}\|^2_{L^2}=o(1)h^{-1}
\end{equation}
where
\begin{equation}\label{def:agl}
\alpha_{\rm GL}(x,y) = \frac{1}{2 h^2} (\psi_0(x) + \psi_0(y))
   \frac 1{(2\pi)^{3/2} } \alpha_0\left( h^{-1} (x-y)\right)  
 \end{equation}
 and the norm is understood in $L^2(\R^3 \times \calC)$. 
\end{theorem}

\begin{remark}
The parameter $\lambda_0$ in Eq.~\eqref{equthm} could of course be absorbed into the definition of $\mathcal{ E}^{\rm GL}$, and this was the convention used in \cite{FHSS}. Here we chose to work with the standard definition (\ref{def:GL}) having a factor $1$ in front of the kinetic energy term, which results in the factor $\lambda_0$ showing up in \eqref{equthm}.
\end{remark}

\begin{remark}
Interpreting $\alpha_{\rm GL}(x,y)$ in (\ref{def:agl}) as the integral kernel of a corresponding (periodic) operator $\alpha_{\rm GL}$ on $L^2(\R^3)$, one can write the definition in (\ref{def:agl}) as 
\begin{equation}\label{def:aglop}
\alpha_{\rm GL} = \frac h 2 \big( \psi_0(x) \hat \alpha_0(-ih\nabla) + \hat \alpha_0(-ih\nabla) \psi_0(x)\big)\,. \end{equation}
That is, $\alpha_{\rm GL}$ is the symmetrized product of the periodic function $\psi_0$, acting as a multiplication operator in configuration space, and  $\hat \alpha_0$ acting a multiplication operator in momentum space. Hence one can think of it as the quantization of the semiclassical symbol $h\psi_0(x)\hat\alpha_0(p)$. 
\end{remark}

\begin{remark}
To leading order in $h$, the right side of (\ref{def:agl}) could as well be  replaced  by the expression 
\begin{equation}\label{alsy}
\frac 1 {(2\pi)^{3/2} h^2}\psi_0\left(\frac{x+y}{2}\right)
    \alpha_0\left(\frac{x-y}{h}\right) \,,
\end{equation}  
corresponding to the Weyl quantization of $h\psi_0(x)\hat\alpha_0(p)$. 
In this way Theorem~\ref{main:thm} demonstrates the role of the
function $\psi$ in the GL model: It describes the center-of-mass motion of the
Cooper-pair wavefunction, which close to the critical temperature
equals (\ref{alsy}) to leading order in $h$.
\end{remark}

\begin{remark}
  A similar analysis as the one leading to the proof of Theorem~\ref{main:thm} can be used at $T = 0$ to study the low-density
  limit of the BCS model. In this limit, one obtains a Bose--Einstein
  condensate of fermion pairs, described by the Gross--Pitaevskii
  equation \cite{HS3,HSch}.
\end{remark}

\subsubsection{The coefficients $\lambda_i$}

The coefficients $\lambda_0$, $\lambda_1$, $\lambda_2$ and $\lambda_3$
in Theorem~\ref{main:thm} can be explicitly calculated. They are all
expressed in terms of   the critical temperature $T_c$  and the eigenfunction $\alpha_0$ corresponding to
the zero eigenvalue of the operator $K_{T_c} + V$.

Specifically, if we denote by $t$ the Fourier transform of $2 K_{T_c}\alpha_0$, we have
\begin{equation}
\lambda_0  =  \frac{1}{16 T_c^2 }   \int_{\mathbb{R}^3} t(q)^2 \left(  g_1(\beta_c(q^2-\mu)) + \frac 2 3 \beta_c q^2 \, g_2(\beta_c(q^2-\mu)) \right) \frac{dq}{(2\pi)^3}\,,
\end{equation}
\begin{equation}
\lambda_1 = \lambda_0^{-1}   \frac  {1}{4 T_c^2}  \int_{\mathbb{R}^3} t(q)^2 \, g_1(\beta_c(q^2-\mu)) \, \frac{dq}{(2\pi)^3}  \,,
\end{equation}
\begin{equation}
\lambda_2 = \lambda_0^{-1} \frac{1}{8 T_c} \int_{\mathbb{R}^3}  {t(q)^2} \cosh^{-2} \left( \frac {\beta_c}{2}(q^2 -\mu) \right) \frac{dq}{(2\pi)^3}
\end{equation}
and
\begin{equation}\label{def:b3}
 \lambda_3 =   \lambda_0^{-1} \frac {1} {16 T_c^2}  \int_{\mathbb{R}^3} t(q)^4 \, \frac{g_1(\beta_c(q^2-\mu))}{q^2-\mu}\,\frac{dq}{(2\pi)^3} \,.
\end{equation}
Here $\beta_c= T_c^{-1}$, and $g_1$ and $g_2$ denote the functions
\begin{equation}
g_1(z) = \frac{ e^{2 z} - 2 z e^{z}-1}{z^2 (1+e^{z})^2} \quad \text{and} \quad 
g_2(z) = \frac{2 e^{z} \left( e^{ z}-1\right)}{z
  \left(e^{z}+1\right)^3} \,,
\end{equation}
respectively. 
One can show \cite{FHSS} that $\lambda_0>0$. Note that $g_1(z)/z > 0$,
hence also $\lambda_3>0$. The coefficient $\lambda_2$ is clearly positive.
As mentioned in Section~\ref{sec:gl}, the terms in the second line of
the GL functional (\ref{def:GL}) are often written as $\frac{\kappa^2}{2}
(1-|\psi(x)|^2)^2$ instead, with a suitable coupling constant
$\kappa>0$. In our notation, $\kappa$ corresponds to 
\begin{equation}
\kappa = \sqrt{\lambda_2 D}
\end{equation} (in case $D > 0$, i.e., $T<T_c$).
\begin{remark}
Note that the normalization of $\alpha_0$ is irrelevant. If we
multiply $\alpha_0$ by a factor $\lambda>0$, then $\lambda_0$ and
$\lambda_3$ get multiplied by $\lambda^2$, while $\lambda_1$ and
$\lambda_2$ stay the same. Hence the GL minimizer $\psi_0$ gets
multiplied by $\lambda^{-1}$, leaving both $\lambda_0 \mathcal{E}^{\rm
  GL}(\psi_0)$ and the product $\psi_0 \alpha_0$ unchanged. In particular,
\eqref{equthm} and (\ref{def:agl}) are independent of the normalization of $\alpha_0$.
\end{remark}

The methods developed to prove Theorem~\ref{main:thm} can also be used to 
 obtain a more precise estimate for the  critical temperature in the BCS model \eqref{def:rBCS2}
in the presence the external fields, which we shall denote as  $T_c^{\rm BCS}$. 
In fact, the following Theorem was proved in \cite{FHSS3}.

\begin{theorem}[External field dependence of critical temperature]\label{thm:ct}
Under assumptions \eqref{as0} and \eqref{as1} and with  $D_c$ defined in  \eqref{critparD}, we have, for small $h$, 
\begin{itemize}
\item[(a)] the minimum of $\F(\Gamma)$ is attained by a state $\Gamma$ with non-vanishing $\alpha$, i.e., the system is in a superconducting state, 
if 
\begin{equation}
T \leq  T_c\left( 1 - h^2(D_c- O(h^{1/2}))\right);
\end{equation}
\item[(b)] the minimum of $\F(\Gamma)$ is attained  by the normal state $\Gamma_0$ if
\begin{equation}
T \geq  T_c\left( 1 - h^2(D_c + O(h^{2/5}))\right)\,.
\end{equation}
\end{itemize}
\end{theorem}

Strictly speaking, this theorem does not prove the existence of a critical temperature $T_c^{\rm BCS}$ such that the normal state $\Gamma_0$ minimizes $\mathcal{F}$ if and only if $T \geq T_c^{\rm BCS}$. However in case it does exist then Theorem \ref{thm:ct} 
implies that $$ T_c^{\rm BCS} = T_c\left( 1 - h^2D_c\right) + o(h^2)$$
as $h\to 0$.

\begin{remark}
Note that the critical parameter $D_c$ in the Ginzburg--Landau model does not necessarily have a sign. The presence of an electric potential $W$ can cause $D_c$ to be negative, which means that the critical temperature can go up as a consequence of introducing external fields.
\end{remark}

\subsection{Sketch of the proof of Theorem \ref{main:thm} }

The strategy of  the proof of Theorem~\ref{main:thm} in \cite{FHSS} can be summarized in three steps as follows. 
To simplify things, we sketch the proof only in the case $A = 0$. 

\medskip

{\em Step 1}:  Since we are only concerned with low energy states  we restrict our attention to states $\Gamma$ with lower energy than the normal state $\Gamma_0$, i.e.,
$$ \F (\Gamma) \leq \F(\Gamma_0).$$
As a first step we argue that 
 for temperatures $T$ close to the critical temperature $T_c$, the corresponding pairing term $\alpha$ necessarily has a specific form, corresponding to  the translation-invariant minimizer as far as the behavior in the relative variable (on the microscopic scale) is concerned. 
More precisely, $\alpha$ has to be of the following form.

\begin{proposition}[Structure of the pairing term] \label{proppair}
If $|T - T_c| = O(h^2)$ then for any state $\Gamma$ with $ \F (\Gamma) \leq \F(\Gamma_0) + O(h)$, the corresponding $\alpha$ satisfies
\begin{equation}\label{formalpha}
  \alpha = \tfrac h 2 \big( \psi(x) \hat \alpha_0(-ih\nabla) + \hat \alpha_0(-ih\nabla) \psi(x)\big)  + \xi
\end{equation}
for some periodic function $\psi$ with $H^1(\calC)$ norm bounded
independently of $h$, and with $$\|\xi\|_{H^1} \leq O(h) \|\alpha\|_{H^1} \leq O(h^{1/2}).$$
\end{proposition}
Here we use the notation 
\begin{equation}\label{def:h1}
  \|O\|_{H^1}^2 = \Tr \left[ O^\dagger \left(1-h^2\nabla^2\right) O  \right]\,
\end{equation}
as definition for the $H^1$-norm of a periodic operator $O$.
In other words, $$\|O\|_{H^1}^2 = \|O\|_2^2 + h^2 \|\nabla
O\|_2^2 \,,$$ where $\|O\|_2 = (\Tr O^\dagger O )^{1/2}$. Note that this definition of the $H^1$-norm is not symmetric, i.e.,
$\|O\|_{H^1} \neq \|O^\dagger\|_{H^1}$ in general.
The $H^1(\calC)$-norm of $\psi$ is the usual one
$$ \|\psi\|_{H^1(\calC)}^2= \int_{\calC} |\nabla \psi(x)|^2 dx + \int_{\calC}|\psi(x)|^2 dx.$$

Recall that the function $\alpha_0$ in \eqref{formalpha} is the eigenfunction to the eigenvalue $0$ of the operator
$K_{T_c} + V$. Modulo a normalization factor, this is the same to  leading order in $h$ as taking the actual minimizer of the translation-invariant problem \eqref{freeenergy} for $T<T_c$, $|T-T_c|\leq O(h^2)$. 
Proposition~\ref{proppair}  shows that the Cooper-pair wavefunction $\alpha$ of any approximate minimizer of the BCS functional close to the critical temperature 
necessarily equals the translation-invariant minimizer, which varies on the microscopic scale, times a function which
accommodates for the spatial variations coming from the external fields $A$ and $W$. These variations take place on the macroscopic scale. 

\medskip

{\em Step 2}: Associated with the above state $\Gamma$ we now construct a new state $\Gamma_\Delta$ via the decomposition \eqref{formalpha}, where we will omit the part $\xi$ which is of higher order for small $h$. 
To this aim we define the periodic operator
\begin{equation}\label{def:del}
\Delta = -\frac h2 (\psi_<(x)
t(-ih\nabla) + t(-ih\nabla) \psi_<(x)),
\end{equation}
where, for technical reasons, we cut the high frequencies from the function $\psi$, i.e.,
$$
 \hat
  \psi_<(p) = \hat \psi(p) \theta(\epsilon h^{-1} - |p|)  $$
  for suitable (small) $\epsilon>0$, 
and with
\begin{equation}\label{deft}
  t(p) = - 2 (2\pi)^{-3/2} \int_{\R^3} V(x) \alpha_0(x) e^{-ip\cdot x} dx \,.
\end{equation}
The cutting of the high frequencies allows estimates of $\psi_<$ in the $H^2(\calC)$-norm which is necessary for 
the semiclassical calculations we will have to perform. 
In terms of operator kernels we can write
\begin{equation}\label{int:delta}
  \Delta(x,y) = \frac{h^{-2}}{(2\pi)^{3/2}} \left( \psi_<(x) + \psi_<(y) \right) V(h^{-1}(x-y)) \alpha_0(h^{-1}(x-y)) \,.
\end{equation}
The new state $\Gamma_\Delta$ is now defined as
\begin{equation}\label{def:gammad2}
  \Gamma_\Delta = \left( \begin{matrix} \gamma_\Delta & \alpha_\Delta  \\ \bar\alpha_\Delta & 1-\bar\gamma_\Delta \end{matrix} \right) = \frac 1{1 + e^{\beta H_\Delta}}
\end{equation}
where
\begin{equation}\label{defk}
H_\Delta=\left(\begin{matrix}k&\Delta\\
\overline{\Delta}&-\overline{k}
\end{matrix}\right)\ ,\quad 
k=-h^2 \nabla^2  -\mu +
    h^2 W(x)\,.
\end{equation}
Step 2 consists of showing that 
\begin{equation}\label{dif}
 \F(\Gamma) - \F(\Gamma_\Delta) \geq - o(h).
 \end{equation}
 That is, to the accuracy we are interested in, we might as well work with $\Gamma_\Delta$ instead of the original $\Gamma$. 
 
 \medskip

{\em Step 3}: This step, which is actually the first step in the proof in \cite{FHSS},
consists of showing that at the temperature $T = T_c(1 - D h^2)$ 
\begin{equation}\label{sca}
\F(\Gamma_\Delta) - \F(\Gamma_0) =  h\,\lambda_0  \mathcal{E}^{\rm GL}(\psi_<) +o(h) \,,
\end{equation}
where the error term depends on the $H^1$ and $H^2$-norms  of $\psi_<$. The calculation leading to (\ref{sca}) is a type of semiclassical analysis, with a semiclassical parameter $h$ playing the role of an effective Planck constant. 
In combination with \eqref{dif} this gives the desired result \eqref{equthm} for $\F(\Gamma)$. 

\subsection{Useful identity and relative entropy inequality} 

In the following we will deal with operator-valued $2\times 2$ matrices that are not necessarily locally trace-class due to their off-diagonal terms. The diagonal terms are locally trace-class, however, and for this reason we will 
have to introduce a slightly weaker notion of trace, via the trace of the diagonal blocks, i.e.,
\begin{equation}\label{deftrs}
\Trs O = \Tr \left[ P_0 O P_0 + Q_0 O Q_0 \right]
\end{equation}
with 
\begin{equation}\label{defp0}
P_0 =  \left( \begin{matrix} 1 & 0 \\ 0 & 0 \end{matrix}\right)
\end{equation}
and $Q_0 = 1 - P_0$. Note that if $O$ is locally trace class, then $\Trs O = \Tr O$. This identity also holds for all non-negative operators $O$, in the sense that either both sides are infinite or otherwise equal. 
The trace $\Trs$ will be useful for a convenient identity explained in Lemma~\ref{lem:id} below.

Take  any function $\psi \in H^1_{\rm per}(\calC)$ and define 
$\alpha_{\rm GL}$ as in \eqref{def:agl} i.e.,
$$ \alpha_{\rm GL}=\tfrac h 2 \big( \psi(x) \hat \alpha_0(-ih\nabla) + \hat \alpha_0(-ih\nabla) \psi(x)\big). $$
Construct as above the corresponding operator $\Delta$ via its integral kernel
\begin{equation}\label{Delal}
\Delta(x,y)=2V(h^{-1}(x-y))\alpha_{\rm
  GL}(x,y) \,.
\end{equation}
Then we obtain the following very useful identity.

\begin{lemma}\label{lem:id}
\begin{align}\nonumber
\mathcal{F}(\Gamma) - \mathcal{F}(\Gamma_0) & = - \frac T2 \Trs \left[ \ln\left(1+e^{-H_\Delta/T}\right) - \ln\left(1+e^{-H_0/T} \right) \right] \\ \nonumber & \quad + \frac T 2 \mathcal{H}_0(\Gamma,\Gamma_\Delta) - \iint_{\R^3 \times \calC } V(h^{-1}(x-y))|\alpha_{\rm GL}(x,y)|^2 \, dx \,dy  \\ & \quad + \iint_{\R^3 \times \calC } V (h^{-1}(x - y))\left|\alpha_{\rm GL}(x, y) - \alpha(x, y)\right|^2 \, dx\, dy  \label{fi}
\end{align}
where $\Gamma_\Delta = (1+e^{H_\Delta/T})^{-1}$,  $H_\Delta$ is given as in \eqref{defk}, and 
 $\mathcal{H}_0(\Gamma,\Gamma_\Delta)$ denotes the relative entropy
\begin{equation}\label{defrelent}
\mathcal{H}(\Gamma,\Gamma_\Delta) = \Trs \left[ \Gamma\left( \ln \Gamma - \ln \Gamma_\Delta \right) + \left(1-\Gamma\right) \left( \ln \left(1-\Gamma\right) - \ln\left(1-\Gamma_\Delta\right) \right) \right]\,.
\end{equation}
\end{lemma}

\begin{proof}
Note that $H_\Delta$ is unitarily equivalent to $-\bar H_{\Delta}$,
\begin{equation}\label{eq:unit}
  U H_{\Delta} U^\dagger = - \bar H_\Delta \quad \text{with}\quad  U = \left( \begin{matrix} 0 & 1 \\ -1 & 0 \end{matrix} \right)\,.
\end{equation}
Hence also $U\Gamma_\Delta U^\dagger = 1 - \bar \Gamma_\Delta$ and, in
particular,
\begin{equation}\label{ent:ref}
  S(\Gamma_\Delta) = - \tfrac 12 \Tr \left[ \Gamma_\Delta \ln \Gamma_\Delta  +  (1-\Gamma_\Delta) \ln (1-\Gamma_\Delta)\right]\,.
\end{equation}
Here, $\Tr$ could as well be replaced by $\Trs$, the sum of
the traces per unit volume of the diagonal entries of a $2\times 2$
matrix-valued operator defined in (\ref{deftrs}), since the operator in question is negative.
A simple calculation shows that
\begin{align}\nonumber
  & \Gamma_\Delta \ln \Gamma_\Delta + (1-\Gamma_\Delta) \ln
  (1-\Gamma_\Delta) - \Gamma_0 \ln \Gamma_0 - (1-\Gamma_0) \ln
  (1-\Gamma_0) \\ & = - \beta H_\Delta \Gamma_\Delta + \beta H_0
  \Gamma_0 - \ln \left( 1+e^{-\beta H_\Delta}\right) + \ln \left(
    1+e^{-\beta H_0}\right) \,. \label{tci1}
\end{align}
By using this identity we infer that 
\begin{align}\nonumber
& \F(\Gamma) - \F(\Gamma_0) \\  \nonumber & = \frac 12 \Trs\left[H_0 \Gamma - H_0 \Gamma_0  \right] -  TS(\Gamma) + TS(\Gamma_0)  + \iint_{\R^3 \times \calC } V\left(\tfrac{x-y}h\right) |\alpha(x,y)|^2 \,dx \, dy \\ \nonumber
& =  \frac 12 \Trs\left[H_\Delta \Gamma - H_0 \Gamma_0  \right] - TS(\Gamma) + TS(\Gamma_\Delta) \\ \nonumber & \quad   - \Re \Tr \Delta \bar \alpha + \iint_{\R^3\times \calC } V\left(\tfrac{x-y}h\right) |\alpha(x,y)|^2 \, dx\, dy \\ \nonumber & = - \frac T2 \Trs \left[ \ln\left(1+e^{-H_\Delta/T}\right) - \ln\left(1+e^{-H_0/T} \right) \right]  \\ \nonumber & \quad + \frac 12 \Trs H_\Delta (\Gamma-\Gamma_\Delta) - T S(\Gamma) + T S(\Gamma_\Delta) \\ & \quad - \Re \Tr \Delta \bar \alpha + \iint_{\R^3\times\calC} V\left(\tfrac{x-y}h\right) |\alpha(x,y)|^2 \, dx\, dy  \end{align}
where we additionally used that
\begin{equation}\label{tci2}
  H_\Delta \Gamma - H_0 \Gamma = \left( \begin{matrix} \Delta \bar\alpha &   \Delta ( 1- \bar\gamma)  \\ \bar\Delta \gamma  &   \bar\Delta \alpha \end{matrix} \right) 
\end{equation}
in the first step.
Recalling the definition of $\Delta$ we can complete the square and write 
\begin{align}\nonumber
& - \Re \Tr \Delta \bar \alpha + \iint_{\R^3 \times \calC} V\left(\tfrac{x-y}h\right) |\alpha(x,y)|^2 \,dx\, dy \\  \nonumber  & = \iint_{\R^3 \times \calC} V\left(\tfrac{x-y}h\right) |\alpha(x,y) - \alpha_{\rm GL}(x,y)|^2 \,dx\, dy \\ & \quad -  \iint_{\R^3 \times \calC} V\left(\tfrac{x-y}h\right) |\alpha_{\rm GL}(x,y)|^2 \,dx\, dy.
 \end{align} 
Finally we use that 
\begin{multline} 
\frac{\beta}2 \Trs H_\Delta ( \Gamma - \Gamma_\Delta)  -  S(\Gamma) + S(\Gamma_\Delta)  \\ = 
\frac 12 \Trs \left[ \Gamma\left( \ln \Gamma - \ln \Gamma_\Delta \right) + \left(1-\Gamma\right) \left( \ln \left(1-\Gamma\right) - \ln\left(1-\Gamma_\Delta\right) \right) \right] ,
 \end{multline}
the last term being the relative entropy   ${\mathcal H}_0(\Gamma,\Gamma_\Delta)$, see \eqref{defrelent}. In order to see the first identity simply use the definition
 $ \Gamma_\Delta = [1 + e^{\beta H_\Delta}]^{-1} $ and the obvious fact that  $$ \beta H_\Delta = \ln(1-\Gamma_\Delta) - \ln\Gamma_\Delta.$$
 This implies the statement of the lemma. 
\end{proof}

One of the main ingredients in the proof in \cite{FHSS} is the following 
lower bound on the relative entropy, whose proof in \cite[Lemma~1]{FHSS} is done by applying Klein's inequality,
which basically says that one can treat the corresponding operators as if they were numbers. In fact, consider a non-negative function of two variables, $f(x,y)$, that can be decomposed as  $f(x,y) =\sum_i f_i(x) g_i(y) \geq 0$. Then, for any two self-adjoint operators $A$ and $B$, $\tr f(A,B) \geq 0$, where the latter has to be understood as 
$$
\tr f(A,B) = \sum_i  \tr f_i(A) g_i(B) \geq 0 \,.
$$
To see this let $|a_m\rangle, |b_n\rangle$ be the eigenvectors corresponding to the eigenvalues $\lambda_m$ and $\mu_n$ of $A$ and $B$, respectively. 
  Then
  \begin{align}\nonumber
  \tr f(A,B) & = \sum_n \sum_m \sum_i \langle a_m | f_i (A) g_i(B) | a_m\rangle \\ \nonumber &= \sum_n \sum_m \sum_i \langle a_m | f_i (A) |b_n \rangle \langle b_n | g_i(B) | a_m\rangle 
 \\ \nonumber & = \sum_n \sum_m \sum_i \langle a_m | f_i (\lambda_m ) |b_n \rangle \langle b_n | g_i(\mu_n) | a_m\rangle \\ & =
  \sum_n \sum_m |\langle a_m |b_n\rangle|^2  f(\lambda_m ,\mu_n) \geq 0. 
  \end{align}
The extension of this inequality to our setting, where we consider the trace per unit volume of periodic operators, is straightforward \cite[Sect.~3]{FHSS}.

\begin{lemma}\label{lem:klein}
  For any $0\leq \Gamma\leq 1$ and any $\Gamma'$ of the form
  $\Gamma' = (1+e^{H})^{-1}$ commuting with $P_0$ in (\ref{defp0}),
  \begin{equation}\label{eq:lem:klein}
    \H_0(\Gamma,\Gamma') \geq  \Trs\left[ \frac {H}{\tanh (H/2)} \left( \Gamma - \Gamma'\right)^2\right]  + \frac 43 \Tr\left[ \Gamma(1-\Gamma) - \Gamma'(1-\Gamma')\right]^2\,.
  \end{equation}
\end{lemma}

  \begin{proof}
  It is elementary (but tedious) to show that for real numbers
  $0<x,y<1$,
  \begin{equation}
    x \ln\frac xy + (1-x) \ln \frac{1-x}{1-y} \geq \frac { \ln \frac{1-y}{y} } {1-2y} (x-y)^2 + \frac 43  \left( x(1-x) - y(1-y) \right)^2\,.
  \end{equation}
  The result then follows from Klein's inequality.
  \end{proof}

A similar bound as in (\ref{eq:lem:klein}), but without the last positive
term, was  derived earlier in \cite{HLS}. 

\subsection{Proof of the key steps}

\subsubsection{Step $1$}

If we set $\psi = 0$ in \eqref{Delal} and apply  Lemma \ref{lem:id} 
we  obtain the identity
\begin{equation}
  \F(\Gamma) - \F(\Gamma_0) = \tfrac 12 T\, \H_0(\Gamma,\Gamma_0) + \iint_{\R^3 \times \calC} V(h^{-1}(x-y)) |\alpha(x,y)|^2\, {dx \, dy}.
\end{equation}
Recall the definition of $H_0$ in \eqref{def-h0}. It is a diagonal operator-valued $2\times 2$ matrix, with diagonal entries  given by $k$ and $- \bar k$, see \eqref{defk}. Hence also $$\frac {H_0}{\tanh (\frac \beta 2 H_0)} $$ is diagonal,  and its diagonal entries of 
 are given by the operator $K_{T,W}$, which is equal to 
\begin{equation}\label{def:ktaw}
  K_{T,W} =  \frac {-h^2 \nabla^2 -\mu + h^2 W(x)}{ \tanh\left( \tfrac \beta 2 \left(-h^2 \nabla^2 -\mu + h^2 W(x)\right) \right)}\,.
\end{equation}
In particular, we have that  
\begin{equation}\label{kleincon}
  \Trs\left[ \frac {H_0}{\tanh \big(\tfrac \beta 2 H_0\big)} \left( \Gamma - \Gamma_0\right)^2 \right] 
  = 2\, \Tr \left[ K_{T,W} (\gamma-\gamma_0)^2\right] + 2 \, \Tr \left[ K_{T,W}  \alpha \bar \alpha\right] \,.
\end{equation}
The last term can be conveniently rewritten as 
\begin{equation}
\Tr \left[ K_{T,W}  \alpha \bar \alpha\right] = \int_{\calC} \langle \alpha(\,\cdot\,,y) | K_{T,W}  | \alpha(\,\cdot\,,y)\rangle \,  {dy} \,,
\end{equation}
where the inner product is in $L^2(\R^3)$ and we think of $K_{T,W}$ as acting only on the first variable of $\alpha$, at fixed second variable $y$.

Applying Lemma \ref{lem:klein} with $H = \beta H_0$, we thus obtain the lower bound
\begin{multline}\label{apribound1}
\F(\Gamma) - \F(\Gamma_0) \geq  \int_{\calC} \langle \alpha(\,\cdot\,,y) | K_{T,W} + V(h^{-1}(\, \cdot\, - y)) | \alpha(\,\cdot\,,y)\rangle \,  {dy}  \\
+  \Tr \left[ K_{T,W} (\gamma-\gamma_0)^2\right]+ \frac {2T}3 \Tr\left[ \Gamma(1-\Gamma) - \Gamma_0(1-\Gamma_0)\right]^2 \,.
\end{multline}
In the first term on the second line, we can use the simple fact that $K_{T,W} \geq 2T$. Moreover, for the last term, we observe that 
\begin{equation}
  \Tr\left[ \Gamma(1-\Gamma) - \Gamma_0(1-\Gamma_0)\right]^2 \geq 2\, \Tr \left [ \gamma(1-\gamma)- \gamma_0(1-\gamma_0) - \alpha\bar\alpha\right]^2\,.
\end{equation}
We further claim that
\begin{equation}
  2  \, \Tr (\gamma -\gamma_0)^2 + \frac {4 }{3}\, \Tr \left [ \gamma(1-\gamma)- \gamma_0(1-\gamma_0) - \alpha\bar\alpha\right]^2 \geq \frac 4 5\, \Tr ( \alpha\bar\alpha )^2\,.
\end{equation}
This follows easily from the triangle inequality
\begin{equation}
  \|\alpha\bar\alpha\|_2 \leq 
  \|  \gamma(1-\gamma)- \gamma_0(1-\gamma_0) - \alpha\bar\alpha \|_2 + \|\gamma(1-\gamma)-\gamma_0(1-\gamma_0)\|_2
\end{equation}
together with the fact that
\begin{equation}
  \|\gamma(1-\gamma)-\gamma_0(1-\gamma_0)\|_2 \leq \|\gamma-\gamma_0\|_2\,,
\end{equation}
which can be seen using Klein's inequality, for
instance. In combination with (\ref{apribound1}), we have thus shown that 
\begin{equation}\label{apribound}
\F(\Gamma) - \F(\Gamma_0) \geq  \int_{\calC} \langle \alpha(\,\cdot\,,y) | K_{T,W} + V(h^{-1}(\, \cdot\, - y)) | \alpha(\,\cdot\,,y)\rangle \,  {dy} 
+ \frac {4T}5 \tr [\bar \alpha \alpha ]^2\,.
\end{equation}

As a next step, we want to get rid of the $W$ in the operator $K_{T,W}$. 
It was shown in \cite[Lemma~2]{FHSS} that
\begin{equation}\label{18}
 K_{T,W} + V \geq \frac 18(K_T + V) - C h^2 \,,
 \end{equation}
where $V$ is short for $V(h^{-1}(\,\cdot\,-y))$, $K_{T}$ is short for $K_{T,0}$ (which agrees with the  definition  \eqref{linopfinT} in the previous section, up to a change of variables $x \to x/h$),  and $C$ is a constant  depending only on $\|W\|_\infty$.
The statement looks rather obvious, since the perturbation $h^2W$ is bounded in norm by $C h^2$.
However the detailed calculation is rather lengthy, and we will not repeat it here but rather refer to \cite{FHSS}. It uses, e.g.,  the operator monotonicity of the function $x/\tanh x$ in $x^2$. The factor $1/8$ in (\ref{18}) can be replaced by any number less than $1$, at the expense of increasing the constant $C$. 

From (\ref{apribound}) and (\ref{18}) we conclude that 
\begin{align}\nonumber 
\F(\Gamma) - \F(\Gamma_0) & \geq  \frac 18  \int_{\calC} \langle \alpha(\,\cdot\,,y) | K_{T} + V(h^{-1}(\, \cdot\, - y)) | \alpha(\,\cdot\,,y)\rangle \,  {dy}  \\ & \quad  - Ch^2 \|\alpha\|^2_2 + \frac {4T}5 \tr [\bar \alpha \alpha ]^2\,. \label{hal}
\end{align}
Moreover, since  $T-T_c= O(h^2)$ by assumption, one easily sees that  
$$ K_T \geq  K_{T_c} -O(h^2)\,,$$
and thus \eqref{hal} holds also with $K_T$ replaced by $K_{T_c}$ in the first line, at the expense of increasing the constant $C$. 
Recall that, by assumption, the operator $ K_{T_c} + V(h^{-1}(\,\cdot\,-y))$ on
$L^2(\R^3)$ has a unique ground state, proportional to
$\alpha_0(h^{-1}(x-y))$, with ground state energy zero, and a gap
above, which we shall denote by $\kappa>0$. Note that $\kappa$ is independent of $h$.

 We are now considering states $\Gamma$ with low energy, i.e.,  states 
 which satisfy 
 \begin{equation}\label{gamcom} 
 \F(\Gamma) \leq  \F(\Gamma_0)\,.
 \end{equation}
For such states we thus have
\begin{equation}\label{suchs}
\frac 18  \int_{\calC} \langle \alpha(\,\cdot\,,y) | K_{T_c} + V(h^{-1}(\, \cdot\, - y)) | \alpha(\,\cdot\,,y)\rangle \,  {dy}  + \frac {4T}5 \tr [\bar \alpha \alpha ]^2 \leq  O( h^2) \|\alpha\|^2_2\,.
 \end{equation}
Define the periodic function $\psi \in L^2(\calC)$ by 
\begin{equation}\label{defxi0}
  \psi(y) =  (2\pi)^{3/2} \left(h \int_{\R^3} |\alpha_0(x)|^2 \, dx  \right)^{-1} \int_{\R^3} \alpha_0(h^{-1}(x-y)) \alpha(x,y) \,dx\,,
\end{equation}
and define an operator $\xi_0$ via the decomposition 
\begin{equation}\label{dec1}
  \alpha(x,y)=\psi(y) \frac{ h^{-2}}{(2\pi)^{3/2}} \alpha_0(h^{-1}(x-y))  + \xi_0(x,y) \,.
\end{equation}
Then the properties of $K_{T_c} + V$ discussed above imply that 
\begin{align}\nn
& \int_{\calC} \langle \alpha(\,\cdot\,,y) | K_{T_c} + V(h^{-1}(\,\cdot\,-y)) | \alpha(\,\cdot\,,y)\rangle \,  {dy} \\ &=  \int_{\calC} \langle \xi_0(\,\cdot\,,y) | K_{T_c} + 
V(h^{-1}(\,\cdot\,-y)) | \xi_0(\,\cdot\,,y)\rangle \,  {dy} \geq \kappa \|\xi_0\|_2^2\,.
\end{align}
In combination with (\ref{suchs}), this implies that 
$\|\xi_0\|_2^2 \leq h^2 \|\alpha\|^2_2$,  showing that for all states close enough to the normal state the part $\xi_0$ in the decomposition (\ref{dec1}), measuring in a certain sense the depletion of the condensate of Cooper pairs described by $\alpha_0$, 
has to be of lower order. 
Additionally, thanks to the gap of order one above $0$, there is a parameter $\delta>0$ such that 
\begin{align}\nonumber
&   \int_{\calC} \langle \xi_0(\,\cdot\,,y) | K_{T_c} + V(h^{-1}(\,\cdot\,-y)) | \xi_0(\,\cdot\,,y)\rangle \,  {dy} \\ & 
 \geq  \delta \int_{\calC} \langle \xi_0(\,\cdot\,,y) | - h^2 \nabla^2 | \xi_0(\,\cdot\,,y)\rangle \,  {dy} ,
\end{align}
hence also 
\begin{equation}
\label{xo0bound}
\|\xi_0\|_{H^1}^2 \leq C h^2 \|\alpha\|^2_2.
\end{equation}

An application of Schwarz's inequality to \eqref{defxi0} further yields
\begin{equation}
  \|\psi\|^2_2 = \int_\calC | \psi(x)|^2\, dx \leq  (2\pi)^3\frac{h \|\alpha\|_2^2}{ \int_{\R^3} |\alpha_0(x)|^2 dx } \,.
\end{equation}
Moreover,
\begin{multline}\label{alpbound}
 h \|\alpha\|_2^2 = \frac{1}{(2\pi)^3} \int_\calC |\psi(x)|^2 dx \, \int_{\R^3} |\alpha_0(x)|^2 dx +h  \|\xi_0\|_2^2\,
\\ \leq (1 + O(h^2)) \frac{ 1}{(2\pi)^3} \int_\calC |\psi(x)|^2 dx\, \int_{\R^3} |\alpha_0(x)|^2 dx \,.
\end{multline}
Hence $\|\psi\|_2^2$ is bounded from above and below by $h \|\alpha\|_2^2$, and also $ \|\xi_0\|_{H^1}^2 \leq c h \|\psi\|_2^2$. 
Additionally we need an a-priori bound on $\int_\calC |\nabla \psi(x)|^2 dx.$  
Again by using Schwarz's inequality,
\begin{equation}\label{schw2}
  \int_\calC |\nabla \psi(x)|^2\, dx \leq  (2\pi)^3 h \frac { \iint_{\R^3\times \calC}  \left| \left(\nabla_x + \nabla_y\right) \alpha(x,y)\right|^2 \, dx\,dy }{ \int_{\R^3} |\alpha_0(x)|^2dx } \,.
\end{equation}
In other words, the kinetic energy of $\psi$ is dominated by the center-of-mass kinetic energy of $\alpha$. 
In order to bound the latter, we use the inequality 
  \begin{align}\nonumber
    & h^2 \iint_{\R^3\times \calC} \left| \left(\nabla_x +
        \nabla_y\right) \alpha(x,y)\right|^2 \, dx\,dy \\ & \leq C
    \int_{\calC} \langle \alpha(\,\cdot\,,y) | K_{T_c} +
    V(h^{-1}(\,\cdot\,-y)) | \alpha(\,\cdot\,,y)\rangle \,
    dy \label{lem:eq:com}
  \end{align}
  for suitable $C>0$. The bound \eqref{lem:eq:com} is  proved in \cite[Lemma~3]{FHSS}; it uses again the gap in the spectrum of $K_{T_c} + V$ as well as the  asymptotic behavior of $K_{T_c}$ for large momentum. 
Eqs.~\eqref{schw2} and~\eqref{lem:eq:com}  in combination with \eqref{suchs} imply that 
 \begin{equation}
  \int_\calC |\nabla \psi(x)|^2\, dx  \leq O(h) \|\alpha \|^2_2 \leq C \|\psi \|_2^2\,.
  \end{equation}
  
In order to complete the proof of Prop.~\ref{proppair}, 
we symmetrize $\alpha$ in the form
 \begin{equation}
  \alpha(x,y) = \tfrac 12 \left( \psi( x)+\psi(y)\right) \frac{h^{-2}}{(2\pi)^{3/2}} \alpha_0(h^{-1}(x-y)) + \xi(x,y) \,,
\end{equation}
where now 
\begin{equation}
  \xi(x,y) = \xi_0(x,y) + \tfrac 12 \left(\psi(x) - \psi(y)\right)\frac{ h^{-2}}{(2\pi)^{3/2}} \alpha_0(h^{-1}(x-y)) \,.
\end{equation}
Using the fact that 
\begin{align}\nn
  &h^{-6} \iint_{\R^3 \times\calC} |\psi(x)-\psi(y)|^2 |\nabla
  \alpha_0(h^{-1}(x-y))|^2\,dx \, dy \\ \nn & = 4 h^{-3} \sum_{p\in
    (2\pi\Z)^3} |\hat \psi(p)|^2 \int_{\R^3} |\nabla
  \alpha_0(x)|^2 \sin^2\left( \tfrac 12 h p\cdot x \right) \, dx \\ \nn &  \leq
  h^{-1} \int_{\calC} |\nabla \psi(x)|^2 \, dx   \int_{\R^3} |\nabla
  \alpha_0(x)|^2 |x|^2 \, dx 
  \end{align}
it is then easy to see that $\xi$ satisfies the same $H^1$-bound as $\xi_0$ does, i.e., 
 \begin{equation}\label{xihn}
 \|\xi\|^2_{H^1} \leq C h \|\psi\|_{H^1(\calC)}^2 \leq C h \| \psi\|_2^2 \,. \end{equation} 
 
The only thing left to complete the proof of Prop.~\ref{proppair} is to given an a-priori bound on $\|\psi\|_2$ that is independent of $h$.  In order to achieve this, the second term on the left side of (\ref{suchs}), which was not used up to now, will be essential. To leading order in $h$, we expect that $ \tr [\bar \alpha \alpha]^2 \approx h \|\psi\|_4^4 (2\pi)^{-3} \int_{\R^3} | \hat \alpha_0(q)|^4 dq$, and indeed in 
 \cite[Lemma~4]{FHSS} a bound was proved that quantifies the extent to which this holds true. Note  that for functions on the unit cube 
 $$ \|\psi\|_4 \geq \|\psi\|_2\,.$$ 
The precise form of the bound in  \cite[Lemma~4]{FHSS} is not important for us here; the only relevant fact is the rough estimate 
 \begin{equation}\label{rough}
 \tr [\bar \alpha \alpha]^2 \geq C h \|\psi\|_4^4  - C h \|\psi \|_2^2
 \end{equation}
 that follows from it.\footnote{To see that \eqref{rough} follows  from \cite[Lemma~4]{FHSS} simply note that we can assume that $\|\psi\|_2 \leq C \|\psi\|_4^2$ otherwise the right side of \eqref{rough} is negative. Since we already know that $\|\psi\|_{H^1} \leq C \|\psi\|_2$ and thus also $\|\psi\|_4 \leq C \|\psi\|_2$ by Sobolev's inequality, the desired result follows in a straightforward way.}
 In combination with \eqref{suchs} it implies that 
  \begin{equation}
C h \|\psi\|_2^2 \geq  C h^2 \| \alpha\|^2_2   \geq C h \|\psi\|_2^4 - C h \|\psi \|_2^2  \,, \end{equation}
 which implies the desired a priori bound 
 $$ \|\psi\|_{2}^2 \leq C.$$ 
 This concludes the proof of Prop.~\ref{proppair}, and thus also Step 1.

 \subsubsection{Step $3$} 
 
We shall present  Step $3$ before Step $2$.  
In this step we shall evaluate the BCS energy of a trial state of the form 
\begin{equation}\label{trials}
  \Gamma_{\Delta} = \frac 1{1 + e^{\beta H_{\Delta}}} \,,
  \end{equation}
with 
\begin{equation}
\Delta(x,y)=2V(h^{-1}(x-y))\alpha_{\rm
  GL}(x,y) 
\end{equation}
and
$$ \alpha_{\rm GL}=\tfrac h 2 \big( \psi_0(x) \hat \alpha_0(-ih\nabla) + \hat \alpha_0(-ih\nabla) \psi_0(x)\big), $$
where $\psi_0 \in H^2(\calC)$. 
In fact, for a trial state of this kind  $\Gamma_\Delta$ we obtain the following theorem.

 \begin{theorem}[BCS energy of trial state]\label{thm8}
  \begin{align}
 \nn & h^3 \left( \F(\Gamma_\Delta) - \F(\Gamma_0) \right) \\ \nonumber 
  &\quad =  -h^2 (2\pi)^{-3} \|\psi_0\|_2^2 \left( \langle \alpha_0|K_{T_c} K_T^{-1} K_{T_c} \,|\alpha_0 \rangle + \langle \alpha_0 |V|\alpha_0\rangle \right) \\ \nonumber & \quad\quad 
   + \lambda_0 h^4
  (\mathcal{E}^{\rm GL}(\psi_0) + \lambda_2 D  \|\psi_0\|_2^2 ) \\ &\quad\quad + O(h^5) \left( \|\psi_0\|_{H^1(\calC)}^4 + \|\psi_0\|_{H^1(\calC)}^2\right)  + O(h^6)\left( \|\psi_0\|_{H^1(\calC)}^6 +
  \|\psi_0\|_{H^2_{\calC}}^2 \right).
  \label{inn2}
\end{align}
\end{theorem}

This theorem shows that the trial state \eqref{trials} for an appropriate choice of $\psi_0$ gives the correct upper bound to the BCS energy. The result of this theorem will also be essential for the lower bound, however. We will show below in Section~\ref{ss:s2} that the difference between $h^3 \F(\Gamma)$ and $h^3 \F(\Gamma_\Delta)$ is of higher order than $h^4$ if we choose $\psi_0$ to be the $\psi$ from Proposition~\ref{proppair} or, more precisely, a slightly mollified version of it. 

Note that the second line on the  right side of \eqref{inn2} contains all the terms in the GL functional \eqref{def:GL} except for the term $-\lambda_2 D \|\psi_0\|_2^2$. This last term is actually contained in the first line on the right side of \eqref{inn2}, as we shall demonstrate below after the proof of Theorem~\ref{thm8}.

\begin{proof} 
For our choice of trial state $\Gamma_\Delta$ we know from the identity in Lemma \ref{lem:id} 
 that we can write the difference of the corresponding BCS energies in the  form
 \begin{align}\nonumber
  &\F(\Gamma_\Delta) - \F(\Gamma_0) \\ \nonumber &
  = - \frac 1{2\beta} \Trs\left[ \ln(1+e^{-\beta
      H_\Delta})-\ln(1+e^{-\beta H_0})\right] \\ \nonumber
  & \quad -  \iint_{\R^3 \times \calC} V(\tfrac {x-y}h)|\alpha_{\rm GL}(x, y)|^2\, dx\,dy \\
  & \quad + \iint_{\R^3 \times \calC} V(\tfrac{x-y}h)\left|
    \alpha_{\rm GL}(x,y)-
    \alpha_\Delta(x,y)\right|^2\,{dx\,dy}\,. \label{equ:diff}
\end{align}
The first term on the right  side can be calculated via a semiclassical expansion using
Cauchy's integral formula
\begin{equation}\label{sp2}
  f( \beta H_\Delta)-f(\beta H_0)  = \frac 1{2\pi i}\int_C f(\beta z) \, \left( \frac 1{z-H_\Delta} - \frac 1{z-H_0} \right)\, dz
\end{equation}
where $C$ is the contour $z= r \pm i\tfrac \pi {2\beta}$, $r\in
\R$, and $f(z) = \ln (1 + e^{\beta z})$. The integral on the right side has to be  understood in the weak sense.
By expanding the resolvent $1/(z - H_\Delta)$ in $\Delta$ and the external fields 
one obtains the following lemma.

\begin{lemma}\label{lem412}
With errors controlled by $H^1$ and $H^2$ norms of $\psi_0$, exactly as in \eqref{inn2}, we have 
\begin{multline}
 -\frac {h^3} 2 T\,  \Trs  \left[ \ln\left(1 + e^{-H_\Delta/T}\right) -  \ln\left(1 + e^{-H_0/T} \right)\right] \\
 = - h^2 (2\pi)^{-3}\|\psi_0\|_2^2 \langle \alpha_0 |K_{T_c} K_T^{-1}K_{T_c} |\alpha_0 \rangle  +  h^4 \mathcal{D}_4(\psi_0)  \\ + \lambda_0 h^4
 \left (\mathcal{E}^{\rm GL}(\psi_0)  + \lambda_2 D \|\psi_0\|_2^2 \right) + O(h^5) \label{inn3}
\end{multline}
and
\begin{multline}\label{in2}
h^{3}\iint_{\R^3 \times \calC} V (h^{-1}(x - y))|\alpha_{\rm GL}(x, y)|^2 \, dx\, dy  \\ =  h^2 (2\pi)^{-3} \|\psi_0\|_2^2 \langle \alpha_0 |V|\alpha_0\rangle  + h^4 \mathcal{D}_4(\psi_0)  +O(h^6)
\end{multline}
for a suitable $\mathcal{D}_4(\psi_0)$ proportional to $\|\nabla \psi_0\|_2^2$.
\end{lemma}

\begin{proof}
The proof of this lemma is given in Section~7 in \cite{FHSS}, together with the estimates in  Eqs.~(4.10)--(4.13) there. 
We give here a sketch of the proof. 
Let again $k$ denote the operator
\begin{equation}\label{def:kk}
  k = -h^2 \nabla^2 - \mu + h^2 W(x) \,.
\end{equation}
The resolvent identity and the fact that
\begin{equation}
  \delta:=H_\Delta-H_0 = \left(\begin{matrix} 0 & \Delta \\ \bar\Delta  & 0 \end{matrix}\right)
\end{equation}
is off-diagonal (as an operator-valued $2\times 2$ matrix) implies the upper left entry of the difference between the resolvents of $H_\Delta$ and the one of $H_0$ satisfies the identity
\begin{equation}
  \left[ \frac 1{z-H_\Delta} - \frac 1{z-H_0} \right]_{11} = I_1 + I_2 + I_3
\end{equation}
where
\begin{equation}\label{def:i1}
  I_1 =  \frac 1{z-k} \Delta \frac 1{z+\bar k} \Delta^\dagger \frac 1{z-k} \,,
\end{equation}
\begin{equation}\label{def:i2}
  I_2  =   \frac 1{z-k} \Delta \frac 1{z+\bar k} \Delta^\dagger \frac 1{z-k} \Delta \frac 1{z+\bar k} \Delta^\dagger \frac 1{z-k} \,,
\end{equation}
and
\begin{equation}\label{def:i3}
  I_3 =  \left[\frac 1 {z-H_\Delta} \right]_{11} \Delta \frac 1{z+\bar k} \Delta^\dagger \frac 1{z-k} \Delta \frac 1{z+\bar k} \Delta^\dagger \frac 1{z-k}\Delta \frac 1{z+\bar k} \Delta^\dagger \frac 1{z-k} \,.
\end{equation}
The main contributions leading to the Ginzburg--Landau functional are extracted from $I_1$ and $I_2$. The rest is estimated via the $H^1(\calC)$ and $H^2(\calC)$ norm of $\psi_0$. For instance, the contribution of the term $I_3$ can be bounded as 
 \begin{equation}
    \left\| \int_\Gamma f(\beta z)\, I_3 \, dz\right\|_1  \leq C h^{3}  \|\psi_0\|_{H^1(\calC)}^6 \|t\|_6^6\,.
  \end{equation}
The proof is given in \cite[Lemma~8]{FHSS}.
As an example we calculate one term in the expansion explicitly.
We have  
\begin{equation}\label{res}
    \frac 1{z-k} = \frac 1{z-k_0} + \frac 1{z-k_0} \left( k - k_0  \right) \frac 1{z-k} \,, 
  \end{equation}
  where
  \begin{equation}\label{def:k0}
    k_0 = -h^2 \nabla^2 - \mu\,,
  \end{equation}
  and hence
  \begin{equation}\label{difk}
    k-k_0 =  h^2 W(x). 
  \end{equation}
Using the resolvent identity (\ref{res}) we can write $I_1=
I_1^a+I_1^b$, where the main contribution comes from the term
\begin{equation}
  I_1^a=  \frac 1{z-k_0} \Delta \frac 1{z+k_0} \Delta^\dagger \frac 1{z-k_0}\,.
\end{equation}

By using the residue formula, one calculates explicitly
\begin{align}\nn
 & \frac 1{2\pi i} \Tr  \int_\Gamma f(\beta z) \,I_1^a \, dz \\ \nn  
& = \frac {h^2}{2 \pi i}   \sum_{p \in (2\pi \Z)^3}  |\hat \psi_0(p)|^2 \\ \nn & \quad \quad \quad \times  \int_\Gamma f(\beta z) \int_{\R^3}  \frac {\left[ t (h(p+q))+ t(hq) \right]^2 } {\left (z -  h^2(p+q)^2 + \mu \right) \left(z+ h^2 q^2 - \mu \right) }  \frac {dq}{(2\pi)^3}\, dz\\  & 
= h^{-1} \sum_{p \in (2\pi \Z)^3} | \hat \psi_0(p) |^2  G (hp) \,, \label{469}
\end{align}
where $G(p)$ equals
\begin{equation}
  -\frac \beta{4}
  \int_{\R^3} \frac{\left( t(q+p) + t(q)\right)^2}{4}  \frac { \tanh\left(\tfrac 12 \beta((q+p)^2-\mu)\right) +   \tanh\left(\tfrac 12 \beta(q^2-\mu)\right)}{(q+p)^2 + q^2 -2\mu}\, \frac{dq}{(2\pi)^3}\,.
\end{equation}
One now expands $G(hp)$ in powers of $h^2$, and uses
  \begin{equation}\label{gexp}
    \left| G(h p) - G(0) - h^2\tfrac 12 (p\cdot \nabla)^2 G(0) \right| \leq C h^4|p|^4
  \end{equation}
  for some constant $C$ which can be estimated in terms of $t$ and its derivatives. After inserting this bound in \eqref{469}, this gives an 
error of the order $ h^{3} \|\psi_0\|_{H^2}^2$.
The term $ \tfrac 12 (p\cdot \nabla)^2 G(0) $ recovers both the gradient term in the Ginzburg--Landau functional and the term $\mathcal{D}_4(\psi_0)$, and is proportional
to $ h \int_{\calC} |\nabla \psi_0(x)|^2 dx$ after insertion into \eqref{469}. 
For the lowest order term, we
have
\begin{equation}
  G(0)  = -\frac {\beta}{4}
  \int_{\R^3} t(q)^2  \frac{ \tanh \left( \tfrac 12 \beta (q^2-\mu) \right)}{q^2 -\mu}  \frac{dq}{(2\pi)^3}\,.
\end{equation}
Using the definition of $t$ in terms of $\alpha_0$, this can be rewritten as  
\begin{equation}
G(0) = \frac {\beta}{(2\pi)^3}  \langle \alpha_0|K_{T_c} \frac 1{K_T} K_{T_c} |\alpha_0 \rangle \,.
\end{equation}
This yields the first term on the right side of \eqref{inn3}. The other terms can be calculated in a similar manner, and we refer to \cite{FHSS} for the details.
\end{proof}

To conclude the proof of Thm.~\ref{thm8}, it remains to estimate the last term in \eqref{equ:diff}. This involves similar techniques and semiclassical estimates as in the proof of Lemma~\ref{lem412}, and we refer to \cite{FHSS}, specifically Theorem~3 and Eqs. (4.14)--(4.18) there, for details. The result is that 
\begin{multline}
\left| \iint_{\R^3 \times \calC} V(\tfrac{x-y}h)\left|
    \alpha_{\rm GL}(x,y)-
    \alpha_\Delta(x,y)\right|^2\,{dx\,dy} \right| \\ \leq C  h^3\left( \|\psi_0\|_{H^1(\calC)}^6 +
  \|\psi_0\|_{H^2_{\calC}}^2 \right) \,.
  \end{multline}
This concludes our (sketch of the) proof of Theorem~\ref{thm8}.
\end{proof}

Finally observe that 
$$ \langle \alpha_0| K_{T_c} \frac 1{K_T} K_{T_c}| \alpha_0 \rangle + \langle \alpha_0 |V|\alpha_0\rangle  
=  \langle \alpha_0| K_{T_c} \left( \frac 1{K_T} - \frac 1{K_{T_c}}\right) K_{T_c}| \alpha_0 \rangle,$$
where we used that $(K_{T_c} + V)\alpha_0 = 0 $. 
We can rewrite this expression in terms of $t$ as 
\begin{multline} \langle \alpha_0| K_{T_c} \left( \frac 1{K_T} - \frac 1{K_{T_c}} \right)K_{T_c} |\alpha_0 \rangle \\ = \frac 14  \int_{\R^3} t(q)^2 \frac{ \tanh\left( \tfrac 12 \beta( q^2 -\mu) \right)- \tanh\left( \tfrac 12 \beta_c( q^2 -\mu) \right)}{  q^2 -\mu}   \, dq \,.
\end{multline}
For temperatures  $T = T_c(1 - D h^2)$ we evaluate
\begin{multline}
 \frac 14  \int_{\R^3} t(q)^2 \frac{ \tanh\left( \tfrac 12 \beta( q^2 -\mu) \right)- \tanh\left( \tfrac 12 \beta_c( q^2 -\mu) \right)}{  q^2 -\mu}   \, dq
     \\ = \frac 18  {h^2 \beta_c} D \int_{\R^3} t(q)^2  \cosh^{-2} \left( \tfrac 12 \beta_c (q^2 -\mu) \right)  \, dq + O(h^4)\,,\\ =  (2\pi)^3
    \lambda_0 D \lambda_2 + O(h^4). 
\end{multline}
Altogether we therefore obtain for $T= T_c(1 - D h^2)$
\begin{multline}\label{impequ}
 h^3 [\F(\Gamma_\Delta) - \F(\Gamma_0)] = \lambda_0 h^4 \E^{\rm GL}(\psi_0) \\ +O(h^5) \left( \|\psi_0\|_{H^1(\calC)}^4 + \|\psi_0\|_{H^1(\calC)}^2\right)  + O(h^6)\left( \|\psi_0\|_{H^1(\calC)}^6 +
  \|\psi_0\|_{H^2({\calC})}^2 \right).
  \end{multline}
 
\subsubsection{Step $2$}\label{ss:s2}

Recall the decomposition 
$$\alpha(x,y) =  \tfrac 12 \left( \psi( x)+\psi(y)\right) \frac{h^{-2}}{(2\pi)^{3/2}} \alpha_0(h^{-1}(x-y)) + \xi(x,y) $$
from Prop.~\ref{proppair}. Since $\xi$ is of higher order in $h$ than the first term, we can use the $\psi$ from this decomposition as an input in the construction of the trial state \eqref{trials} in the previous step. The semiclassical estimates involve the $H^2$-norm of $\psi$, however, while a-priori we only have a bound on its $H^1$-norm. Hence we perform a high-momentum cutoff of $\psi$ and absorb the error term into $\xi$.

More precisely, pick some small $\epsilon>0$, which we will later choose as $\epsilon = h^{1/5}$, we define the part of $\psi$ whose momenta are smaller than 
$\epsilon h^{-1}$ as  $\psi_<$. That is, in terms of the Fourier coefficients we have 
\begin{equation}
  \hat
  \psi_<(p) = \hat \psi(p) \theta(\epsilon h^{-1} - |p|) \,, 
\end{equation}
where $\theta$ denotes the Heaviside step function, i.e., $\theta(t) = 1$ for $t\geq 0$, and $0$ otherwise.  The function
$\psi_<$ is thus smooth,
and its $H^2$-norm is bounded by  
\begin{equation}\label{h2e}
 \|\psi_<\|^2_{H^2(\calC)}\leq \|\psi\|^2_2 + \epsilon^2 h^{-2} \|\psi\|_{H^1(\calC)}^2 \,.
\end{equation}

Let also $\psi_> = \psi - \psi_<$. Its  $L^2(\calC)$ norm  is bounded by 
$$ \int_{\calC} |\psi_>(x)|^2 dx = \sum_p |\hat{\psi}_>(p)|^2  \leq h^2 \epsilon^{-2} \sum_p |\hat{\psi}_>(p)|^2 |p|^2 \leq h^2 \epsilon^{-2} \|\psi\|_{H^1(\calC)}^2. $$
We absorb the part $\frac
12(\psi_>(x)+\psi_>(y))\alpha_0(h^{-1}(x-y))$ into $\xi(x,y)$, and write
\begin{equation}\label{defpsi}
  \alpha(x,y) = \tfrac 12 \left( \psi_<(x)+\psi_<(y)\right) \frac{h^{-2}}{(2\pi)^{3/2}} \alpha_0(h^{-1}(x-y)) +\sigma(x,y)
\end{equation}
where
\begin{equation}\label{def:sigma}
  \sigma (x,y) = \xi(x,y) + \tfrac 12 \left( \psi_>(x)+\psi_>(y)\right) \frac{h^{-2}}{(2\pi)^{3/2}} \alpha_0(h^{-1}(x-y)) \,.
\end{equation}
With \eqref{xihn} we conclude that 
\begin{equation}\label{h1ss}
\|\sigma\|_{H^1}^2 \leq  C h( 1 + \epsilon^{-2})\|\psi\|_{H^1(\calC)}^2
\end{equation}
(recall the definition \eqref{def:h1} for the $H^1$-norm of a periodic operator). 
   
Let now $$\Delta = -\frac h2 (\psi_<(x)
t(-ih\nabla) + t(-ih\nabla) \psi_<(x)).$$  Its integral kernel is given
in (\ref{int:delta}), with $\psi$ replaced by $\psi_<$.  Using the identity \eqref{fi}
with $\psi_0 = \psi_<$ we obtain 
\begin{align}\nonumber
  & \F^{\rm BCS}(\Gamma) - \F^{\rm BCS}(\Gamma_0) \\ \nonumber & =
  -\frac T 2 \Trs\left[ \ln(1+e^{-\beta H_\Delta})-\ln(1+e^{-\beta
      H_0})\right] \\ \nonumber & \quad - h^{-4}\iint_{\R^3 \times \calC} V(h^{-1}(x-y)) \tfrac 14 \left| \psi_<(x)+\psi_<(y)\right|^2
  |\alpha_0(h^{-1}(x-y))|^2\, \frac{dx\,dy}{(2\pi)^3}\\ & \quad +
  \tfrac 12 T \, \H_0(\Gamma,\Gamma_\Delta) + \iint_{\R^3 \times \calC}
  V(h^{-1}(x-y))|\sigma(x,y)|^2\, {dx\,dy}\,, \label{off}
\end{align}
where $\H_0$ denotes again the relative entropy. 
Using the semiclassical estimates from the previous step,  together with our a-priori bounds on $\psi$, we arrive at
\begin{align}\nonumber
  & \F^{\rm BCS}(\Gamma) -  \F^{\rm BCS}(\Gamma_0)   \\ \nn & \geq h \lambda_0
    \E^{\rm GL}(\psi_<)  -O(h^2)  - O(h^3)  \|\psi_<\|_{H^2({\calC})}^2  \\ & \quad +
  \tfrac 12 T \, \H_0(\Gamma,\Gamma_\Delta) + \iint_{\R^3 \times \calC}
  V(h^{-1}(x-y))|\sigma(x,y)|^2\, {dx\,dy}\,. \label{lb2:s1}
\end{align}
As long as $\epsilon\ll 1$, the terms in the second line give the desired Ginzburg--Landau energy, due to \eqref{h2e}. 

With the aid of Lemma~\ref{lem:klein}, applied to $H=\beta H_\Delta$, we can obtain the lower bound 
\begin{equation}\label{in}
T\, \mathcal{H}_0(\Gamma,\Gamma_\Delta) \geq \Trs  \left[  \frac {H_\Delta}{\tanh \tfrac 12 \beta H_\Delta}     \left(\Gamma -\Gamma_\Delta\right)^2 \right]
\end{equation}
on the relative entropy. Strictly speaking, we cannot directly apply Lemma~\ref{lem:klein}, since $H_\Delta$ is not diagonal. A careful analysis shows that the inequality nevertheless holds, see \cite[Lemma~5]{FHSS}. The right side should be properly interpreted as 
$$
 \Tr  \left[    \left(\Gamma -\Gamma_\Delta\right) \frac {H_\Delta}{\tanh \tfrac 12 \beta H_\Delta}     \left(\Gamma -\Gamma_\Delta\right) \right]\,,
 $$
i.e., it is the trace of a non-negative operator. 

Using the operator-monotonicity of $x \mapsto \sqrt x/ \tanh \sqrt x$ and our a-priori bounds on $\psi_<$, it is not hard to show that 
$$   \frac {H_\Delta}{\tanh \tfrac 12 \beta H_\Delta} \geq \left( 1 - O(h) - O (\sqrt{\epsilon h}) \right) K_T \otimes \id_{\C^2} \,,$$
which together with \eqref{in} implies that 
$$
\frac T 2 \mathcal{H}_0(\Gamma,\Gamma_\Delta) \geq \left( 1 - O(h) - O (\sqrt{\epsilon h}) \right)  \Tr \left[ K_T (\alpha-\alpha_\Delta)^2 \right] \,.
$$
If we write $\alpha-\alpha_\Delta = \sigma - \phi$, a semiclassical analysis as in the proof of Theorem~\ref{thm8} implies that $\|\phi\|_{H^1}^2 \leq O(h \epsilon^2)$. The $H^1$-norm of $\sigma$ can be bigger, see \eqref{h1ss}, but in combination with the last term in \eqref{lb2:s1} we can again use the (almost) positivity of $K_T + V$ to bound this term. The details are somewhat messy and we refer to \cite[Lemma~6.2]{FHSS3} where it is shown that
\begin{multline}
 \tfrac 12 T \, \H_0(\Gamma,\Gamma_\Delta) + \iint_{\R^3 \times \calC}
  V(h^{-1}(x-y))|\sigma(x,y)|^2\, {dx\,dy}\\ \geq - C \left( h^2 \epsilon^{-1} + h\epsilon + h^{3/2} \epsilon^{-3/2} \right)\, . 
  \end{multline}
The optimal choice of $\epsilon$ is $\epsilon = O(h^{1/5})$ in order to minimize the error. This concludes the (sketch of the) proof of Theorem~\ref{main:thm}.

\subsection{Absence of external fields} \label{timin}

Let us now consider the  BCS energy functional \eqref{def:rBCS2} in the absence of external  fields $W$ and $A$,
$$ \mathcal{F}({\Gamma})={\rm Tr\,} \left[ \left(
      -h^2\nabla^2 -\mu \right) \gamma \right] - T\, S(\Gamma)
+ \iint_{\R^3 \times \calC}
  V(h^{-1}(x-y)) |\alpha(x,y)|^2 \, {dx \, dy} \,.$$
We want to investigate whether its infimum is actually attained by the minimizer of the translation-invariant problem, which was studied in detail in Section \ref{ss:ti}. 
The proof that this  holds true under  suitable assumptions on the interaction potential $V$ is implicitly contained in \cite{FHSS} and is a simple consequence of Lemma \ref{lem:klein}. We shall spell it out in the following.

Let us assume we have a potential $V$ such that $T_c(V) > 0$.  Then, as we saw in Section \ref{ss:ti}, there exists a non-trivial  solution $\Delta_0$ of the BCS gap equation. If we write $\Delta_0(x) = 2 V(x) \alpha_0(x)$, the gap equation has the form
$$ (K_T^{\Delta_0} + V)\alpha_0 = 0.$$
Let us now make the {\em assumption} that 
\begin{equation}\label{assp}
K_T^{\Delta_0} + V \geq 0 \,.
\end{equation}
By a Perron--Frobenius type argument, this can for instance be shown for potentials with non-positive Fourier transform, i.e, $\hat V \leq 0$. 
Let us  sketch this proof. The main point is that, if $\hat V \leq 0$ then
\begin{equation}\label{negv} \langle \hat \alpha|\hat V \ast \hat \alpha \rangle \geq \langle |\hat \alpha|| \hat V \ast |\hat \alpha|\rangle.
\end{equation}
Assume now that $ (\hat \gamma,\hat \alpha)$ minimizes the translation-invariant BCS functional \eqref{freeenergy}. Then 
$\F(\hat \gamma,\hat \alpha) \geq \F(\hat \gamma,|\hat \alpha|)$, hence also $(\hat \gamma,|\hat \alpha|)$ is a minimizer. Consequently  
the Fourier transform of $|\hat \alpha(p)|$ is an eigenvector of 
$K_T^{\Delta_0} + V$ to the eigenvalue zero. Using \eqref{negv} again implies the same property for the ground state eigenfunction of $K_T^{\Delta_0} + V$, implying that $0$ has to be the 
lowest eigenvalue. In other words, \eqref{assp} holds.  

We remark that by continuity \eqref{assp} also holds under Assumption \ref{as1} for $T$ close to $T_c$.

Now let $\Gamma_{\Delta_0}$ be the translation-invariant minimizer of the BCS functional, i.e.,  $$\Gamma_{\Delta_0} = \frac 1{1+ e^{\beta H_{\Delta_0}}}.$$ Recall that $\Gamma_{\Delta_0}$ has $\alpha_0$ as its off-diagonal entry. With the aid of Lemma~\ref{lem:id}, applied to $\psi =1$, 
we can write 
\begin{multline}
\F(\Gamma) - \F(\Gamma_{\Delta_0}) = \tfrac T2 {\mathcal H}_0(\Gamma,\Gamma_{\Delta_0}) +  \iint_{\R^3\times \calC} V(x-y) |\alpha(x,y) - \alpha_0(x-y)|^2\, dx \,dy,
\end{multline}
where $\mathcal{H}_0$ denotes again the relative entropy. 
Applying now Lemma \ref{lem:klein} gives
\begin{align}\nn
 \frac T2 {\mathcal H}_0(\Gamma,\Gamma_{\Delta_0}) & \geq \frac 12  \Trs \frac{H_{\Delta_0}}{\tanh \tfrac \beta 2 H_{\Delta_0}} (\Gamma- \Gamma_{\Delta_0})^2 \\  \nn & \geq \Tr K^{\Delta_0}_T \left(\bar \alpha - \bar \alpha_0\right)\left( \alpha - \alpha_0\right)  \\ & = \int_\calC \langle \phi(\,\cdot\, ,y)| K_T^{\Delta_0} | \phi(\cdot,y)\rangle dy,
\end{align}
where we used that  $$\frac{H_{\Delta_0}}{\tanh \tfrac \beta 2 H_{\Delta_0}} = K^{\Delta_0}_T \otimes \id_{\C^2}$$
and denoted  $\alpha - \alpha_0 = \phi$ for brevity.  
Hence we obtain 
$$ \F(\Gamma) - \F(\Gamma_{\Delta_0})  \geq  \int_\calC \langle \phi(\,\cdot\,,y)| K_T^{\Delta_0} + V(h^{-1}(\cdot - y)) | \phi(\,\cdot\,,y)\rangle dy \geq 0$$
under the assumption \eqref{assp}. We conclude that under this assumption the trans\-lation-invariant minimizer is in fact the true minimizer of the BCS functional   \eqref{def:rBCS2} in the absence of external  fields, i.e.,  translation symmetry is not broken.

\appendix

\section{Quasi-free states}

In the Appendix we summarize some features of quasi-free states. The results and statements are taken from \cite{BLS94} (see also \cite{solovej} for further details).

\subsection{Bogoliubov transformations}

A \emph{Bogoliubov transformation} on $\mathcal{F}_{\mathcal{H}}$ is a
unitary operator $W:\mathcal{F}_{\mathcal{H}} \rightarrow
\mathcal{F}_{\mathcal{H}}$ associated to linear operators $v,w:\mathcal{H}\rightarrow \mathcal{H}$ 
such that for $\psi
\in \mathcal{H}$, one has
\begin{equation}\label{bogo}
  W a^\dagger(\psi) W^\dagger = a^\dagger(v\psi) + a(\overline{w\psi}),
\end{equation}
where $\overline{\psi}(x) = \overline{\psi(x)}$ denotes complex
conjugation.\footnote{As already mentioned in Section~\ref{QM}, complex conjugation in an abstract Hilbert space corresponds to the choice of an anti-linear involution. 
The reason for its necessity is the antilinearity of the annihilation operator.
  Alternatively, one could define the annihilation operator to accept as its
  argument an element of
  $\mathcal{H}^*$ instead of $\mathcal{H}$ (i.e., one 
  replaces $a$ by $\widetilde{a}$, where $\widetilde{a}(J\psi) = a(\psi)$ with $J: \mathcal{H} \rightarrow
  \mathcal{H}^*$ being the conjugate linear map such that $(J\psi)(\phi) =
  \langle \psi| \phi\rangle$), in which case the antilinearity would be
  naturally absorbed in $J$. This is  the approach followed in \cite{solovej}.}
As a unitary operator, $W$ leaves the canonical
anticommutation relations invariant, i.e., the operators $W a^\dagger(\psi) W^\dagger$ and $W a(\phi)
W^\dagger$ satisfy the canonical anticommutation relations. This implies that 
\begin{equation}\label{defu}
  U := \left(
    \begin{matrix}
      v & \overline{w}\\
      w & \overline{v}
    \end{matrix}
  \right) : \mathcal{H}\oplus \mathcal{H} \rightarrow \mathcal{H}\oplus \mathcal{H}
\end{equation}
has to be  unitary, where $\overline{T} \psi := \overline{T\overline{\psi}}$
for an arbitrary operator $T:\mathcal{H}\rightarrow \mathcal{H}$. Moreover, in order for $W$ to map the vacuum to a vector in $\F_{\mathcal{H}}$, it is necessary for $w$ to be a Hilbert--Schmidt operator. The latter property is known as the Shale--Stinespring criterion.

Conversely, any unitary operator of the form \eqref{defu}  with $w$  Hilbert--Schmidt defines a Bogoliubov transformation via  \eqref{bogo}.

\subsection{$SU(2)$-invariance}

We denote vectors $\psi \in \mathcal{H} = L^2(\mathbb{R}^3) \oplus
  L^2(\mathbb{R}^3) \cong L^2(\mathbb{R}^3) \otimes \mathbb{C}^2$ by
  \begin{equation*}
    \psi = (\psi_{\uparrow}, \psi_{\downarrow}).
  \end{equation*}
In other words, we think of $\psi$ as an element of $L^2(\mathbb{R}^3;
\mathbb{C}^2)$ such that $\psi(x) \in \mathbb{C}^2$ for each $x\in\R^3$. 
A rotation in   spin space is described by a matrix $S\in SU(2)$ which
acts on $\mathcal{H}$ according to
\begin{equation*}
  (S\psi)_\nu(x) = \sum_{\sigma \in \{\uparrow,\downarrow\}} S_{\nu,\sigma}\psi_\sigma (x).
\end{equation*}
On the Fock space $\mathcal{F}_{\mathcal{H}}$ the action of $S\in SU(2)$ is given by
the Bogoliubov transformation
$W_S\in\mathcal{L}(\mathcal{F}_{\mathcal{H}})$, which
transforms the creation and annihilation operators according to
\begin{align*}
  W_S a^\dagger(\psi)W_S^\dagger &= a^\dagger(S \psi) \\
  W_S a(\psi) W_S^\dagger &= a(S \psi).
\end{align*}
It thus corresponds to $w=0$ and $v=S$ in the notation \eqref{defu}.

A state $\rho$ is said to be invariant under spin rotations or
shortly $SU(2)$-invariant if
\begin{align}
  \langle W_S A W_S^\dagger\rangle_\rho = \langle A \rangle_\rho
\end{align}
for all operators $A$ on the Fock space. 
For quasi-free states $\rho$, $SU(2)$-invariance implies 
\begin{equation}
\begin{split}
  \langle W_S a^\dagger(\psi) a(\varphi)  W_S^\dagger\rangle_\rho  & =   \langle  a^\dagger(S\psi) a(S\varphi) \rangle_\rho =\langle \varphi | S^\dagger \gamma S \psi \rangle =\langle \varphi | \gamma \psi \rangle \\
 \langle W_S a(\psi) a(\varphi)  W_S^\dagger\rangle_\rho & =   \langle  a(S\psi) a(S\varphi) \rangle_\rho =\langle \varphi |S^\dagger \alpha \overline{ S} \overline{ \psi} \rangle =\langle \varphi | \alpha \overline{\psi }\rangle   \,.
\end{split}
\end{equation}
Consequently, 
in terms of $\Gamma$ the $SU(2)$-invariance can be conveniently expressed 
as
 $$  \mathcal{S}^\dagger \Gamma \mathcal{S} = \Gamma$$
with 
$$  \mathcal{S}=    \left(\begin{matrix}S&0\\ 0 &\bar S \end{matrix}\right).$$

\subsection{The von-Neumann entropy $S(\rho)$}\label{ss:vne}

\begin{lemma}
  \label{lemma:entropy}
  Let $\rho$ be a quasi-free state with one-particle density matrix
  $\Gamma$. Then
  \begin{equation*}
    S(\rho) = -\tr_{\mathcal{F}_\mathcal{H}}\big(\rho
    \ln \rho\big)
    = -\tr_{\mathcal{H}\oplus\mathcal{H}} \big(\Gamma \ln \Gamma\big).
  \end{equation*}
\end{lemma}

For the proof of this equality, the following characterization of
quasi-free states is useful:

\begin{lemma}\label{lemma:qf}
  For each quasi-free state $\rho$ with finite particle number, i.e., $ \tr \N \rho < \infty$, there is an orthonormal basis
  $\{\varphi_i\}_{i\in\mathbb{N}}$ of $\mathcal{H}$ and a Bogoliubov
  transformation $W$ on $\F_{\mathcal{H}}$ 
  such that
  \begin{equation}    \label{eq:rho_Q}
    W^\dagger\rho W = \frac{1}{\tr_{\mathcal{F}_\mathcal{H}} P e^Q }P e^Q\,,
  \end{equation}
  where 
  \begin{equation*}
    Q = \sum_{i\in I} q_i a^\dagger(\varphi_i)a(\varphi_i)
  \end{equation*}
  with the $q_i$ satisfying 
  \begin{equation*}
    \frac{\ee^{q_j}}{1+\ee^{q_j}} = \langle a^\dagger(\varphi_j)a(\varphi_j) \rangle_{W^\dagger\rho\, W},
  \end{equation*}
  $I = \{j\in\mathbb{N}| \langle a^\dagger(\varphi_j)a(\varphi_j)
  \rangle_{W^\dagger\rho\, W} \neq 0\}$,  and  $P
$ is the projection onto the kernel of 
  $\sum_{i\in\mathbb{N}\setminus I}
  a^\dagger(\varphi_i)a(\varphi_i)$.
\end{lemma}

\begin{proof}
  Following \cite[Proof of Theorem 2.3]{BLS94},
  we first find an orthonormal basis of $\mathcal{H}\oplus \mathcal{H}$
  that diagonalizes the one-particle
  density matrix $\Gamma$ of $\rho$.
  Note that $\tr\big(\Gamma(1-\Gamma)\big) <
  \infty$ by the assumption of finite particle number. Hence, there is an
  orthonormal basis of eigenvectors of $\Gamma(1-\Gamma)$.
  If $\psi$ is an eigenvector of $\Gamma(1-\Gamma)$ to the
  eigenvalue $\lambda$, then so is
  $\Gamma\psi$. Since $\Gamma^2 \psi = \Gamma \psi - \lambda
  \psi$, it follows that $\Gamma$ leaves invariant the subspace
  $\{\psi, \Gamma\psi\}$, which is at most two-dimensional. We
  conclude that there is an orthonormal basis of $\mathcal{H}\oplus
  \mathcal{H}$ consisting of eigenvectors of $\Gamma$. 
  
  If $\psi = (\phi_1,
  \phi_2)$ is an eigenvector of $\Gamma$ with eigenvalue $\lambda$,
  then using the property that
  $$  
  R^\dagger\Gamma R = 1- \Gamma \quad \text{with} \quad R =\left(
    \begin{matrix}
      0 & J\\
      J & 0
    \end{matrix}\right)$$
    and $J$ denoting complex conjugation, 
 we find that
  $ (\overline{\phi_2}, \overline{\phi_1})$ is an
  eigenvector of $\Gamma$ with eigenvalue $1-\lambda$. Thus we can
  find a unitary transformation $U$ on $\mathcal{H}\oplus\mathcal{H}$ of the form \eqref{defu} that diagonalizes $\Gamma$, and an orthonormal basis of $\mathcal{H}$  such
  that
  \begin{align*}
    U^\dagger\Gamma U (\varphi_i, 0) &= \lambda_i (\varphi_i, 0) \\
    U^\dagger\Gamma U (0, \varphi_i) &= (1-\lambda_i) (0, \varphi_i).
  \end{align*}
  By exchanging $\lambda_i$ and $1-\lambda_i$, if necessary, we can assume that $\sum_i \lambda_i <\infty$ since $\Tr \Gamma (1-\Gamma) = 2\sum_i \lambda_i (1-\lambda_i) < \infty$. 
  
In order to show that   $U$ corresponds to a Bogoliubov transformation $W$ on
  $\mathcal{F}_\mathcal{H}$ we need to show that the Shale--Stinespring criterion is satisfied, i.e., that the offdiagonal entry of $U$ is Hilbert--Schmidt. The finiteness of $\sum_i \lambda_i$ can be expressed as 
  \begin{equation}\label{amounts}
  \Tr  \left[ v^\dagger \gamma v + v^\dagger \alpha w + w^\dagger \alpha^\dagger v + w^\dagger( 1-\bar \gamma) w \right] < \infty \,,
  \end{equation}
  where we use the notations \eqref{defgamma}    and  \eqref{defu} for the entries of $\Gamma$ and $U$, respectively. The fact that $\gamma$ is trace-class implies that $\alpha$ is Hilbert--Schmidt, since $\alpha \alpha^\dagger \leq \gamma(1-\gamma)$ due to $\Gamma^2\leq \Gamma$. Hence \eqref{amounts} amounts to saying that 
  $$\Tr \left( w^\dagger + v^\dagger \alpha\right) \left( w + \alpha^\dagger v \right) < \infty\,,$$
  i.e.,   $w + \alpha^\dagger v$ is Hilbert--Schmidt, which implies the desired result.
  
 We conclude that $U^\dagger\Gamma U$ is  the one-particle density
  matrix of a quasi-free state $W^\dagger\rho W$, with $W$ the Bogoliubov transformation corresponding to $U$. It is then a simple computation to check that 
  $W^\dagger \rho W$ is indeed of the form \eqref{eq:rho_Q}. In fact, with 
  $$ a_i = a(\varphi_i), \quad a^\dagger_i = a^\dagger(\varphi_i) $$ 
we have that 
 $$ \langle a^\dagger_n a_m \rangle_{ W^\dagger \rho W} = \delta_{m,n} \lambda_n ,$$
 and the off-diagonal terms $\langle a^\dagger_n a^\dagger_m\rangle_{ W^\dagger \rho W} $ vanish, since 
$U^\dagger\Gamma U$ is a diagonal matrix. The state $W^\dagger \rho W$ thus corresponds to the vacuum for all the modes with $\lambda_i = 0$. For any other mode, we have 
\begin{align}
 \tr_{{\mathcal{F}_i}} ( e^{ q_i a^\dagger_i a_i} ) & = 1 + e^{q_i} \\
 \tr_{{\mathcal{F}_i}} (a^\dagger_i a_i e^{ q_i a^\dagger_i a_i} ) &= e^{q_i},
\end{align}
where the trace is understood only over the (two-dimensional) Fock space $\F_i$ corresponding to the mode $\varphi_i$. This leads to
$$
\lambda_i = \langle a^\dagger_i a_i \rangle_{ W^\dagger \rho W}  = \frac { \tr_{{\mathcal{F}_i}} (a^\dagger_i a_i e^{ q_i a^\dagger_i a_i} ) } { \tr_{{\mathcal{F}_i}} ( e^{ q_i a^\dagger_i a_i} )  }  = \frac {e^{q_i}}{1 + e^{q_i}} \,.
$$
Note that 
$$
\Tr P e^{\sum_i q_i a^\dagger_i a_i} = \prod_i \left( 1 + e^{q_i}\right) < \infty
$$
since $\sum_i \lambda_i < \infty$ and hence also $\sum_i e^{q_i} < \infty$. 

We have thus shown that both sides of \eqref{eq:rho_Q} are well-defined and both have the same generalized one-particle density matrix $\Gamma$. Since they are both quasi-free states, we conclude that they are indeed equal (since quasi-free states are uniquely determined by their one-particle density matrix).
\end{proof}

\begin{proof}[Proof of Lemma~\ref{lemma:entropy}]
  By Lemma~\ref{lemma:qf}, we can find a unitary transformation $U$ on
  $\mathcal{H}\oplus\mathcal{H}$ with corresponding Bogoliubov
  transformation $W$ on $\F_{\mathcal{H}}$ 
  such that
  $U^\dagger\Gamma U$ is the one-particle density matrix of
  \begin{equation*}
    W^\dagger\rho W = \frac{1}{\tr_{\mathcal{F}_\mathcal{H}} P e^{\sum_i q_i a^\dagger_i a_i} }
      P e^{\sum_i q_i a^\dagger_i a_i}.
  \end{equation*}
 Since the von-Neumann entropy $S(\rho)$ is invariant under unitary transformations, we have 
  \begin{equation}
    \label{eq:S}
    \begin{split}
        -S(\rho) & = \tr_{\mathcal{F}_\mathcal{H}}\big(\rho
    \ln \rho\big)
    = \tr_{\mathcal{F}_\mathcal{H}}\big(W^\dagger\rho W
    \ln(W^\dagger\rho W)\big)
    \\ 
      &= \frac{1}{\tr_{\mathcal{F}_\mathcal{H}} P e^{\sum_i q_i a^\dagger_i a_i} }
      \tr_{\mathcal{F}_\mathcal{H}} P e^{\sum_i q_i a^\dagger_i a_i} \ln(e^{\sum_i q_i a^\dagger_i a_i})  -
      \ln\big(\tr_{\mathcal{F}_\mathcal{H}} P e^{\sum_i q_i a^\dagger_i a_i}\big)\\
      &=
        \sum_n q_n    \frac{1}{\tr_{\mathcal{F}_\mathcal{H}} P  e^{\sum_i q_i a^\dagger_i a_i}}
  \tr_{\mathcal{F}_\mathcal{H}} P  e^{\sum_i a^\dagger_i a_i } a^\dagger_n a_n  -
      \ln\big(\prod_i(1+e^{q_i})\big)\\
      &=     \sum_i q_i \lambda_i - \sum_i \ln(1+ e^{q_i}),
    \end{split}
    \end{equation} 
    where we  used that $ \tr_{\mathcal{F}_\mathcal{H}}\big( W^\dagger \rho W  a^\dagger_i a_i\big) = \lambda_i$. 
    Since $ q_i = \ln \lambda_i - \ln (1-\lambda_i) $ and $1 - \lambda_i = 1/(1+e^{q_i})$, we see further that the above expression
    equals
  \begin{equation}
    -S(\rho)  =    \sum_{i}\lambda_i \ln \lambda_i     +\sum_{i} (1-\lambda_i)\ln(1-\lambda_i) \,,
  \end{equation}
  which is obviously 
  equal to 
  $$ \tr_{\mathcal{H}\oplus\mathcal{H}} \big(U^\dagger \Gamma U \ln(U^\dagger \Gamma U)\big) = -\tr_{\mathcal{H}\oplus\mathcal{H}} \Gamma \ln \Gamma\,. $$
  This completes the proof.
\end{proof}

\section*{Acknowledgments} 
This review originates from lecture notes put together for lectures on the subject given by C.H. at the summer school \lq\lq Current Topics in Mathematical Physics\rq\rq\ at the CIRM in Marseille in Sept. 2013, as well as by R.S. at the workshop  \lq\lq Spectral Theory and Dynamics of Quantum 
Systems\rq\rq\ at  Blaubeuren in  Feb. 2014, and during a summer course at McGill University in July 2014. 
We are grateful to Gerhard Br\"aunlich and Andi Deuchert for many useful remarks and suggestions on various versions of this manuscript.

\end{document}